\documentclass{LMCS}

\def\dOi{9(3:28)2013}
\lmcsheading%
{\dOi}
{1--52}
{}
{}
{Dec.~\phantom08, 2010}
{Sep.~27, 2013}
{}

\usepackage{ifthen}
\usepackage{hyperref}
\newboolean{full}
\setboolean{full}{true}

\newboolean{appendix}

\usepackage{style,enumerate,rotating}

\newcommand{\set}[1]{\{ #1 \}}
\newcommand{\pair}[1]{\langle #1 \rangle}
\newcommand{\pow}{\Pow}
\newcommand{\extK}{M}

\begin{document}
%
%

%
%
\title[Abstract GSOS Rules and a Modular Treatment of Recursive Definitions]{Abstract GSOS Rules and a Modular Treatment of Recursive Definitions}
\author[S.~Milius]{Stefan Milius\rsuper a}
\address{{\lsuper a}Lehrstuhl f\"ur Theoretische Informatik, Friedrich-Alexander
  Universit\"at Erlangen-N\"urnberg, Germany}
\email{mail@stefan-milius.eu}

\author[L.~S.~Moss]{Lawrence S.~Moss\rsuper b}
\address{{\lsuper b}Department of Mathematics, Indiana University, Bloomington,
  IN, USA}
\email{lsm@cs.indiana.edu}
\thanks{{\lsuper b}This work was partially supported by a grant from the Simons Foundation (\#245591 to Lawrence Moss).}

\author[D.~Schwencke]{Daniel Schwencke\rsuper c}
\address{{\lsuper c}Institute of Transportation Systems, German Aerospace Center (DLR), Braunschweig, Germany}
\email{daniel.schwencke@dlr.de}


\ACMCCS{[{\bf Theory of computiation}]: Semantics and
  reasoning---Program semantics---Categorical semantics; Semantics and
  reasoning---Program reasoning---Program specifications}

\begin{abstract}
  Terminal coalgebras for a functor serve as semantic domains for
  state-based systems of various types. For example, behaviors of CCS
  processes, streams, infinite trees, formal languages and
  non-well-founded sets form terminal coalgebras. We present a uniform
  account of the semantics of recursive definitions in terminal
  coalgebras by combining two ideas: (1) \emph{abstract GSOS rules}
  $\ell$ specify additional algebraic operations on a terminal
  coalgebra; (2) terminal coalgebras are also initial \emph{completely
    iterative algebras} (cias).  We also show that an abstract GSOS
  rule leads to new extended cia structures on the terminal
  coalgebra. Then we formalize recursive function definitions
  involving given operations specified by $\ell$ as recursive program
  schemes for $\ell$, and we prove that unique solutions exist in the
  extended cias.  From our results it follows that the solutions of
  recursive (function) definitions in terminal coalgebras may be used
  in subsequent recursive definitions which still have unique
  solutions.  We call this principle \emph{modularity}.  We illustrate
  our results by the five concrete terminal coalgebras mentioned
  above, e.\,g., a finite stream circuit defines a unique stream
  function. %
\end{abstract}

\maketitle

%
%

\section{Introduction}

Recursive definitions are a useful tool to specify infinite system
behavior. For example, Milner~\cite{milner} proved that in his
calculus CCS, one may specify a process uniquely by the equation
\begin{equation}
  \label{eq:CCS}
  P = a.(P|c) + b.0 \,\text{.}
\end{equation}
In other words, the process $P$ is the unique solution of the recursive 
equation $x = a.(x|c) + b.0$. 
Another example is the shuffle product on streams of
real numbers  defined uniquely by a \emph{behavioral differential equation}
\cite{rutten_stream}:
\begin{equation} \label{eq:shuffle_product}
(\sigma \otimes \tau)_0 = \sigma_0\cdot\tau_0 \qquad\qquad
  (\sigma \otimes \tau)' = \sigma \otimes \tau' + \sigma'\otimes \tau
  \,\text{.}
\end{equation}
Here the real number $\sigma_0$ is the head of $\sigma$, and $\sigma'$ is 
the tail; the operation $+$ is the componentwise addition of infinite streams.
Besides these and other examples in the theory of computation,  we
shall mention below recursive specifications which are also important in
other realms; we shall consider non-well-founded sets~\cite{aczel,bm}, a
framework originating as a semantic basis for circular definitions. 
Operations on non-well-founded sets can be  specified  uniquely by
recursive function definitions. For example, we prove that the equation
\begin{equation} \label{eq:function_non-well-founded}
x \| y = \{\, x \| y' \mid y' \in y\,\} \cup \{\,x' \| y \mid x' \in
x\,\} \cup \{\,x' \| y' \mid x' \in x, y' \in y\,\}
\end{equation}
has a unique solution viz.~a binary operation $\|$ on the class of all
non-well-founded sets, which is reminiscent of parallel composition in
process calculi. It is the aim of this paper to develop abstract tools
and results that explain why there exist unique solutions to all the
aforementioned equations.

\paragraph{\bf Terminal coalgebras.}
The key observation is that process behaviors, streams and
non-well-founded sets constitute \emph{terminal coalgebras}
$(C, c: C\to HC)$ for certain endofunctors $H$ on appropriate categories. The functor $H$
describes the type of behavior of a class of state-based systems---the
$H$-coalgebras---and the terminal coalgebra serves as the semantic
domain for the behavior of states of systems of type $H$. 
All of the theory and examples in this paper pertain to terminal coalgebras.

\paragraph{\bf Abstract GSOS rules.}
Let us return to the case of CCS operations to situate the work in a
historical context.  
In the case of CCS one obtains algebraic operations on processes
using structural operational semantics (sos)~\cite{afv01}. Operations
are specified by operational rules such as
\[
\frac{x \stackrel{a}{\to} x' \qquad y \stackrel{a}{\to} y'}{x | y
  \stackrel{a}{\to} x' | y'}
\]
exhibiting the interplay between the operation and the behavior
type. Syntactic restrictions on the rule format then ensure nice
algebraic properties of the specified operations, e.\,g., GSOS
rules~\cite{bim} ensure that bisimilarity is a congruence. Turi and
Plotkin gave in their seminal paper~\cite{tp} a categorical
formulation of GSOS rules, and they show how an abstract GSOS rule
gives rise to a \emph{distributive law} of a monad over a comonad,
where the monad describes the signature of the desired algebraic
operations and the comonad arises from the ``behavior'' functor
$H$. Later Lenisa et al.~\cite{lpw04} proved that abstract GSOS rules
correspond precisely to distributive laws of a free monad $M$ over the
cofree copointed functor on the behavior functor $H$.  In this paper
we shall not need the original formulation of GSOS rules, and so the
reader who is unfamiliar with this formulation will not be at a loss. What will be
important are the categorical generalizations which we call abstract
GSOS rules, and we begin with that topic in the next section.  The
theme of this paper is that abstract GSOS rules, either in a form
close to the original one or in a more general one that we introduce,
account for many interesting recursive definitions on terminal
coalgebras in a uniform way.  These include all of the examples we
mentioned above, and many more.

\paragraph{\bf Modularity.}
An important methodological goal in this paper is that our results
should be \emph{modular}.  What we mean is that we want results which,
given a class of algebras for endofunctors, enable us to expand the
algebraic structure of an algebra from that class by a recursively
defined operation, and stay inside the class.  Thus the results will
iterate.


We mentioned above that our results all had to do with 
terminal coalgebras, and so the reader might wonder
how this idea of modularity could possibly apply. The answer is that a classical result due to 
Lambek~\cite{lambek} implies that the structure $c: C \to HC$ of a terminal
coalgebra is an isomorphism. The corresponding inverse $c^{-1}: HC \to
C$ turns $C$ into an $H$-algebra. Moreover, we are typically interested to start
with an algebraic structure on $C$ that is already extended by additional operations.
For example, in the case of the shuffle product $\sigma\otimes\tau$
of streams, the definition of $\otimes$ uses the stream addition
operation $+$ on the right-hand side. So the definition is made on
the algebra
$$
(C, \quad ( \sigma_0, \sigma') \mapsto \sigma,   \quad \sigma,\tau \mapsto \sigma + \tau)
$$
consisting of the set of streams of reals, its structure $c^{-1}$ as an $H$-algebra,
and its structure as an algebra for $X\times X$ given by componentwise addition.
Returning to modularity, we aim to isolate an appropriateness
condition on algebra structures expanding the inverse $c^{-1}$ of the
terminal coalgebra structure with the property that 
given an appropriate structure and a definition in a certain format,
the definition specifies a new operation in a unique way,
and if we add this operation to the given algebra structure, the
resulting algebra structure on $C$ is again appropriate. This is what
we will call \emph{modularity}, and precise formulations of
appropriateness may be found in the Summaries~\ref{sum:1} and~\ref{sum:2}. 

%

\paragraph{\bf Completely iterative algebras.}
The desired class of algebras mentioned in the previous paragraph is
formed by completely iterative algebras (cias). 
Complete iterativity means that recursive equations involving algebraic
operations corresponding to the $H$-algebra structure $c^{-1}$ can be
solved uniquely (see Definition~\ref{def:cia}). For example, let 
$H_\Sigma$ be a polynomial endofunctor on $\Set$ arising from a signature
of operation symbols with prescribed arity. In this case a cia is   a
$\Sigma$-algebra $A$ in which  every system of recursive equations
\begin{equation}\label{eq:eq}
  x_i = t_i\qquad (i \in I),
\end{equation}
where $t_i$ is either a single operation symbol applied to recursion
variables (i.\,e., $t_i = \sigma(x_{i_1}, \ldots, x_{i_n})$ for an
$n$-ary $\sigma$ from $\Sigma$) or $t_i = a \in A$, has a unique
solution. It can then be proved that more general equation systems
involving (almost) arbitrary $\Sigma$-terms on the right-hand side and
even recursive function definitions have unique solutions in a
cia~\cite{m_cia,mm}. 
Continuing, it was shown in~\cite{m_cia} that the inverse $c^{-1}: HC \to
C$ of the structure of the terminal coalgebra turns $C$ into the
initial cia for $H$. 
However, as we mentioned previously, cia structures for the
``behavior'' functor $H$ are not sufficient to yield the 
existence and uniqueness of solutions in our motivating examples;
these involve additional algebraic operations not captured by
$H$. These operations are $|$ and $+$ in the  CCS process definition above, the stream addition $+$ in \refeq{eq:shuffle_product},
and union $\cup$ in the example~\refeq{eq:function_non-well-founded} from non-well-founded set
theory.

\paragraph{\bf Extended cia structures}
We will show how the abstract GSOS rules of Turi and Plotkin extend the
structure of an initial cia for $H$ (alias terminal
$H$-coalgebra). This then allows us to equip the initial cia with the
desired additional operations such that circular definitions of
elements of the carrier admit a unique solution, e.\,g., the 
equation $\sigma = 1. (\sigma + \sigma)$ (for streams) defining the
stream of powers of $2$ or our first example~\refeq{eq:CCS} for CCS
processes above. 

The first steps in this direction were taken
in Bartels' thesis~\cite{bartels_thesis} (see also~\cite{bartels}); he systematically studies
definition formats giving rise to distributive
laws and shows how to solve parameter-free first order recursive equations involving
operations presented by a distributive law; Uustalu et al.~\cite{uvp} present
the dual of this result. 

We review some basic material on abstract GSOS rules and the
solution theorem of Bartels in Section~\ref{sec:bialg}. 
In Section~\ref{sec:cias} we extend these solution theorems to
equations with parameters, thereby combining them with our previous work
on cias in~\cite{aamv,m_cia}. More concretely, given an abstract GSOS rule
(equivalently, a distributive law of the free monad $M$ over the
cofree copointed functor $H \times \Id$) we prove in Theorems~\ref{thm:distcia}
and~\ref{thm:sandwich} that the terminal $H$-coalgebra carries the structure of
a cia for $HM$ and for $MHM$, respectively. These results show how to construct new
structures of cias on $C$ using an abstract GSOS rule.
This improves Bartels' result in the sense that recursive equations
may employ constant parameters in the terminal coalgebra. 

In Section~\ref{sec:solthms} we obtain new ways to provide the
semantics of recursive definitions by applying the existing solution
theorems from~\cite{aamv,m_cia,mm} to the new cia structures. For
example, we consider solutions of \emph{recursive program schemes}. Classical
recursive program schemes~\cite{guessarian} are function definitions
such as
\begin{equation}
  \label{eq:rps}
  f(x) = F(x, G(f(x)))
\end{equation}
defining a new function $f$ recursively in terms of given ones $F,
G$. In~\cite{mm} is was shown how recursive program schemes can be
formalized categorically. It was proved that every
\emph{guarded}\footnote{Guardedness is a mild syntactic restriction
  stating that terms on right-hand sides of equations do not have a
  newly defined function at their head.} recursive program scheme has
a unique (interpreted) solution in every cia $HA \to A$, where the
functor $H$ captures the signature of given operations. This solution
is an algebra structure $VA \to A$, where $V$ captures the signature
of recursively defined operations. As a new result we now prove in 
Theorem~\ref{thm:extcia} that the unique solution extends the structure of
the given cia for $H$ to a cia for the sum $H + V$; this yields
modularity of unique solutions of recursive program schemes. 

\medskip\noindent{\bf Extended formats of function definition.}  The
recursive program schemes from~\cite{mm} do not capture recursive
function definitions like the one of the shuffle product
in~\refeq{eq:shuffle_product} above or the one for the operation
$\|$ 
from~\refeq{eq:function_non-well-founded} above on non-well-founded
sets.  So in Section~\ref{sec:lambdarps} we provide results that do
capture those examples. We introduce for an abstract GSOS rule $\ell$
two formats of recursive program schemes w.\,r.\,t.~$\ell$: $\ell$-rps
and a variant called sandwiched $\ell$-rps ($\ell$-srps, for
short). They give a categorical formulation of recursive function
definitions such as~\refeq{eq:shuffle_product}
and~\refeq{eq:function_non-well-founded} that refer not only to the
operations provided by the behavior functor $H$ but also to additional
given operations specified by the abstract GSOS rule $\ell$. We will also see
that $\ell$-srps's allow specifications which go beyond the format of
abstract GSOS rules. In Theorems~\ref{thm:ellrps}
and~\ref{thm:sandwich_rps} we prove that every $\ell$-(s)rps has a
unique solution in the terminal coalgebra $C$. Moreover, we show that
this solution extends the cia structure on $C$. This again gives rise
to a modularity principle: operations defined uniquely by
$\ell$-(s)rps's can be used as givens in subsequent (sandwiched) rps's
(we make this precise in the Summary~\ref{sum:2}). 
This modularity of taking solutions of recursive specifications of
operations does not appear in any previous work in this generality.

\medskip\noindent{\bf A set of examples.}
Finally, in Section~\ref{sec:app} we demonstrate the value of our results by
instantiating them in five different concrete applications: (1)
CCS-processes---we explain how Milner's solution theorem from~\cite{milner}
arises as a special case of Theorem~\ref{thm:sandwich}, and we also show how
to define new process combinators recursively from given ones; (2) streams of
real numbers---here we prove that every finite stream circuit defines a
unique stream function, we obtain the result from~\cite{rutten_stream} that
behavioral differential equations  specify operations on streams in a unique way
as a
special instance of our Theorem~\ref{thm:ellrps}, and we show how to solve 
recursive equations uniquely that cannot be captured by behavioral differential
equations by applying Theorem~\ref{thm:sandwich_rps}; (3) infinite trees---we
obtain the result from \cite{rs10} that behavioral differential
equations have unique solutions as a special case of Theorem~\ref{thm:ellrps};
(4) formal languages---here we show how
operations on formal languages like union, concatenation, complement,
etc.~arise step-by-step using the modularity of unique solutions
of $\ell$-rps's, and how languages generated by grammars arise as the
unique solutions of flat equation morphisms in cias; (5) non-well-founded sets---we prove
that operations on non-well-founded sets are uniquely determined as
solutions of $\ell$-(s)rps's.

\smallskip
\noindent
{\bf Related Work.} As already mentioned, Turi's and Plotkin's work~\cite{tp} was taken further by Lenisa,
Power and Watanabe in~\cite{lpw00,lpw04}. Fiore and Turi~\cite{ft01} 
applied the mathematical operational semantics to provide semantics of message-passing process calculi such as
Milner's $\pi$-calculus. But these papers do not consider the
semantics of recursive definitions. 

Turi~\cite{turi97} gives a treatment of guarded recursion. He does not
isolate the notion of a solution of a recursive specification and
whence does not prove that a solution exists uniquely. In addition
that paper does not deal with recursive function definitions as we do
here. 

Jacobs~\cite{jacobs} shows
how to apply Bartels' result to obtain the (first order) solution
theorems from~\cite{aamv,m_cia}. Capretta et al.~\cite{cuv} work in a
dual setting and generalize the results of~\cite{bartels} beyond
terminal coalgebras and they also obtain the (dual of) the solution
theorem from~\cite{aamv,m_cia} by an application of their general
results. Our Theorems~\ref{thm:distcia} and~\ref{thm:sandwich} are similar to
results in~\cite{cuv}, but in the later sections we extend the work in~\cite{bartels} in a different
direction by considering parameters in recursive definitions. So our
results in the present paper go beyond what can be accomplished with
previous work. %

Modularity in mathematical operational semantics has been studied
before in~\cite{pow03,lpw04}. These papers show how to combine two
different specifications of operations over the same behavior by
performing constructions at the level of distributive laws. This gives an abstract
explanation of adding operations to a process calculus. At the heart
of our proof of Theorem~\ref{thm:ellrps} lies a construction similar
to the combination of two distributive laws that arises by taking the coproduct of
the corresponding monads. However, while in the coproduct construction of~\cite{pow03,lpw04} the two 
distributive laws are independent of each other, our case is different
because the operations specified by the second distributive law interact with
the operations specified by the first one. 

Much less related to our work is the work of Kick and
Power~\cite{kp04} who show how to combine distributive laws of one
monad over two sorts of behavior possibly interacting with one another.

To the best of our knowledge the results on modularity of solutions of
recursive specifications we present are new. 

The present paper is a completely revised version containing full proofs of
the conference paper~\cite{mms}.

\section{Abstract GSOS Rules and Distributive Laws}
\label{sec:bialg}

We shall assume some familiarity with basic notions from category
theory such as functors, (initial) algebras and (terminal) coalgebras,
and monads, see e.\,g.~\cite{maclane,rutten,adamek_survey}.

Suppose we are given an endofunctor $H$ on some category $\A$
describing the behavior type of a class of systems. In our work we
shall be interested in additional algebraic operations on the
terminal coalgebra $C$ for $H$. The type of these algebraic operations is
given by an endofunctor $K$ on $\A$, and the algebraic operations are given
by an \emph{abstract GSOS rule} (cf.~\cite{tp,lpw00,lpw04}). Our goal
is to provide a setting in which recursive equations with operations specified by abstract
GSOS rules have unique solutions. We now review the necessary preliminaries.

\begin{ass}
  \label{ass:1}
  Throughout the rest of this paper we assume that $\A$ is a category with
  binary products and coproducts, $H: \A \to \A$ is a functor, and that $c: C \to HC$ is
  its terminal coalgebra. We also assume that $K:\A\to\A$ is a functor
  such that for every object $X$ of $\A$ there exists a free
  $K$-algebra $\extK  X$ on $X$. 
\end{ass}

\begin{rem} \label{rem:free_algebras_monad}
  (1)~We denote by $\varphi_X:K\extK X\to \extK X$ and $\eta_X:X\to \extK X$ the structure and universal morphism of the free $K$-algebra
   $\extK X$. Recall that the corresponding universal property states that for every $K$-algebra $a:KA\to A$ and every morphism $f:X\to A$ there exists a unique $K$-algebra homomorphism $h:(\extK X,\phi_X)\to (A,a)$ such that $h \cdot \eta_X = f$.

  \medskip\noindent
  (2)~Free algebras for the functor $K$ exist under mild assumptions on $K$.
  For example, whenever $K$ is an accessible endofunctor on $\Set$ it has all
  free algebras $\extK X$ (see e.\,g.~\cite{at}).

  \medskip\noindent
  (3)~As proved by Barr~\cite{barr} (see also Kelly~\cite{kelly80}), the existence of free $K$-algebras as
  stated in Assumption~\ref{ass:1} implies that there is a free monad on $K$.
  Indeed, $\extK $ is the object assignment of a
  monad with the unit given by $\eta_X$ from item~(1) and the multiplication
  $\mu_X:\extK \extK X\to\extK X$ given as the unique homomorphic extension of
  $\id_{\extK X}$. Also $\varphi:K\extK \to\extK $ is a natural transformation
  and
  $$
  \kappa = (\xymatrix@1{
      K
      \ar[r]^-{K\eta}
      &
      K\extK 
      \ar[r]^-{\varphi}
      &
      \extK 
      })
  $$
  is the universal natural transformation of the free monad $\extK$. 
  Notice that for complete categories all free monads arise
  from free algebras in this way. 
\end{rem}

\begin{notation} \label{not:extalg}
  For any $K$-algebra $a: KA \to A$, let $\extalg a: \ext{K} A \to A$ be the unique
  $K$-algebra homomorphism with $\extalg a \cdot \eta_A = \id_A$. Then $\extalg a$ is the structure of an Eilenberg-Moore algebra for $\extK$. Notice also that 
  \begin{equation} \label{diag:kappa}
    a = (
    \xymatrix@1{
    KA
    \ar[r]^{\kappa_A}
    &
    \extK A
    \ar[r]^{\extalg a}
    &
    A
    }
    )\,\text{.}
  \end{equation}
\end{notation}

\begin{exa} \label{ex:sigma_algebras}
  In our applications the functor $K$ will be a polynomial functor on $\Set$ most of the time. 
  In more detail, let $\A=\Set$ and let $\Sigma=(\Sigma_n)_{n\in\Nat}$ be a signature of operation symbols with prescribed arity. The polynomial functor $K_{\Sigma}$ associated with $\Sigma$ is given by
$$K_{\Sigma}X = \coprod_{n\in\Nat}\Sigma_n\times X^n\,\text{.}$$
Algebras for the functor $K_{\Sigma}$ are precisely the usual $\Sigma$-algebras and the free monad $\extK _{\Sigma}$ assigns to a set $X$ the $\Sigma$-algebra $\extK _{\Sigma}X$ of $\Sigma$-terms (or finite $\Sigma$-trees) on $X$, where a $\Sigma$-tree on $X$ is a rooted and ordered tree with leaves labeled by constant symbols from $\Sigma$ or elements of $X$ and with inner nodes with $n$ children labeled in $\Sigma_n$.
\end{exa}

\begin{defi} \cite{tp} \label{def:gsos_rule}
  An \emph{abstract GSOS rule} is a natural transformation
  $$\ell:K(H\times\Id)\to H\extK \,\text{.}$$
\end{defi}

\begin{rem} \label{rem:ell_lambda}
  (1)~The name ``abstract GSOS'' is motivated by the fact that for $H$ the
  set functor $\Pfin(A\times -)$ (whose coalgebras are labeled
  transition systems) and $K$ a polynomial set functor, a natural
  transformation $\ell$ as in Definition \ref{def:gsos_rule}
  corresponds precisely to a transition system specification with
  operational rules in the GSOS format of Bloom, Istrail and
  Meyer~\cite{bim}; this was proved by Turi and Plotkin~\cite{tp}.
  
  \medskip\noindent
  (2)~As proved by Lenisa et al.~\cite{lpw00,lpw04}, abstract GSOS
  rules are in one-to-one correspondence with distributive laws of the
  free monad $\extK $ over the copointed functor $H \times \Id$.
More precisely, recall from loc.\,cit.~that a \emph{copointed functor} $(D,\varepsilon)$ is a functor $D:\A\to\A$ equipped with a natural transformation $\varepsilon:D\to\Id$.
Also, notice that $H\times\Id$ together with the projection $\pi_1:H\times\Id\to\Id$ is the cofree copointed functor on $H$. Further recall that a distributive law of a monad $(M,\eta,\mu)$ over a copointed functor $(D,\varepsilon)$ is a natural transformation $\lambda:MD\to DM$ such that the following diagrams commute:
  \begin{equation} \label{diags:distributive_law}
    \vcenter{
    \xymatrix@!C=0.6cm{
      	& D \ar[dl]_{\eta D} \ar[dr]^{D\eta}	&	\\
      MD \ar[rr]^{\lambda}	&	& DM
    }
    }
    \vcenter{
    \xymatrix@!C=0.6cm{
      MD \ar[rr]^{\lambda} \ar[dr]_{M\varepsilon}	&	& DM \ar[dl]^{\varepsilon M}	\\
      	& M	&
    }
    }
    \quad
    \vcenter{
    \xymatrix{
      MMD \ar[r]^{M\lambda} \ar[d]_{\mu D}	& MDM \ar[r]^{\lambda M}	& DMM \ar[d]^{D\mu}	\\
      MD \ar[rr]^{\lambda}	&	& DM
    }
    }
  \end{equation}
  Then to give an abstract GSOS rule $\ell:K(H\times\Id)\to H\extK $ is equivalent to giving a distributive law $\lambda$ of the monad $\extK $ over the copointed functor $H\times\Id$.
\takeout{ 
  \medskip\noindent
  (3)~For some of our results we shall need to assume a distributive law $\lambda$ of $\extK $ over the functor $H$, i.\,e., $\lambda:\extK H\to H\extK $ is a natural transformation such that the first and last diagram in \refeq{diags:distributive_law} commute for $D$ replaced by $H$; in symbols: $\lambda\cdot \eta H = H\eta$ and $\lambda\cdot\mu H = H\mu\cdot \lambda M\cdot M\lambda$. \\
  To give such a distributive law is equivalent to giving a natural transformation $\ell:KH\to H\extK $, see Bartels~\cite{bartels_thesis}, Lemma 3.4.24. We call such a natural transformation $\ell$ a \emph{restricted abstract GSOS rule}. Every restricted abstract GSOS rule yields the abstract GSOS rule
  \begin{equation} \label{eq:restricted_rule}
    \xymatrix@1{
      K(H\times \Id)
      \ar[r]^-{K\pi_0}
      &
      KH
      \ar[r]^-{\ell}
      &
      H\extK 
      }\,\text{.}
  \end{equation}
  Conversely, every abstract GSOS rule $\ell:K(H\times\Id)\to H\extK $ gives, of course, a restricted one for the functor $H\times\Id$. However, there is no obvious way to produce a natural transformation $KH\to H\extK $ from $\ell$.
}
\end{rem}

\begin{thm} {\rm\cite{bartels_thesis,tp}} \label{thm:interp_gsos_rule}
  Let $\ell$ be an abstract GSOS rule. There is a unique structure $b:KC\to C$ of a $K$-algebra on the terminal $H$-coalgebra $C$ such that the square below commutes:
  \begin{equation} \label{diag:interp_gsos_rule}
  \vcenter{
  \xymatrix@C+1pc{
    KC \ar[r]^-{K\langle c,\id_C\rangle} \ar[d]_{b}	& K(HC\times C) \ar[r]^-{\ell_C}	& H\extK C \ar[d]^{H\extalg b}	\\
    C \ar[rr]^{c}	&	& HC
  }
  }
  \end{equation}
\end{thm}

\begin{rem} \label{rem:interp_gsos_lambda}
  (1)~In the terminology of \cite{bartels_thesis,tp} the triple $(C,b,c)$ is a \emph{model} of the abstract GSOS rule $\ell$; in fact, it is the terminal one.

  \medskip\noindent
  (2)~For the distributive law $\lambda:\extK (H\times\Id)\to(H\times\Id)\extK $ from Remark~\ref{rem:ell_lambda}(2) corresponding to the abstract GSOS rule $\ell$ we have the following commutative diagram (see e.\,g.~\cite{bartels_thesis}, Lemma 3.5.2, where this is formulated for $\A = \Set$):
  \begin{equation} \label{diag:interp_gsos_lambda}
  \vcenter{
  \xymatrix@C+1pc{
    \extK C \ar[r]^-{\extK \langle c,\id_C\rangle} \ar[d]_{\extalg b}	& \extK (HC\times C) \ar[r]^{\lambda_C}	& (H\times\Id)\extK C \ar[d]^{(H\times \Id)\extalg b}	\\
    C \ar[rr]^{\langle c,\id_C\rangle}	&	& HC\times C
  }
  }
  \end{equation}

  \medskip\noindent
  (3)~In some of the examples below an abstract GSOS rule $\ell$ will arise from a natural transformation $\ell':K(H\times \Id)\to HK$ as
  $$
  \ell = (
  \xymatrix{
    K(H\times\Id)
    \ar[r]^-{\ell'}
    &
    HK
    \ar[r]^{H\kappa}
    &
    H\extK 
  }
  )\,\text{.}
  $$
  In this case Diagram~\refeq{diag:interp_gsos_rule} can be simplified; indeed, $b:KC\to C$ is then the unique morphism such that the diagram below commutes:
  $$
  \xymatrix@C+1pc{
    KC \ar[r]^-{K\langle c,\id_C\rangle} \ar[d]_{b}	& K(HC\times C) \ar[r]^-{\ell_C'}	& HKC \ar[d]^{Hb}	\\
    C \ar[rr]^{c}	&	& HC
  }
  $$
\end{rem}

\takeout{ 
\begin{rem} \label{rem:liftdist}
  (1)~Let $K$ and $H$ be endofunctors on $\A$, and let $\ell: KH \to HK$ be a
  distributive law of $K$ over $H$---this is just a natural transformation. 
  The free pointed endofunctor on $K$ is $M = K + \Id$. Then
  $\ell$ gives rise to a distributive law of the pointed endofunctor $M$ over
  $H$ as follows:
  $$
  \lambda =
  (\xymatrix@1{
    MH = KH + H \ar[r]^-{\ell + H}
    &
    HK + H \ar[r]^-{\can}
    &
    H(K + \Id) = HM
  })\, .
  $$
  
  \medskip\noindent
  (2) As in~(1) above, any natural transformation 
  $\ell: KH \to H(K + \Id)$ yields a distributive law $\lambda$ of the free
  pointed endofunctor $M = K + \Id$ over $H$:
  $$
  \lambda =
  (\xymatrix@1@C+1pc{
    MH = KH + H \ar[r]^-{[\ell, H\inr]} & H(K+\Id) = HM
  })\, .
  $$
  
  \medskip\noindent
  (3)~In most concrete examples, $M$ in
  Definition~\ref{def:dist} is part of a monad $(M, \eta, \mu)$.
  A distributive law of a monad $M$ over an endofunctor $M$ is 
  $\lambda: MH\to HM$ such that (as above) $H\eta = \lambda \cdot \eta H$,
  and also $\lambda \cdot \mu H = H\mu \cdot \lambda M \cdot M \lambda$.

  \medskip\noindent
  (4)~Suppose
  that $M$ is a free monad $\ext{K}$ on an endofunctor $K: \A \to \A$
  given objectwise by free $K$-algebras $\extK  X$.  
  Bartels~\cite{bartels_thesis} (Lemma~3.4.24)  shows  that
  distributive laws $\lambda: \ext{K}H \to H\ext{K}$ (of the monad $\ext{K}$
  over $H$) correspond to  natural transformations $\ell: KH \to
  H\ext{K}$. Indeed, given $\lambda$ we get $\ell = \lambda \cdot \kappa
  H$, where $\kappa: K \to \extK $ is the universal natural transformation of
  the free monad. Conversely, let $\eta$ and $\mu$ be the unit and
  multiplication of the free monad $\extK $. Given
  $\ell$, we see that $H\mu_X \cdot \ell_{\extK  X}: KH\extK  X \to
  H\extK  \extK  X \to H\extK  X$ is an algebra for $K$. Thus, by the
  freeness of the algebra $\extK  HX$, there exists a unique
  $K$-algebra homomorphism $\lambda_X: \extK  H X \to H \extK  X$
  extending $H\eta_X$, i.\,e., such that $\lambda_X \cdot \eta_{HX} =
  H\eta_X$. One readily shows that $\lambda$ is a distributive law of
  $\extK $ over $H$.

  \takeout{
  \medskip\noindent
  (5)~In our examples we will need to combine several distributive laws into one. In
  our applications we are only interested in distributive laws of pointed
  endofunctors over an endofunctor $H$ arising either from natural
  transformations as in~(1) or~(2) above. In these two cases
  the combination of appropriate natural transformations is easy:

  Let $k: KH \to HK$ and $\ell: LH \to HL$ be distributive laws of the
  endofunctors $K$ and $L$, respectively, over $H$. Then we obtain another
  distributive law of $K+L$ over $H$ as follows:
  \begin{equation}\label{eq:combone}
    \xymatrix@1{
      (K+L)H = KH + LH 
      \ar[r]^-{k + \ell}
      &
      HK + HL
      \ar[r]^-{\can}
      &
      H(K+L) \,\text{.}
    }
  \end{equation}

  Similarly, one can combine two natural transformations $k: KH \to
  H(K+\Id)$ and $\ell: LH \to H(L+\Id)$ into one of the form
  $(K+L)H \to H(K+L+\Id)$.
}
\takeout{
  \medskip\noindent
  (6)~Another result of this type concerns two natural transformations
  $\ell_i: K_iH \to H\ext{K_i}$ as in~(4). Notice that the coproduct
  injections $\inl: K_1 \to K_1 + K_2 \leftarrow K_2: \inr$ lift to
  monad morphisms $\ext{\inl}: \ext{K_1} \to \ext{K_1 + K_2}
  \leftarrow \ext{K_2}: \ext{\inr}$. So we can form
  \begin{equation}\label{eq:combtwo}
    \ell = (
    \xymatrix@1{
      (K_1 + K_2)H \ar[r]^-{\ell_1 + \ell_2} &
      H\ext{K_1} + H\ext{K_2} \ar[r]^\can &
      H(\ext{K_1 + K_2})
    })
  \end{equation}
  where $\can = [H\ext{\inl},H\ext{\inr}]$ is the canonical
  morphism. And this then gives a combined distributive law of
  $\ext{K_1 + K_2}$ over $H$.
}
\end{rem}
}

\takeout{ 
\begin{construction} \cite{bartels, bartels_thesis}
  \label{constr:interp}
  Let $\lambda: MH \to HM$ be a distributive law of the pointed
  endofunctor $M$ over $H$. By using coinduction
  (i.\,e., the universal property of the terminal coalgebra) we define an
  $M$-algebra structure on $C$ as the unique coalgebra homomorphism
  $b$ from the coalgebra $\lambda_C \cdot Mc: MC \to HMC$ to the terminal
  coalgebra, i.\,e., $b: MC \to C$ is the unique morphism
  making the following square commutative:
  \begin{equation}
    \vcenter{
      \xymatrix{
        MC 
        \ar[r]^-{Mc}
        \ar[d]_b
        & 
        MHC
        \ar[r]^-{\lambda_C}
        &
        HMC
        \ar[d]^{Hb}
        \\
        C \ar[rr]_-c & & HC
      }}
    \label{eq:interp:M:c}
  \end{equation}
  Bartels~\cite{bartels} showed that $(C,b)$ is indeed an algebra for $M$.
  Moreover, if $M$ is a monad and $\lambda$ a distributive law of $M$
  over $H$, then $(C,b)$ is an Eilenberg-Moore algebra for the monad
  $M$. 
\end{construction}
}

\begin{defi}
  \label{dfn:interp}
  For every abstract GSOS rule $\ell$ we will call the algebra structure $b: KC \to C$ from
  Theorem~\ref{thm:interp_gsos_rule} the \emph{$\ell$-interpretation in $C$}.
\end{defi}

\takeout{
\begin{prop}
  \label{prop:combi}
  Let $\ell_i:K_iH \to H\ext{K_i}$ be two
  natural transformations, and $b_i: \ext{K_i} C \to C$ be the
  corresponding interpretations. Then the interpretation $b: \ext{K_1
    + K_2} C \cdot C$ induced by $\lambda$ in~\refeq{eq:combtwo} 
  extends $b_1$ and $b_2$ in the sense that $b \cdot \ext{\inl} = b_1$
  and $b\cdot \ext{\inr} = b_2$. 
\end{prop}
}

\begin{exas} \label{ex:dist}
  We review a couple of examples of interest in this
  paper where $\A = \Set$. We shall elaborate on these examples in
  Section~\ref{sec:app} and present two more examples. 
  
  \medskip\noindent
    (1)~Processes. We shall be interested in Milner's CCS~\cite{milner}.
  Let $\kappa$ be an infinite cardinal and 
  $\Pow_\kappa$ be the functor assigning to the set $X$ the set of all
  $Y \subseteq X$ with $|Y| < \kappa$. Here we consider the functor $HX = \Pow_\kappa(A
  \times X)$ where $A$ is some fixed alphabet of actions. Following
  Milner~\cite{milner}, we assume that for every $a \in A$ we also have
  a complement $\bar a \in A$ (so that $\bar{\bar a} = a$) and a
  special silent action $\tau \in A$.

  To describe the terminal coalgebra for $\Pow_\kappa(A \times
  X)$ we first recall the description of the terminal coalgebra for the finite power set functor
  $\Pfin$ by Worrell~\cite{worrell}: it is carried by
  the set of all strongly extensional finitely branching trees, where
  an unordered tree $t$ is called \emph{strongly extensional} if two
  subtrees rooted at distinct children of some node of $t$ are never
  bisimilar as trees.  Similarly, the terminal coalgebra for the countable power
  set functor $\Pcount$ is carried by the set of all strongly
  extensional countably branching trees, see
  \cite{schwencke_logacc}. The technique by which this result is
  obtained in loc.\,cit. generalizes to the functor $\Pow_\kappa(A
  \times X)$ from above: its terminal
  coalgebra $C$ turns out to consist of all strongly extensional
  $\kappa$-branching trees with edges labeled in $A$; \emph{strong
  extensionality} has the analogous meaning as above: two
  subtrees rooted at distinct children of some node are never
  bisimilar as trees if both edges to the children carry the same label.  The
  elements of $C$ can be considered as (denotations of) CCS-agents
  modulo strong bisimilarity.

  Notice that the inverse $c^{-1}: \Pow_\kappa(A \times C) \to C$
  assigns to a set $\{(a_i, E_i) \mid i < \kappa\}$ of pairs of actions and agents
  the agent $ \sum_{i < \kappa} a_i.E_i$. The usual process
  combinators ``$a.-$'' (prefixing), ``$|$'' (parallel composition), ``$\sum_{i<\kappa}$'' (summation), ``$-[f]$''
  (relabeling) and ``$-\backslash L$" (restriction) are given by sos rules. 
  Let $E$, $E'$, $F$, $F'$ be agents and $a\in A$ some action, then
  these rules are:
      \[
      \frac{E\stackrel{a}{\to}E'}{E|F\stackrel{a}{\to}E'|F} \qquad\quad
      \frac{F\stackrel{a}{\to}F'}{E|F\stackrel{a}{\to}E|F'} \qquad\quad
      \frac{E\stackrel{a}{\to}E'\quad F\stackrel{\bar
          a}{\to}F'}{E|F\stackrel{\tau}{\to}E'|F'} \thickspace
      {(a\neq\tau)}
      \]
      \[
      \frac{}{a.E \stackrel{a}{\to} E}
      \qquad
      \frac{E_j \stackrel{a}{\to} E_j'}
      {(\sum_{i < \kappa} E_i) \stackrel{a}{\to} E_j'}\thickspace {(j < \kappa)}
      \qquad
      \frac{E \stackrel{a}{\to}E'}{E[f] \stackrel{f(a)}{\to}
        E'[f]}
      \qquad
      \frac{E \stackrel{a}{\to}E'}{E\backslash L
        \stackrel{a}{\to}E'\backslash L}
      \thickspace {(a, \bar a \not\in L)}
      \]
  Now let $K$ be the polynomial functor for the signature given by
  taking these combinators as operation symbols. It easily follows from the work in \cite{bartels_thesis} and \cite{lpw04} that the  rules
  above give
  an abstract GSOS rule $\ell: K(H\times\Id) \to H\extK $, and the
  $\ell$-interpretation $b:KC \to C$ in $C$ provides the desired operations on CCS-agents (modulo strong bisimilarity). Further details are
  presented in Section~6.1.
  
  \medskip\noindent
  (2)~Streams. Streams have been studied in a coalgebraic setting by
  Rutten~\cite{rutten_stream}. Here we take the functor $HX = \Real \times X$
  whose terminal coalgebra $(C,c)$ is carried by the set $\Real^\omega$ of all
  streams over $\Real$ and $c = \<\hd,\tl\>: \Real^\omega \to \Real \times
  \Real^\omega$ is given by the usual head and tail functions on
  streams.

  Operations on streams can be defined by so-called \emph{behavioral differential equations} \cite{rutten_stream}. Here one uses for every stream $\sigma$ the notation $\sigma_0=\hd(\sigma)$ and $\sigma'=\tl(\sigma)$. Then, for example, the function $\zip$ merging two streams is specified by
  $$\zip(x,y)_0 = x_0 \qquad\qquad (\zip(x,y))' = \zip(y,x') \,\text{.}$$
  This gives rise to an abstract GSOS rule as follows. Let $KX = X \times X$ (representing the binary operation $\zip$),
   and let $\ell:K(H\times\Id)\to H\extK $ be
  $$
    \ell = (
    \xymatrix@1{
      K(H\times\Id) \ar[r]^-{\ell'} &
      HK \ar[r]^-{H\kappa} &
      H\extK
    })
  $$
  where $\ell'$ is given by
  $$\ell_X'((r,x',x),(s,y',y)) = (r,(y,x'))\,\text{.}$$
  Notice that in a triple $(r,x',x)\in HX\times X$, $x$ is a variable for a stream with head $r$ and tail referred to by the variable $x'$. It is now straightforward to work out that the $\ell$-interpretation $b:KC\to C$ is the operation $\zip$.

For another example, the componentwise addition of two streams $\sigma$ and $\tau$ is specified by
\begin{equation} \label{eq:stream_addition}
(\sigma + \tau)_0 = \sigma_0 + \tau_0 \qquad\qquad
(\sigma + \tau)' = (\sigma' + \tau') \,\text{.}
\end{equation}
In this case, we let $KX = X\times X$, and 
$$
    \ell = (
    \xymatrix@1{
      K(H\times\Id) \ar[r]^-{\ell'} &
      HK \ar[r]^-{H\kappa} &
      H\extK
    })\,\text{,}
  $$
  where $\ell'$ is given by
  $$\ell_X'((r,x',x),(s,y',y)) = (r+s,(x',y'))\,\text{.}$$
  Again, it is not hard to show that the $\ell$-interpretation $b: C\times C \to C$ is componentwise addition.
  (For related details, see Example~\ref{expl} below.)


  \medskip\noindent
  (3)~Formal languages. Consider the
  endofunctor $HX =X^A \times 2$ on $\Set$, where $2 = \{\,0,1\,\}$. Coalgebras
  for $H$ are precisely the (possibly infinite) deterministic automata
  over the set $A$ (as an alphabet). 
  The terminal coalgebra $c: C \to HC$ consists of all formal languages with
  $c(L) = (\lambda a.L^a, i)$ with $i = 1$ iff the empty word $\eps$ is
  in $L$ and where $L^a = \{\, w \mid aw \in L\,\}$.

  To specify e.\,g.\ the intersection of formal languages by an abstract GSOS rule,
  let $KX = X \times X$ and let $\ell: K(H\times\Id) \to H\extK $ be
  $$
    \ell = (
    \xymatrix@1{
      K(H\times\Id) \ar[r]^-{K\pi_0} &
      KH \ar[r]^-{\ell'} &
      HK \ar[r]^-{H\kappa} &
      H\extK
    })
  $$
  where $\ell':KH\to HK$ is given by $\ell_X'((f,i), (g,j)) = (\langle f,g\rangle, i \wedge j)$
  where $\wedge$ denotes the ``and''-operation on $\{\,0,1\,\}$. 
  Then the $\ell$-interpretation is easily verified to be the intersection of formal languages.
\end{exas}

Our first result (Theorem~\ref{thm:sandwich_ellsolution}) improves a result
from~\cite{bartels,bartels_thesis} that we now recall. For the rest of
this section we assume that an abstract GSOS rule
$\ell:K(H\times\Id)\to H\extK $ is given. Recall that the structure of
the terminal coalgebra $c: C \to HC$ is an isomorphism by the Lambek
Lemma~\cite{lambek}. From now on we will regard $C$ as an $H$-algebra
with the structure $c^{-1}: HC \cdot C$ most of the time.

\begin{defi}
  \label{dfn:elleq}
  An \emph{$\ell$-equation} is an $H\extK $-coalgebra; 
  that is, a morphism of the form 
  \[
  e: X \to H\extK X.
  \]
  A \emph{solution} of $e$ in the terminal coalgebra 
  $C$ is a morphism $e^\dag: X\to C$ such that the diagram below commutes:
  \begin{equation} \label{eq:sol_ell-equation}
  \vcenter{
  \xymatrix@+1pc{
    X \ar[dd]_{e} \ar[r]^{e^\dag}	& C	\\
    	& HC \ar[u]_{c^{-1}}	\\
    H\extK X \ar[r]^{H\extK e^\dag}	& H\extK C \ar[u]_{H\extalg b}
  }
  }
  \end{equation}
\end{defi}

\begin{thm} {\rm\cite{bartels,bartels_thesis}} \label{thm:ellsolution} 
  For every $\ell$-equation there exists a unique solution in $C$.
\end{thm}

This result follows from Corollaries 4.3.6 and 4.3.8 and Lemma 4.3.9 in \cite{bartels_thesis}. The first of these results is the dual of a
result obtained independently and at the same time by Uustalu, Vene
and Pardo (see~\cite{uvp}, Theorem 1). Notice that
both~\cite{bartels_thesis} and~\cite{uvp} work at the level of
generality provided by distributive laws rather than with (the dual
of) abstract GSOS rules. And Capretta, Uustalu and Vene~\cite{cuv}
generalize this further. In their Theorems 19 and 28 they replace the
inverse $c^{-1}$ of the terminal coalgebra by an algebra $a: HA \to
A$ having the property that for every coalgebra $e: X \to HX$ there
exists a unique coalgebra-to-algebra homomorphism $h: X \to A$,
i.\,e., $h$ exists uniquely such that $a \cdot Hh \cdot e = h$.  
Our generalization in this paper goes in a different direction. We
keep the algebra $c^{-1}: HC \to C$ but generalize the format of
equations having a unique solution in this algebra. But before we do
this let us give an example of an $\ell$-equation and its solutions
for streams.

\begin{exas} \label{ex:moredist}
As in our discussion of  streams in Example~\ref{ex:dist}(2),
let $HX = \Real\times  X$ so that that $C$ is the set of streams over $\Real$,
let $KX = X\times X$, 
and 
let $\ell$ be the abstract GSOS rule for zipping streams so that $b: KC\to C$
is the operation of zipping.
As an example of an $\ell$-equation, let $X = \set{x,y}$,
and let $e: X \to \Real \times \extK X$ be given by
\[
e(x) = (0, \zip(y,x)) \qquad\text{and}\qquad e(y) = (1, \zip(x,y)).
\]
Then the streams $t = e^\dag(x)$ and $u = e^\dag(y)$ satisfy:
\begin{equation}\label{TM}
\begin{array}{lcl}
t  & = & 1. \zip(u,t) \\
u & = & 0.\zip(t,u)
\end{array}
\end{equation}
\end{exas}
To continue the discussion from streams just
above, we shall solve equations such as (\ref{TMmore}) below which are
more complicated than (\ref{TM}). To start with, for the proof of
Theorem~\ref{thm:sandwich} we shall need a variant of Theorem
\ref{thm:ellsolution} for equations of the form $e:X \to \extK H\extK
X$:

\begin{defi}
  A \emph{sandwiched $\ell$-equation} is a $\extK H\extK $-coalgebra; 
  that is, a morphism of the form $e: X \to \extK H\extK X$.
  A \emph{solution} of $e$ in the terminal coalgebra 
  $C$ is a morphism $e^\dag: X\to C$ such that the diagram below commutes:
  \begin{equation} \label{eq:sol_ext_ell-equation}
  \vcenter{
  \xymatrix@C+1pc{
    X \ar[ddd]_{e} \ar[r]^{e^\dag}	& C	\\
    	& \extK C \ar[u]_{\extalg b}	\\
    	& \extK HC \ar[u]_{\extK c^{-1}}	\\
    \extK H\extK X \ar[r]^{\ \extK H\extK e^\dag\ }	& \extK H\extK C \ar[u]_{\extK H\extalg b}
  }
  }
  \end{equation}
\end{defi}

\begin{exas} \label{ex:moreTM}
This is a variation on  Example~\ref{ex:moredist}.
Using a sandwiched $\ell$-equation, we 
can solve
\begin{equation}\label{TMmore}
\begin{array}{lcl}
t  & = & \zip(1.u,0.t) \\
u & = & \zip(0.t,1.u) \\
\end{array}
\end{equation}
Note the difference between \refeq{TM} and \refeq{TMmore}.
The key point about a sandwiched system is that the ``guards'' by prefix
operations $r.-$ need not occur at the head of the term on the
right-hand sides. (In this example, they are applied to the variables
inside $\zip$.) 
Incidentally, the solution assigns to $u$ the famous Thue-Morse
stream and to $t$ its dual; the solutions to \refeq{TM} are different.
A more complicated sandwiched system would be 
$$
\begin{array}{lcl}
t  & = &\zip(0.u,\zip(1.u,0.\zip(1.u,0.u))) \\
u & = & \zip(\zip(0.t,1.u),  \zip(0.t,1.u)).
\end{array}
$$
\end{exas}

\begin{thm} \label{thm:sandwich_ellsolution}
  For every sandwiched $\ell$-equation there exists a
  unique solution in $C$.
\end{thm}
\begin{proof}
  Given a sandwiched $\ell$-equation $e:X \to \extK H\extK X$, we form the following (ordinary) $\ell$-equation:
  $$
  \ol e = (
  \xymatrix@C+1pc{
    H\extK X
    \ar[r]^-{H\extK e}
    &
    H\extK \extK H\extK X
    \ar[r]^-{H\mu_{H\extK X}}
    &
    H\extK H\extK X
  }
  )\,\text{.}
  $$
  From Theorem \ref{thm:ellsolution} we know that $\ol e$ has a unique solution $\ol e^\dag:H\extK X\to C$. Thus, we are finished if we can show that solutions of $e$ and $\ol e$ are in one-to-one correspondence.

  Firstly, from the solution $\ol e^\dag$ of $\ol e$ we obtain
  $$
  e^\dag = (
  \xymatrix{
    X
    \ar[r]^-{e}
    &
    \extK H\extK X
    \ar[r]^-{\extK \ol e^\dag}
    &
    \extK C
    \ar[r]^{\extalg b}
    &
    C
  }
  )\,\text{,}
  $$
  and we will now verify that $e^\dag$ is a solution of $e$. To this end, consider the diagram below:
  $$
  \xymatrix@C+1.5pc{
    	& X \ar[d]_{e} \ar[rrr]^{e^\dag} \ar@{}[drrr]|{\text{(i)}}	&	&	& C	\\
    	& \extK H\extK X \ar[rrr]^{\extK \ol e^\dag} \ar[dr]^{\extK \ol e} \ar[d]_{\extK H\extK e} \ar `l[ld] `d[ddd] `r[ddrrr]_{\extK H\extK e^\dag} [ddrrr] \ar@{}[drrr]|{\text{(ii)}}	&	&	& \extK C \ar[u]_{\extalg b}	\\
    	& \extK H\extK \extK H\extK X \ar[r]^{\extK H\mu_{H\extK X}} \ar[d]_{\extK H\extK \extK \ol e^{\dag}}	& \extK H\extK H\extK X \ar[r]^-{\extK H\extK \ol e^\dag}	& \extK H\extK C \ar[r]^{\extK H\extalg b}	& \extK HC \ar[u]_{\extK c^{-1}}	\\
    	& \extK H\extK \extK C \ar[rru]^(0.35){\extK H\mu_C} \ar[rrr]^{\extK H\extK \extalg b} \ar@{}[drrr]|{\text{(iii)}}	&	&	& \extK H\extK C \ar[u]_{\extK H\extalg b}	\\
    	&	&	&	& 
  }
  $$
  All its inner parts commute: parts (i) and (iii) commute by the definition of $e^\dag$, for part (ii) use that $\ol e^{\dag}$ is a solution of $\ol e$ (i.\,e.~apply $\extK $ to Diagram~\refeq{eq:sol_ell-equation} with $\ol e$ in lieu of $e$), and the remaining parts commute by the definition of $\ol e$, naturality of $\mu$ and the multiplication law for the Eilenberg-Moore algebra $\extalg b$. Thus, the outside commutes,
   proving $e^\dag$ to be a solution of $e$.

  Secondly, suppose we are given any solution $e^\dag$ of $e$. Then we form
  $$
  \xymatrix{
    H\extK X
    \ar[r]^{H\extK e^{\dag}}
    &
    H\extK C
    \ar[r]^{H\extalg b}
    &
    HC
    \ar[r]^{c^{-1}}
    &
    C \,\text{,}
  }
  $$
  and we now prove that this is a solution of $\ol e$. Indeed, in the diagram
  $$
  \xymatrix@C+1.2pc{
    H\extK X \ar[dd]_{H\extK e} \ar[rr]^{H\extK e^\dag}	
    &	
    &
    H\extK C \ar[r]^{H\extalg b} \ar[drr]^(0.65){H\extalg b}	
    & 
    HC \ar[r]^{c^{-1}}	
    & 
    C	
    \\
    &	
    & 
    H\extK \extK C \ar[u]^{H\extK \extalg b} \ar[ddrr]^{H\mu_{C}}
    &
    & 
    HC \ar[u]_{c^{-1}}	
    \\
    H\extK \extK H\extK X 
    \ar[r]^(.4){\begin{turn}{-45}$\labelstyle H\extK \extK H\extK e^\dag$\end{turn}}
    \ar[d]_{H\mu_{H\extK X}}	
    & 
    H\extK \extK H\extK C \ar[r]^(.4){\begin{turn}{-45}$\labelstyle H\extK \extK H\extalg b$\end{turn}}
    \ar[d]_{H\mu_{H\extK C}}
    & H\extK \extK HC \ar[u]^{H\extK \extK c^{-1}} \ar[d]_{H\mu_{HC}}
    &
    &
    \\
    H\extK H\extK X \ar[r]_{H\extK H\extK e^\dag} \ar@{<-} `l[u]
    `[uuu]^{\ol e} [uuu]	
    & 
    H\extK H\extK C \ar[r]_{H\extK H\extalg b}	
    & H\extK HC \ar[rr]_{H\extK c^{-1}}
    &
    &
    H\extK C \ar[uu]_{H\extalg b}
  }
  $$
  all inner parts commute: for the big left-hand square apply $H\extK $ to Diagram~\refeq{eq:sol_ext_ell-equation}, the left-hand part is the definition of $\ol e$, the two lower squares and the lower right-hand triangle commute due to naturality of $\mu$, the upper right-hand triangle is trivial, and the remaining middle right-hand part commutes by the multiplication law for the Eilenberg-Moore algebra $\extalg b$. Thus, the outside commutes proving $c^{-1} \cdot H\extalg b \cdot H\extK e^{\dag}$ to be a solution of $\ol e$. Since $\ol e$ has a unique solution, we have
  $$c^{-1} \cdot H\extalg b \cdot H\extK e^{\dag} = \ol e^\dag \,\text{.}$$

  Lastly, the two constructions are inverse to each other: starting
  with the unique solution $\ol e^\dag$ of $\ol e$, it is clear that
  by applying the two constructions we obtain $\ol e^\dag$
  again. Starting with any solution $e^\dag$ of $e$, the application
  of the second construction results in the solution $c^{-1} \cdot
  H\extalg b \cdot H\extK e^{\dag} = \ol e^\dag$ of $\ol e$. The
  application of the first construction to that solution gives back the solution $e^\dag$ of $e$:
$$\extalg b \cdot \extK \ol e^\dag \cdot e = \extalg b \cdot \extK (c^{-1} \cdot H\extalg b \cdot H\extK e^{\dag})\cdot e = \extalg b \cdot \extK c^{-1} \cdot \extK H\extalg b \cdot \extK H\extK e^{\dag}\cdot e = e^\dag$$
where the last equality uses Diagram~\refeq{eq:sol_ext_ell-equation}. We conclude that $e$ has a unique solution $e^\dag$.
\end{proof}

\takeout{ 
\begin{exas} \label{ex:moreTM2}
We continue  our discussion of streams in Example~\ref{ex:moredist}.
With Theorem~\ref{thm:ellsolution}, one can solve systems like (\ref{TM})
uniquely.   But one cannot solve systems like
\begin{equation} \label{TMmore2}
\begin{array}{lcl}
x & = & \zip(1.\zip(x,y),0.\zip(y,y)) \\
y & = & \zip(1.y,1.x) \\
\end{array}
\end{equation}
The equation is sandwiched, since the right hand sides involve terms
where the prefixing operations $1.-$ and $0.-$ are within the terms
but not directly applied to variables. 
Theorem~\ref{thm:sandwich_ellsolution} tells us that there is a unique solution.
\end{exas}
}
 
\section{Completely Iterative Algebras}
\label{sec:cias}


We already mentioned in the introduction that the $H$-algebra $c^{-1}:
HC \to C$ is the initial completely iterative algebra for $H$. After
recalling this results below, it is our aim in this section to extend
Theorems~\ref{thm:ellsolution} and \ref{thm:sandwich_ellsolution} so as to obtain several new
structures of completely iterative algebras (for functors other than
$H$) on $C$.   Our first main results in this paper, Theorems~\ref{thm:distcia}
and~\ref{thm:sandwich},  go in this direction.  We
shall see in the next section that having these new cia structures allows
us to apply the existing theorems on completely iterative algebras to
uniquely solve much more general recursive equations than what we have seen
up to now and including this section. 

We now briefly recall the basic definitions and some examples; more details and examples can
be found in~\cite{m_cia,amv_elgot,mm}.

\begin{defi} \cite{m_cia}
\label{def:cia}
A \emph{flat equation morphism} in an object $A$ (of parameters) is a morphism $e:X\to HX+A$. An $H$-algebra
$a:HA\to A$ is called \emph{completely iterative} (or a
\emph{cia}, for short) if every flat equation morphism in $A$ has a
unique solution, i.\,e., for every $e:X\to HX+A$ there exists a unique
morphism $e^{\dagger}:X\to A$ such that the square below commutes:
\[
\xymatrix@C+1pc{
  X
  \ar[d]_e
  \ar[r]^{\sol e}
  &
  A
  \\
  HX + A
  \ar[r]_-{H\sol e + A}
  &
  HA + A
  \ar[u]_{[a, A]}
}
\]
\takeout{ 
$$
e^{\dagger}
=
(\xymatrix@1{
  X \ar[r]^(0.35)e & 
  HX + A \ar[rr]^{H\sol e + A} &&
  HA + A \ar[r]^-{[a, A]} &
  A
  }).
$$}
\end{defi}

\begin{exas}
\label{ex:cias}
We recall some examples from previous work.

\medskip\noindent 
(1)~Let $TX$ denote a terminal coalgebra for
$H(-)+X$. (We assume that $TX$ exists for all $X$.) Its structure is an isomorphism by Lambek's Lemma%
\iffull~\cite{lambek}\fi, and so its inverse yields (by composing with the
coproduct injections) an $H$-algebra $\tau_X:HTX\to TX$ and a morphism
$\eta_X:X\to TX$. Then $(TX,\tau_X)$ is a free cia on $X$ with the
universal arrow $\eta_X$, see \cite{m_cia}. Conversely, for any free
cia $(TX, \tau_X)$ with the universal morphism $\eta_X$, $[\tau_X,
\eta_X]$ is an isomorphism and $[\tau_X, \eta_X]^{-1}$ is the
structure of a terminal coalgebra for $H(-) + X$. So in particular,
the inverse of the structure $c: C \to HC$ of the terminal
coalgebra for $H$ is, equivalently, an initial cia for $H$.

\medskip\noindent 
(2)~Let $H_{\Sigma}$ be a polynomial functor (cf.~Example~\ref{ex:sigma_algebras}). The
terminal coalgebra for $H_{\Sigma}(-)+X$ is carried by the set
$T_{\Sigma}X$ of all (finite and infinite) $\Sigma$-trees on $X$. According to the previous
item, this is a free cia for $H_{\Sigma}$ on $X$. As already mentioned
in the introduction, cias for $H_\Sigma$ are $\Sigma$-algebras in
which systems of recursive equations~\refeq{eq:eq} have unique
solutions. For example, let $\Sigma$ be the signature with one binary
operation symbol $*$ and one constant symbol $c$. Consider the free
cia $A = T_\Sigma Y$ and let $t \in T_\Sigma Y$. Then the system
\[
\begin{array}{rcl@{\qquad\qquad}rcl}
  x_0 & = & x_1 * x_2 & x_2 & = & c \\
  x_1 & = & x_0 * x_3 & x_3 & = & t
\end{array}
\]
has the unique solution $\sol e: X \to T_\Sigma Y$ given by
\[
\sol e(x_0) = 
\vcenter{
\xy
\POS (0,0) *+{*} = "w"
,    (-5,-10) *+{*} = "l"
,    (5,-10)  *+{c} = "r"
,    (-10,-20) *+{\vdots} = "ll"
,    (0,-20) *+{t} = "lr"
\ar @{-} "w";"l"
\ar @{-} "w";"r"
\ar @{-} "l";"ll"
\ar @{-} "l";"lr"
\endxy
}
\qquad
\sol e(x_1) =
\vcenter{
\xy
\POS (0,0) *+{*} = "w"
,    (-5,-10) *+{*} = "l"
,    (5,-10)  *+{t} = "r"
,    (-10,-20) *+{\vdots} = "ll"
,    (0,-20) *+{c} = "lr"
\ar @{-} "w";"l"
\ar @{-} "w";"r"
\ar @{-} "l";"ll"
\ar @{-} "l";"lr"
\endxy
}
\qquad
\sol e(x_2) = c
\qquad
\sol e(x_3) = t\,.
\]
\medskip\noindent (3)~The algebra of addition on
$\bar\Nat=\{1,2,3,\dots\}\cup\{\infty\}$ is a cia for $HX=X\times X$,
see \cite{amv_atwork}.

\medskip\noindent (4)~Let $\mathcal{A}=\CMS$ be the category of
complete metric spaces with distances in $[0,1]$ and with non-expanding maps as
morphisms, and let $H$ be a \emph{contracting} endofunctor of
$\CMS$ (see e.\,g.~\cite{america+rutten}). Then any non-empty algebra for $H$ is
a cia, see \cite{m_cia} for details. For example, let $A$ be the
set of non-empty compact subsets of the unit interval $[0,1]$ equipped
with the Hausdorff metric \cite{hausdorff}. This complete metric space can be turned
into a cia such that the Cantor set arises as the unique solution of a
flat equation morphism (see \cite[Example~3.3(v)]{mm}).

\medskip\noindent (5)~Unary algebras over $\Set$. Here we take
$\mathcal{A}=\Set$ and $H=\Id$. An algebra $\alpha:A\to A$ is a cia
iff $\alpha$ has a fixed point $a_0$ and there is no infinite sequence
$a_1,a_2,a_3,\dots$ with $a_i=\alpha(a_{i+1})$, $i=1,2,3,\dots$,
except for the one all of whose members are $a_0$. The second part of
this condition can be put more vividly as follows: the graph with node
set $A\setminus \{a_0\}$ and with an edge from $\alpha(a)\neq a_0$ to
$a$ for all $a$ is well-founded. 

\medskip\noindent (6)~Classical algebras are seldom cias. For example
a group or a semilattice is a cia (for $HX=X\times X$) iff they
contain one element only (consider the unique solutions of $x =
x\cdot 1$ or $x = x\vee x$, respectively).
\end{exas}

\begin{rem}
  \label{rem:comp}
  In~\cite{amv_elgot} the following property---called
  \emph{compositionality}---of taking unique
  solutions of flat equations in a cia $a: HA \to A$ was proved. Suppose
  we have two flat equation morphisms $e: X \to HX + Y$ and $f: Y \to
  HY + A$. We form  
  $$
  \sol f \after e = (
  \xymatrix@C+1pc{
    X
    \ar[r]^-{e}
    &
    HX + Y
    \ar[r]^-{HX + \sol f}
    &
    HX + A
  }
  )
  $$
  and
  $$
  f \plus e = (
  \xymatrix@C+1pc{
    X + Y
    \ar[r]^-{[e, \inr]}
    &
    HX + Y
    \ar[r]^-{HX + f}
    &
    HX + HY + A
    \ar[r]^-{\can + A}
    &
    H(X+Y) + A
  }
  )\,\text{.}
  $$
  Then
  $$
  \sol{(f \plus e)} = (
  \xymatrix@C=2cm{
    X+Y
    \ar[r]^-{[\sol{(\sol f \after e)},\sol f]}
    &
    A
  }
  ) \,\text{.}
  $$
  This gives a first precise formulation of the modularity principle
  we mentioned in the introduction, albeit restricted to solutions of
  flat equation morphisms. Indeed, the above equation states that in
  order to obtain the simultaneous unique solution of $e$ and $f$
  (i.\,e.~$\sol{(f \plus e)}$) one may first solve $f$ in the cia $A$,
  then plug its solution $\sol f$ as constant parameters into $e$ and
  finally solve the resulting equation (i.\,e., one takes $\sol{(\sol
    f \after e)})$. For $\A = \Set$ one can view this as using the
  elements $\sol f (y)$, $y \in Y$, as new constants in $A$ in the
  subsequent recursive equation given by $e$. We shall see modularity
  principles for more general formats of recursive definitions in
  Sections~\ref{sec:solthms} and~\ref{sec:lambdarps}.

  \takeout{ 
  One observation is that this implies \emph{modularity} of unique solutions of flat equations: one may first solve $f$ in the
  cia $A$, then plug its solution $\sol f$ as constant parameters
  into $e$ and finally the resulting equation can be solved uniquely again (one
  takes $\sol{(\sol f \after e)}$ in the cia $A$). Moreover, we have \emph{compositionality} of unique solutions of flat equations in the sense that we may compose flat equations (one forms $f \plus e$) and that we can solve the composite equation by solving its constituents $f$ and $e$ one after the other (i.\,e. in a modular way as described before)---this is provided by the equation $\sol{(f \plus e)} = [\sol{(\sol f \after e)},\sol f]$.}
  \takeout{
  {\bf Functoriality:} Let $e: X \to HX + A$ and $f:
  Y \to HY + A$ be two flat equation morphisms, and let $h: X \to Y$
  be a homomorphism of equations, i.\,e., $f \cdot h = (Hh + \id_A) \cdot
  e$. Then we have $\sol e = \sol f \cdot h$.}
\end{rem}

The following two theorems show that abstract GSOS rules induce further
structures of completely iterative algebras on the (carrier of the) terminal
$H$-coalgebra $C$ besides the structure of an initial cia for $H$.

\begin{ass}
As in the previous section,
 we shall write $(M,\eta,\mu)$ for the free monad on $K$ to simplify notation, and we assume that $\ell:K(H\times\Id)\to HM$ is an abstract GSOS rule.
\end{ass}

\begin{thm}\label{thm:distcia}
  Consider the algebra
  \iffull
  $$k =(\xymatrix@1{HMC \ar[r]^-{H\extalg b} & HC \ar[r]^-{c^{-1}} & C}) \,\text{,}$$
  \else
  $k = (HMC \xrightarrow{H\extalg b} HC \xrightarrow{c^{-1}} C)$,
  \fi
  where $b: KC \to C$ is the
  $\ell$-interpretation in $C$. Then $(C, k)$ is a cia for the functor $HM$.
\end{thm}
\begin{proof}
  Let $e: X \to HMX + C$ be a flat equation morphism. We must prove that there
  exists a unique morphism $\sol{e}: X \to C$ such that the following square
  commutes:
  $$
  \xymatrix@C+2pc{
    X \ar[r]^-{\sol{e}} 
    \ar[d]_e
    & C \\
    HMX + C 
    \ar[r]_-{HM\sol{e} + C}
    &
    HMC + C
    \ar[u]_{[k, C]}
    }
  $$
  We start by forming the $\ell$-equation 
  $$
  \ol{e} = (
  \xymatrix{
    X + C \ar[r]^-{[e, \inr]} 
    & HMX + C \ar[rr]^-{\ HMX + H\eta_C \cdot c\ }
    && HMX + HMC \ar[r]^-{\can}
    & HM(X+C)
  }
  )\,\text{.}
  $$
  By Theorem~\ref{thm:ellsolution}, there exists a unique
  morphism $s: X + C \to C$ such that the square below commutes:
  \begin{equation}\label{diag:distsol}
    \vcenter{
  \xymatrix@C+1pc{
    X + C \ar[r]^-{s} \ar[dd]_{\ol e} 
    & 
    C \ar[d]^c 
    \\
    & 
    HC 
    \\
    HM(X+C) \ar[r]_-{HMs} 
    & 
    HMC \ar[u]_-{H\extalg b}
    }}
  \end{equation}
  We will now prove that the morphism $\sol{e} = s \cdot \inl: X \to C$ is the
  desired unique solution of $e$. We begin by proving that the equation $s \cdot
  \inr =  \id_C$ holds. Indeed, consider the diagram below:
  \begin{equation}\label{diag:id}
    \vcenter{
  \xymatrix@C+1pc{
    C 
    \ar[rr]^-{\inr} 
    \ar[dd]_c
    &
    &
    X + C
    \ar[d]_{\ol e}
    \ar[r]^-s
    &
    C 
    \\
    &
    HMC
    \ar[r]_-{HM\inr}
    &
    HM(X+C)
    \ar[r]_-{HMs}
    &
    HMC 
    \ar[u]_{k}
    \\
    HC 
    \ar[ru]_-{H\eta_C}
    \ar[rrr]_-{H(s \cdot \inr)}
    & 
    &
    &
    HC 
    \ar[u]_{H\eta_C}
    \ar `r[u] `[uu]_{c^{-1}} [uu]
    }}
  \end{equation}
  This diagram commutes: the upper right-hand square is
  Diagram~\refeq{diag:distsol} above, the upper left-hand part commutes by
  the definition of $\ol{e}$, the lower part commutes by the naturality of
  $\eta$, and the right-hand part follows from the definition of $k = c^{-1} \cdot
  H\extalg b$ and the unit law $\extalg b \cdot \eta_C = \id_C$ of the Eilenberg-Moore algebra
  $(C,\extalg b)$, cf.~Notation \ref{not:extalg}. Hence, we see that $s \cdot \inr$ is a coalgebra homomorphism from the
  terminal coalgebra $(C,c)$ to itself. Thus, $s \cdot \inr$ must be the identity as
  desired. 

  Next we prove that $\sol{e}$ is a solution of $e$. To this end we verify that
  the following diagram commutes:
  \begin{equation}\label{diag:sol}
    \vcenter{
  \xymatrix{
    X
    \ar[r]^-{\inl} 
    \ar[ddd]_-e
    &
    X + C 
    \ar[rr]^-{s}
    \ar[d]_{\ol e}
    &
    &
    C
    \ar@{<-} `u[l] `[lll]_{\sol{e}} [lll]
    \\
    & 
    HM(X+C) \ar[r]_{HMs}
    &
    HMC 
    \ar[ru]_{k}
    \\
    &
    HMX + HMC 
    \ar[u]_{\can}
    \ar[ru]_(.7)*++{\labelstyle [HM\sol{e}, HMC]}
    \\
    HMX + C 
    \ar[ru]_-*+{\labelstyle HMX + H\eta_C \cdot c}
    \ar[rrr]_-{HM\sol{e}+ C}
    &&&
    HMC + C 
    \ar[uuu]_{[k, C]}
    }
  }
  \end{equation}
  The upper part is the
  definition of $\sol{e}$, the left-hand part commutes by the definition of
  $\ol{e}$, the upper right-hand part is Diagram~\refeq{diag:distsol}, and that
  the inner triangle commutes follows from the definition of $\sol{e}$ and the
  fact that $s \cdot \inr = \id_C$. Finally, we consider the two coproduct
  components of the lower right-hand triangle separately; the left-hand component trivially
  commutes, and for the right-hand one we compute as follows:
  $$
  \begin{array}{rcl@{\qquad} p{3.5cm}}
    k \cdot H\eta_C \cdot c & = & c^{-1} \cdot H\extalg b \cdot H\eta_C \cdot c & (by the definition of
    $k$)  \\
    & = & c^{-1} \cdot c & (since $\extalg b \cdot \eta_C = \id_C$) \\
    & = & \id_C \,\text{.}
  \end{array}
  $$
  
  To complete our proof we show that $\sol{e}=s \cdot \inl:X\to C$ is the unique solution of
  $e$. So suppose that we are given any solution $\sol{e}$ of $e$. Now form the
  morphism $s = [\sol{e}, \id_C]$. We are finished if we show that  Diagram~\refeq{diag:distsol} commutes
  for this
  morphism $s$.   We verify the two
  coproduct components separately: the right-hand component is checked using
  Diagram~\refeq{diag:id}: since $s \cdot \inr = \id_C$ the
  outside of \refeq{diag:id} commutes, and commutativity of the desired upper right-hand square composed with $\inr$ follows since all other inner parts commute as described below~\refeq{diag:id}. The
  commutativity of the left-hand component is established using
  Diagram~\refeq{diag:sol}; indeed, since the outside and all other parts of that
  diagram commute for our morphism $s$ so does the desired upper right-hand
  part composed with $\inl$. 
\end{proof}

\begin{thm} {\rm (Sandwich Theorem)}
  \label{thm:sandwich}
  Consider the algebra
  \iffull
  $$k' =(\xymatrix@1{MHMC \ar[r]^-{Mk} & MC \ar[r]^-{\extalg b} & C}) \,\text{,}$$
  \else
  $k' = MHMC \xrightarrow{Mk} MC \xrightarrow{\extalg b} C$,
  \fi
  where $b: KC \to C$ is the
  $\ell$-interpretation in $C$ and $k = c^{-1} \cdot H\extalg b$ as
  in Theorem~\ref{thm:distcia}. 
  Then $(C, k')$ is a cia for the functor $MHM$.
\end{thm}
\begin{proof}
  The proof is a slight modification of the proof of Theorem~\ref{thm:distcia}: now we are given a flat equation morphism $e: X \to MHMX + C$ and form the morphism 
  $$
  \ol{e} = (
  \xymatrix{
    X + C \ar[r]^-{[e, \inr]} 
    & MHMX + C \ar[d]^{\ MHMX + MH\eta_C \cdot \eta_{HC} \cdot c\ }
    &
    \\
    & MHMX + MHMC \ar[r]^-{\can}
    & MHM(X+C) \, )\text{.}
    }
  $$
  This morphism $\ol{e}$ is a sandwiched $\ell$-equation and we invoke Theorem~\ref{thm:sandwich_ellsolution} to see that it has a unique solution $s$. The rest of this proof is left to the reader since it is very close to the one of Theorem~\ref{thm:distcia}.
\end{proof}

\begin{rem} \label{rem:arbitrary_monad}
We discuss a few generalizations of the results in this section obtained by weakening
the assumptions on $M$ or on $(C,c^{-1})$.

  \medskip\noindent
  (1) In our statements of Theorems \ref{thm:distcia} and \ref{thm:sandwich}, 
  $M$ was the free monad on $K$.   However, it is possible to abstract away from this, by considering
  an  
 \emph{arbitrary} monad $M$.  
   In this setting,  we take 
  \[
  \lambda: M (H \times \Id) \to (H \times \Id)M
  \] 
 to be a distributive law of the monad $M$ over the cofree copointed
  functor $H\times\Id$.   This is all we need in order to define  the $\lambda$-interpretation $\extalg b:MC\to C$ such that 
  Diagram~\refeq{diag:interp_gsos_lambda} commutes. 
  cf.~Theorem~\ref{thm:interp_gsos_rule} and
  Remark~\ref{rem:interp_gsos_lambda}(2). The version of
  Theorem~\ref{thm:ellsolution} presented in
  \cite{bartels,bartels_thesis} and dually in \cite{uvp} states that
  for every $e:X\to HMX$ there is a unique solution $e^\dag:X\to C$,
  i.\,e., $e^\dag$ is such that the diagram in \refeq{eq:sol_ell-equation} commutes.  Inspection of the proofs reveals that
  all our results so far hold in this generality . So Theorem~\ref{thm:distcia} shows
  that $(C,k)$ is a cia for $HM$, and Theorem~\ref{thm:sandwich} shows that $(C,k')$ is a cia for $MHM$.
  
  \medskip\noindent
  (2) Going even further, not all of the monad structure has been used in our work.
  Observe that our proof of Theorem~\ref{thm:distcia} only makes
  use of the unit $\eta:\Id\to M$ of the monad $M$, not the multiplication. In fact, there are
  versions of Theorem \ref{thm:distcia} and \ref{thm:sandwich} that
  hold for a pointed functor $M$ in lieu of a monad and for a given
  distributive law of $M$ over the cofree copointed functor
  $H\times\Id$ or the functor $H$, respectively. However, in this case
  we need to assume that the category $\A$ is cocomplete. The
  technical details are somewhat different than what we have seen and
  we discuss them in detail in 
  \ifappendix the appendix.\else an appendix that is provided as a
  supplementary file with our paper.

  \medskip\noindent
  (3)~Capretta, Uustalu and Vene~\cite{cuv} extend (the dual) of
  Theorem~\ref{thm:ellsolution} by replacing the algebra $(C, c^{-1})$
  arising from the terminal coalgebra by an algebra $a: HA \to A$
  having the property that for every coalgebra $e: X \to HX$ there
  exists a unique coalgebra-to-algebra homomorphism from $(X,e)$ to
  $(A, a)$. One may ask whether the  two theorems above can be extended
  from the initial cia to an arbitrary cia for $H$. However, note that our
  proof makes use of the fact that $c: C \to HC$ is an isomorphism. So the
  desired extension of our results is not obvious, and we leave this
  as an open problem for further work. 
\end{rem}

Theorems~\ref{thm:distcia} and \ref{thm:sandwich} extend Theorems~\ref{thm:ellsolution} and \ref{thm:sandwich_ellsolution} in two
important ways. Firstly, the structure of a cia allows one to reuse
solutions of a given flat equation morphism by using constants in $C$
on the right-hand sides of recursive equations (cf.~Remark~\ref{rem:comp}). 
This gives a
clear explanation of why it is possible to use recursively defined
objects (processes, streams, etc.) in subsequent recursive
definitions. This kind of modularity of the unique solutions is
a useful and desired property often employed in specifications. We
shall discuss a concrete instance of this in Example~\ref{ex:comp}. 

Secondly, Theorem~\ref{thm:sandwich} permits the right-hand sides of
recursive specifications to be from a wider
class. 
For example, Milner's solution theorem for CCS (see~\cite{milner},
Chapter~4, Proposition~14) allows recursion over process terms $E$ in
which the recursion variables occur within the scope of some prefixing
combinator $a.-$. This combinator can occur \emph{anywhere within
  $E$}, not necessarily at the head of that term (cf.~Example~\ref{ex:moreTM}). Hence,
Theorem~\ref{thm:sandwich} allows us to obtain Milner's result as a
special case, directly. This will be explained in detail in
Section~\ref{sec:app}.1.

\section{Solution Theorems for Free}
\label{sec:solthms}

Using the new cia structures obtained from Theorems~\ref{thm:distcia}
and~\ref{thm:sandwich}, the existing body of results on the semantics of
recursion in cias~\cite{aamv,m_cia,mm}  now gives us further theorems. 
\takeout{We shall state and explain those
results in the current section. First we briefly review the necessary
background material from~\cite{m_cia,mm}. For a well-motivated and
more detailed exposition we have to refer the reader to loc.~cit.
}

\takeout{
\begin{ass}
  Throughout this section we assume that the functor $H$ is
  iteratable, i.\,e., for every of $X$ of $\A$ there exists a terminal
  coalgebra $\T H X$ for $H(-) + X$.
\end{ass}
}
We begin with a terse review of some terminology from the
area. We assume that in addition to the terminal $H$-coalgebra $C$, for
every object $X$ the terminal coalgebra $\T H X$ for $H(-) + X$ exists,
i.\,e., in the terminology of~loc.~cit., $H$ is
\emph{iteratable}. Our examples in~\ref{ex:dist} are all iteratable
endofunctors of $\Set$. 

As explained in Example~\ref{ex:cias}(1), the structure of the
terminal coalgebra $\T H X$ yields the free cia on $X$ with its structure
and universal arrow as displayed below:
\[
\xymatrix@1{
H \T H X \ar[r]^{\tau^H_X} & \T H X
}
\qquad
\xymatrix@1{
X \ar[r]^-{\eta^H_X} & \T H X \,\text{.}
}
\]
From this it easily follows that $\T H$ is the object assignment
of a monad and that $\eta^H$ and $\tau^H$ are natural
transformations. Denote by $\kappa^H$ the natural transformation
\[
\kappa^H = (\xymatrix@1{
H \ar[r]^-{H\eta^H} & H\T H \ar[r]^-{\tau^H} & \T H
}) \,\text{.}
\]
It was proved in~\cite{aamv,m_cia} that the monad $\T H$ is
characterized as the free completely iterative monad on $H$ (with the
universal natural transformation $\kappa^H$). We shall not recall the
concept of a completely iterative monad as it is not needed in the
present paper. However, we shall need that the assignment $H \mapsto
\T H$ readily extends to natural transformations. Let $h: H \to H'$ be
a natural transformation between iteratable endofunctors. Then by the
universal property of $\T H$ we have a unique monad morphism
\begin{equation}
  \label{eq:Tmor}
\T h: \T H \to \T{H'}
\qquad
\text{such that}
\qquad
\vcenter{
  \xymatrix{
    H \ar[r]^-{\kappa^H} \ar[d]_h & 
    \T H \ar[d]^{\T h}\\
    H' \ar[r]^-{\kappa^{H'}} & 
    \T{H'}
  }
}
\end{equation}
commutes.

Finally, let $(A, a)$ be a cia for $H$. Then there is a unique
$H$-algebra homomorphism
\[
\aext: \T H\! A \to A
\qquad
\text{such that}
\qquad
\aext \cdot \eta^H_A = \id_A \,\text{.} 
\]
We call $\aext$ the \emph{evaluation morphism} associated with $A$. It
is easy to prove that $\aext \cdot \kappa^H_A = a$. 

\begin{rem} 
%
  In the case of a polynomial functor $H_\Sigma$ on $\Set$, the
  evaluation morphism $\aext$ can be thought of as a map that takes a
  (\emph{not} necessarily finite)
  $\Sigma$-tree $t$ with variables in the cia $A$ and computes the value
  of $t$ in $A$ using the algebraic operations on $A$ given by the
  structure $a: H_\Sigma A \to A$. 
\end{rem}

In previous work it was shown how to obtain unique solutions of
more general (first order) recursive equations than the flat ones
appearing in the definition of a cia: 
  
\begin{defi} {\rm\cite{aamv,m_cia}}
  An \emph{equation morphism} is a morphism of the form $e: X \to
  \T H (X+A)$. It is called \emph{guarded} if there exists a factorization
  $f: X \to H\T H (X+A) + A$ such that
  $$
  \xymatrix{
    X
    \ar[r]^-e 
    \ar[rd]_-{f}
    &
    \T H (X + A)
    \\
    &
    H\T H (X+A) + A
    \ar[u]_{[\tau^H_{X+A}, \eta^H_{X+A} \cdot \inr]}
  }
  $$

  A \emph{solution} of an equation morphism $e$ in a cia $(A,a)$ is a
  morphism $\sol e: X \to A$ such that the following square commutes:
  $$
  \xymatrix@C+2pc{
    X
    \ar[r]^-{\sol e}
    \ar[d]_e
    &
    A
    \\
    \T H (X+A)
    \ar[r]_-{\T H [\sol e, \id_A]}
    &
    \T H A
    \ar[u]_{\aext}
  }
  $$
\end{defi}

\begin{exa}
  For a polynomial functor $H_\Sigma$ on $\Set$ an equation morphism
  $e: X \to \T{H_\Sigma} (X + A)$
  corresponds to a system of equations~\refeq{eq:eq}, where each
  right-hand side $t_i$ is a (finite or infinite) $\Sigma$-tree with leaves labeled by a
  variable $x_j$ or elements $a \in A$. Guardedness is the syntactic
  restriction that no right-hand side tree $t_i$ is simply a single
  node tree with a variable as a label. A solution assigns to every
  variable $x_i$ an element $a_i \in A$ such that 
  $a_i = \aext(t_i[\vec{a_j} / \vec{x_j}])$ for every $i \in I$,
  i.\,e., if we substitute all the solutions $a_j$ for the corresponding
  variables $x_j$ in $t_i$ and evaluate the resulting tree in $A$
  using $\aext$, then we obtain $a_i$. 
\end{exa}

\begin{thm} {\rm\cite{m_cia}}
  \label{thm:sol}
  Let $(A,a)$ be a cia for $H$. Then every guarded equation morphism
  has a unique solution in $A$.
\end{thm}


An even more general property of cias was proved in~\cite{mm};
one can solve recursive function definitions uniquely in a
cia. We recall the respective result.

\begin{defi}
  \label{dfn:algeq}
  Let $\Var$ be an endofunctor such that $H+\Var$ is iteratable. A
  \emph{recursive program scheme} (rps, for short) is a natural transformation $e: \Var
  \to \T{H+\Var}$. It is called \emph{guarded} if there exists a
  natural transformation $f: \Var \to H\T{H+\Var}$ such that 
  $$
  e = (
  \xymatrix{
  \Var
  \ar[r]^-{f}
  &
  H\T{H+\Var}
  \ar[rr]^-{\inl \T{H+\Var}}
  &&
  (H+\Var)\T{H+\Var}
  \ar[r]^-{\tau^{H+\Var}}
  &
  \T{H+\Var}
  }
  )\,\text{,}
  $$
  where $\inl: H \to H + \Var$ is the coproduct injection.

  Now let $(A,a)$ be a cia for $H$. An \emph{interpreted solution} of $e$ in
  $A$ is a $\Var$-algebra structure $\isol e A: \Var A \to A$ giving
  rise to an Eilenberg-Moore algebra structure $\beta: \T {H+\Var} A
  \to A$ with $\beta \cdot \kappa_A^{H+\Var} = [a, \isol e A]$ and such that we
  have 
  \begin{equation}
    \label{eq:intsol}
    \isol e A = (
    \xymatrix{
      \Var A
      \ar[r]^-{e_A}
      &
      \T{H+\Var} A
      \ar[r]^-{\beta}
      &
      A
    }
    ) \,\text{.}
  \end{equation}
\end{defi}

\begin{exa}
  As explained in~\cite{mm}, for polynomial set functors, recursive program schemes
  as defined above provide a categorical formulation of
  recursive function definitions such as~\refeq{eq:rps} from the
  introduction. For example, let $VX = X$ be the polynomial functor
  associated with the signature with one unary operation symbol $f$ and
  let $HX = X \times X + X$ be the polynomial functor for the signature
  of the givens $F$ and $G$. Then~\refeq{eq:rps} uniquely determines
  the natural transformation
  \[
  e: V \to \T{H+V} 
  \qquad
  \text{with}
  \qquad
  e_X(x) = 
  \vcenter{
    \xy
    \POS   (0,0) *+{F} = "w"
       ,   (-5,-10) *+{x} = "l"
       ,   (5,-10)   *+{G} = "r"
       ,   (5,-20)   *+{f} = "rr"
       ,   (5,-30)   *+{x} = "rrr"
    \ar @{-} "w";"l"
    \ar @{-} "w";"r"
    \ar @{-} "r";"rr"
    \ar @{-} "rr";"rrr"
    \endxy
  }
  \]
  Here guardedness corresponds to the syntactic restriction that the
  right-hand side of~\refeq{eq:rps} starts with a given operation symbol
  such as $F$. Any cia $A$ for $H$ provides interpretations of the given operation
  symbols $F$ and $G$ as actual operations $F_A: A \times A \to A$ and
  $G_A: A \to A$, and an interpreted solution in $A$ is precisely a new unary operation
  $f_A: A \to A$ such that for all $a \in A$, $f_A(a) = F_A(a,
  G_A(f_A(a))$. 
\end{exa}

\begin{thm} {\rm\cite{mm}}
  \label{thm:rps}
  In a cia, every guarded rps has a unique interpreted solution.
\end{thm}

We are now able to prove more. The next theorem implies modularity of taking solutions of recursive program schemes: 
operations obtained as solutions of recursive program schemes can be
used as givens in subsequent  definitions of other recursive program schemes.  These new schemes will still
have unique solutions. For the special case of interpreted rps
solutions in cias this strengthens the results in~\cite{mm_prop}.

\begin{thm}
  \label{thm:extcia}
  Let $e: \Var \to \T{H+\Var}$ be a guarded rps, and let $a: HA \to A$ be a
  cia. Then the interpreted solution $\isol e A: \Var A \to A$ extends the cia
  structure on $A$; more precisely, the algebra $[a, \isol e A]: (H+\Var)A \to
  A$ is a cia for $H+\Var$.
\end{thm}

\begin{rem}\label{rem:techprelims}
  For the proof we need to recall some technical
  details. Recall that any guarded rps $e: \Var \to \T{H+\Var}$ as in 
  Definition~\ref{dfn:algeq} induces a natural transformation
  $$
  \ol{e}: \T{H+\Var} \to H \T{H+\Var} + \Id
  $$ 
  (see~\cite[Lemma~6.9]{mm}).
  The component $\ol{e}_A$ of this natural transformation at $A$ is a flat
  equation morphism with parameters in $A$. Its unique solution in the cia $(A,a)$ is the
  Eilenberg-Moore algebra structure $\beta: \T{H+\Var} A \to A$
  in~\refeq{eq:intsol} satisfying $[a, \isol e A] = \beta \cdot
  \kappa^{H+\Var}_A$ (this follows from~\cite{mm}, see Lemma~7.4 and
  the proof of Theorem~7.3).  
\end{rem}

\begin{proof}[Proof of Theorem~\ref{thm:extcia}]
  Let $m: X \to (H+\Var)X + A$ be a flat equation morphism. We need to prove
  that $m$ has a unique solution $s$. As shortcut notations we shall write $T'$
  for $\T{H+\Var}$, $\tau_X': (H+\Var) T'X \to T'X$ for the 
  corresponding structure of a free cia for $H+\Var$ as well as $\eta'$ and $\mu'$
  for the unit and multiplication of the monad $T'$ and $\kappa' = \tau' \cdot
  (H+\Var)\eta'$ (cf.~the introduction to Section~\ref{sec:solthms}). 

  (1)~Existence of a solution. Since $T'A$ is the terminal coalgebra for
  $(H+ \Var)(-) + A$ we have a unique  homomorphism $h: X \to T'
  A$. We show that 
  $$
  s = (\xymatrix@1{X \ar[r]^-{h} & T'A \ar[r]^\beta & A})
  $$
  is a solution of $m$ in the algebra $(A, [a,\isol e A])$. To see
  this, consider the following diagram:
  $$
  \xymatrix@C+2.5pc{
    X \ar[r]^h \ar[d]_m &
    T'A 
    \ar[r]^\beta &
    A
    \\
    (H+ \Var)X + A
    \ar[r]_-{(H+\Var) h + A}
    &
    (H+\Var)T'A + A
    \ar[u]_{[\tau_A',\eta_A']}
    \ar[r]_{(H+\Var)\beta + A}
    &
    (H+\Var)A + A
    \ar[u]_{[a,\isol e A, A]}
    }
  $$
  The left-hand square commutes since $h$ is a coalgebra homomorphism,
  and for the right-hand component of the right-hand square use the
  unit law $\beta \cdot \eta'_A = \id_A$ of the Eilenberg-Moore algebra
  $\beta$. It remains to prove the commutativity of the left-hand
  component. This is established by inspecting the diagram below:
  \begin{equation}
  \label{diag:42}
  \vcenter{
  \xymatrix@C+1pc{
    (H+\Var) T'A
    \ar[r]_-{\kappa_{T'A}'}
    \ar[d]_{(H+\Var)\beta}
    &
    T'T'A
    \ar[d]_{T'\beta}
    \ar[r]_-{\mu'_A}
    &
    T'A
    \ar[d]^{\beta}
    \ar@{<-} `u[l] `[ll]_{\tau'_A} [ll]
    \\
    (H + \Var)A
    \ar[r]_-{\kappa'_A}
    &
    T'A
    \ar[r]_-{\beta}
    &
    A
    \ar@{<-} `d[l] `[ll]^{[a,\isol e A]} [ll]
  }}
  \end{equation}
The commutativity of the upper part  is standard (see Corollary 3.17 in~\cite{mm}),
  for the lower one see Remark~\ref{rem:techprelims}, the left-hand
  inner square commutes due to naturality of $\kappa'$, and the
  right-hand inner square by one of the laws for the Eilenberg-Moore algebra
  $\beta$. 

  \takeout{
    The existence of a solution for $m$ is clear: the Eilenberg-Moore algebra
    $\beta: T'A \to A$ shows us that $A$ is a complete Elgot algebra for
    $H+\Var$, and so $A$ comes with a canonical choice of a solution for every
    flat equation morphism $m$, see~\cite{amv_elgot,mm}. It is our only task to show
    uniqueness of solutions. }

  \medskip
  (2)~Uniqueness of solutions. Since $e$ is a guarded rps, it factors through some
  $f: \Var \to HT'$ (cf.~Definition~\ref{dfn:algeq}). From $f$
  and $m$ we form a flat equation morphism 
  \[
  g: X + T'X \to H(X+T'X) + A
  \]
  w.\,r.\,t.~$H$ as follows. The left-hand component of $g$ is
  $$
  g \cdot \inl = (
  \xymatrix@1@C+.5pc{
    X 
    \ar[r]^-m
    &
    HX +\Var X +A  
    \ar[rr]^-{HX + {f}_X + A}
    &&
    HX + HT'X + A
    \ar[r]^-{\can + A}
    &
    H(X+T'X) +A
  }
  ) \,\text{,}
  $$
  and the right-hand component of $g$ is 
  $$
  g \cdot \inr = (
  \xymatrix@1@C+1pc{
    T'X 
    \ar[r]^-{\ol e_X} 
    &
    HT'X + X
    \ar[rr]^-{[\inl \cdot H\inr, g \cdot \inl]}
    &&
    H(X + T'X) + A
  }
  ) \,\text{.}
  $$
  Since $(A,a)$ is a cia for $H$ there exists a unique solution $\sol{g}:
  X + T'X \to A$.  
  Now let $s: X \to A$ be any solution of the flat equation morphism $m$ in the
  algebra $[a, \isol e A]: (H+\Var) A \to A$. We will show below that
  $[s,\beta \cdot T's]: X + T'X \to A$ is a solution of $g$ in the
  $H$-algebra $(A,a)$. So since $(A,a)$ is a cia we have the following equation:
  \begin{equation}\label{eq:solg}
    \sol{g} = [s, \beta \cdot T's]: X + T'X \to A \,\text{.}
  \end{equation}
  Then $s$ is uniquely determined by $\sol{g}$.

  In order to prove Equation~\refeq{eq:solg} we need to verify that the
  following square commutes:
  \begin{equation}\label{diag:solg}
    \vcenter{
    \xymatrix@C+.5pc{
      X + T'X 
      \ar[rr]^-{[s, \beta \cdot T's]}
      \ar[d]_g
      &&
      A
      \\
      H(X+T'X) + A
      \ar[rr]_-{H[s, \beta \cdot T's] +A}
      &&
      HA + A
      \ar[u]_{[a,A]}
      }}
  \end{equation}
  We shall verify the commutativity of the two coproduct components separately.
  For the left-hand component we consider the diagram below:
  \begin{equation}\label{eq:firstg}
    \vcenter{
  \xymatrix{
    &
    X
    \ar[rr]^-s
    \ar[d]_m
    \ar `l[ld] `[ddd]_{g \cdot \inl} [ddd]
    &&
    A
    &
    \\
    &
    HX + \Var X + A
    \ar[rr]_-{Hs + Vs + A}
    \ar[d]_{HX + {f}_X + A}
    &&
    HA + \Var A + A
    \ar[u]^{\labelstyle [a, \isol e A, A]}
    \ar[dd]_{\labelstyle HA +[\isol e A,  A]}
    \\
    &
    HX + HT'X + A
    \ar[rrd]|*+{\labelstyle [Hs, H(\beta \cdot T's)] + A}
    \ar[d]_{\can + A}
    &
    &
    &
    \\
    &
    H(X+T'X) + A
    \ar[rr]_-{H[s, \beta \cdot T's] + A}
    &&
    HA + A
    \ar `r[ru] `[uuu]_{[a, A]} [uuu]
    }}
  \end{equation}
  The left-hand part commutes by the definition of $g$, the right-hand part
  commutes trivially, the upper square commutes since $s$ is a solution of $m$
  and the lower triangle commutes trivially, again. It remains to verify that
  the middle part commutes. We check the commutativity of this part
  componentwise: the left-hand and right-hand components commute trivially.  We do not
  claim that the middle component commutes. However, 
  in order to prove that the overall outside of \refeq{eq:firstg} commutes,
  we need only show
  that this middle component commutes when extended by $[a, A]: HA + A \to
  A$. To see this consider the next diagram:
  $$
  \xymatrix@C+1pc{
    &
    \Var X 
    \ar[d]_{e_X}
    \ar `l[ld] `[dd]_{{f}_X} [dd]
    \ar[r]^-{\Var s}
    &
    \Var A
    \ar[r]^-{\isol e A}
    \ar[d]_{e_A}
    & 
    A
    \\
    &
    T'X 
    \ar[r]_-{T's}
    &
    T'A
    \ar[ru]_{\beta}
    \\
    &
    HT'X 
    \ar[u]^{\tau_X'\cdot \inl_{T'X}}
    \ar[r]_-{HT's}
    &
    HT'A
    \ar[u]_{\tau_A'\cdot\inl_{T'A}}
    \ar[r]_{H\beta}
    &
    HA
    \ar[uu]_a
    }
  $$
  This diagram commutes: the left-hand part commutes since $e$ is a guarded
  rps; the upper and lower squares in the middle commute due to the naturality of $e$
  and of $\tau': (H+\Var)T' \to T'$ and $\inl T': HT' \to (H+\Var)T'$,
  respectively; the upper right-hand triangle commutes since $\isol e A$ is an
  interpreted solution of the rps $e$. Finally, to see that the lower
  right-hand part commutes recall from Diagram~\refeq{diag:42} that
  $$[a, \isol e A] \cdot (H+\Var)\beta = \beta\cdot\tau_A' \,\text{.}$$ 
  Compose this last equation with the coproduct injection
  $\inl_{T'A}: HT'A \to (H+\Var)T'A$ to obtain the desired commutativity. 

  Finally, we verify that the right-hand component of~\refeq{diag:solg}
  commutes. Indeed, consider the diagram below:
  $$
  \xymatrix{
    &
    T'X
    \ar[r]^-{T's}
    \ar[d]_{\ol{e}_X}
    \ar `l[ld] `[dd]_{g \cdot \inr} [dd]
    &
    T'A
    \ar[r]^-{\beta}
    \ar[d]_{{\ol e}_A}
    &
    A
    \\
    &
    HT'X + X
    \ar[r]^-{HT's + s}
    \ar[d]_{[\inl \cdot H\inr, g \cdot \inl]}
    &
    HT'A + A
    \ar[rd]^{H\beta + A}
    \\
    &
    H(X + T'X) + A
    \ar[rr]_-{H[s, \beta \cdot T's] + A}
    &&
    HA + A
    \ar[uu]_{[a, A]}
    }
  $$
  The left-hand part commutes by the definition of $g$; the upper middle square commutes by the
  naturality of $\ol e$, the right-hand part 
  commutes since $\beta$ is the solution of $\ol{e}_A$ in the cia $(A,a)$ (see
  Remark~\ref{rem:techprelims}); and for the lower middle part we consider the components
  separately: the left-hand component clearly commutes by the functoriality of
  $H$, and for the right-hand component observe that it commutes when
  extended by $[a, A]: HA + A \to A$ (see
  Diagram~\refeq{eq:firstg}). Thus, the outside of the  diagram above
  commutes, and this completes the proof.
\end{proof}

Coming back to our setting in Section~\ref{sec:cias}, let $\ell:K(H\times\Id)\to HM$ be an abstract GSOS rule, 
where $M$ is the free monad on $K$  (or, more generally, let $\lambda$ be a distributive law of an arbitrary monad over the cofree copointed functor $H\times\Id$). Assume furthermore that the composite $HM$ is iteratable. By applying the
two Theorems~\ref{thm:sol} and~\ref{thm:rps} and also Theorem~\ref{thm:extcia} to the cia $k: HMC
\to C$ from Theorem~\ref{thm:distcia} we get two more solutions
theorems for free: 

\begin{cor}
  \label{cor:sol}
  Every guarded equation morphism $e: X \to \T{HM}(X+C)$ has a unique
  solution in the cia $(C,k)$.
\end{cor}

\begin{cor}
  \label{cor:rps}
  Every guarded rps $e: \Var \to \T{HM + \Var}$ has a unique interpreted solution
  in the cia $(C, k)$, and this solution extends the cia structure on $C$.
\end{cor}

Assuming that $MHM$ is iteratable two similar theorems hold for the cia $k':MHM C \to C$
obtained from Theorem~\ref{thm:sandwich}:

\begin{cor}
  \label{cor:sol_2}
  Every guarded equation morphism $e: X \to \T{MHM}(X+C)$ has a unique
  solution in the cia $(C,k')$.
\end{cor}

\begin{cor}
  \label{cor:rps_2}
  Every guarded rps $e: \Var \to \T{MHM + \Var}$ has a unique interpreted solution
  in the cia $(C, k')$, and this solution extends the cia structure on $C$.
\end{cor}

\subsection{Summary: Equation Formats}

  \label{rem:eqns}
  (1)~Until now we have seen several formats of equation morphisms. These are related in the sense that some formats comprise others as a special case. In particular, guarded ``sandwiched'' equation morphisms comprise all other formats as shown in the following picture:
$$
\xymatrix@!C=6cm@R+0.5cm{
  \genfrac{}{}{0pt}{0}{\genfrac{}{}{0pt}{0}{\text{guarded equation
        m.}}{X\to T^{HM}(X+C)}}{\text{Corollary~\ref{cor:sol}}}
  \ar@{|->}[r]^{\T{\eta HM}_{X+C} \cdot e}	& \genfrac{}{}{0pt}{0}{\genfrac{}{}{0pt}{0}{\text{guarded equation m.}}{X \to T^{MHM}(X+C)}}{\text{Corollary~\ref{cor:sol_2}}}	\\
\genfrac{}{}{0pt}{0}{\genfrac{}{}{0pt}{0}{\text{flat equation m.}}{X\to HMX+C}}{\text{Theorem~\ref{thm:distcia}}} \ar@{|->}[r]^{(\eta_{HMX}+C) \cdot e} \ar@{|->}[u]|{\can \cdot (\kappa^{HM}_X+\eta^{HM}_C) \cdot e}	& \genfrac{}{}{0pt}{0}{\genfrac{}{}{0pt}{0}{\text{flat equation m.}}{X\to MHMX+C}}{\text{Theorem~\ref{thm:sandwich}}} \ar@{|->}[u]|{\can \cdot (\kappa^{MHM}_X+\eta^{MHM}_C) \cdot e}	\\
\genfrac{}{}{0pt}{0}{\genfrac{}{}{0pt}{0}{\text{$\ell$-equation}}{X\to HMX}}{\text{Theorem~\ref{thm:ellsolution}}} \ar@{|->}[r]^{\eta_{HMX} \cdot e} \ar@{|->}[u]|{\;\inl \cdot e}	& \genfrac{}{}{0pt}{0}{\genfrac{}{}{0pt}{0}{\text{sandwiched~$\ell$-equation}}{X\to MHMX}}{\text{Theorem~\ref{thm:sandwich_ellsolution}}} \ar@{|->}[u]|{\;\inl \cdot e}
}
$$
The arrows point to more general formats, and their labels indicate
how one forms an equation morphism of the more general format from a
given equation morphism $e$ of the simpler format. Natural
transformations without a superscript refer to the free monad $M$, the ones
with a superscript refer to $T$-monads. The monad morphism
$\T{\eta HM}: \T{HM} \to \T{MHM}$ arises from the natural
transformation $\eta HM: HM \to MHM$ as explained at the beginning of
this section (cf.~\refeq{eq:Tmor}). In all cases one readily proves
that the solutions are preserved along the arrows. For example,
solutions of $e: X \to HMX + C$ in the cia $(C, k)$ are in
one-to-one correspondence with solutions of $(\eta_{HFX} + C) \cdot e: X \to MHMX
+ C$ in the cia $(C, k')$, etc. 

\medskip\noindent
(2)~Similarly, guarded rps's in Corollary~\ref{cor:rps_2} subsume
those in Corollary~\ref{cor:rps} which in turn subsume the rps of the
form $V \to \T{H+V}$ from Theorem~\ref{thm:rps}. Pictorially, we have
\[
\xymatrix@!C=4cm{
\genfrac{}{}{0pt}{0}{\text{(guarded) rps}}{V \to \T{H + V}}
\ar@{|->}[r]^-{\T{H\eta + V} \cdot e}
&
\genfrac{}{}{0pt}{0}{\text{(guarded) rps}}{V \to \T{HM + V}}
\ar@{|->}[r]^-{\T{\eta HM + V} \cdot e}
&
\genfrac{}{}{0pt}{0}{\text{(guarded) rps}}{V \to \T{MHM + V}}
}
\]
Again, one readily proves that solutions are preserved along the
arrows. For example, solutions of $e: V \to \T{H+V}$ in the cia $(C,
c^{-1})$ are in one-to-one correspondence with solutions of $T^{H\eta + V} \cdot
e$ in the cia $(C, k)$. 

\medskip\noindent
(3)~Finally, guarded rps's subsume guarded equation morphisms. To see
this let $e: X \to \T H (X+C)$ be a guarded equation
morphism. Consider the constant functors $\conf X$ and $\conf C$ and
let $H' = H + \conf C$. Then
$e$ corresponds precisely to a natural transformation
$\ol e: \conf X \to \T H (\conf X + \conf C),$
and this gives rise to an rps as follows
\[
\xymatrix@1{
  \conf X \ar[r]^-{\ol e}
  &
  \T H (\conf X + \conf C)
  \ar[r]^-{\T H \kappa^{\conf X + \conf C}}
  &
  \T H \T{\conf X + \conf C}
  \ar[d]^-{\T{\inl} \T{[\inr,\inm]}}
  \\
  &&
  \T{H+\conf C + \conf X}\T{H+\conf C + \conf X}
  \ar[rr]^-{\mu^{H + \conf C + \conf X}}
  &&
  \T{H'+\conf X}.
}
\]
It is straightforward but rather tedious to check that this rps is
guarded and that its solutions are in one-to-one correspondence with
the solutions of $e$.

\section{Recursive Function Definitions over the Behavior}
\label{sec:lambdarps}

Even with 
all the results we have seen so far, we are still not able to obtain functions such as
the shuffle product $\otimes$ on streams (see~\refeq{eq:shuffle_product}) as a unique solution since
its definition refers to the behavior of the arguments of the function.
Notice also that the specification of $\otimes$ makes use of the stream addition $+$ operation, 
so this operation is assumed as given or previously specified and the
specification of $\otimes$ is built on top of the specification of
$+$. Our aim in this section is to prove results that yield unique solutions of such
specifications in terminal coalgebras.

\begin{notation}   
  From now on we shall also need to consider free monads of other
  functors than $K$ from Asumption~\ref{ass:1}. We follow the
  convention that whenever we write $\ext F$ for a functor $F$ we
  assume that a free monad $\ext F$ exists and is given objectwise by
  free algebras for $F$ (cf.~Remark~\ref{rem:free_algebras_monad}(2)).
 
  Finally, to shorten notation, we usually abbreviate $K + V$ by $F$.
\end{notation}

\begin{plan}\label{plan:lrps}
Here is the basic plan for the results in this section.
We begin, as before,  with a endofunctor $H$
and its terminal coalgebra $(C,c)$.   We also have a
separate functor $K$, and an 
abstract GSOS rule
$\ell:K(H\times\Id)\to H\ext K$ specifying a set of ``given'' operations.
When we say ``specify'' here, we refer to the algebra structure
 $b: KC\to C$, the unique morphism so that
$$c\cdot b = H\extalg{b}\cdot \ell_C \cdot K\pair{c,\id_C} $$
(see Theorem~\ref{thm:interp_gsos_rule}).
We also have another functor $V$ corresponding to 
the ``new'' operation symbols that we wish to interpret.
What our approach requires at this point is 
a natural transformation $$e: V(H\times \Id)
\to H \ext{K+V}.$$  We shall call $e$ a \emph{recursive program scheme
  w.\,r.\,t.~$\ell$} (or, an \emph{$\ell$-rps}, for short).
  From $e$ and $\ell$ we shall obtain a natural transformation
  $$n: (K+V)(H\times\Id) \to H\ext{K+V}
  $$
  in a canonical way.
Again,  $n$  is  an abstract GSOS rule---but notice
that the functor involved is $K+V$, not $K$ as it is for $\ell$.
This natural transformation $n$ has an interpretation in $C$,
call it
$$
a: (K+V)C\to C.
$$
Then it will turn out that $a\cdot \inl = b$, so that the algebra
structure $a$ is an extension of $b$.
  The interpretation of the new operation symbols
  in $C$ will correspond to the  $V$-algebra $(C, a\cdot \inr)$. We will
  prove in Theorem~\ref{thm:ellrps} that this algebra is uniquely determined as
  a \emph{solution} of the $\ell$-rps $e$. 
  
  The upshot is that we start and end with the same kind of data,
  but with a different functor.
  We start with $K$, $\ell$, and $b$,
  and we end with $K+V$, $n$, and $a$.
  The point we are trying to make here is that the situation
  repeats, and so we can apply the result successively.
 \end{plan}
  
 The rest of this section works out the details of this outline.  To
 see that this approach actually accounts for a large number of
 interesting operations on terminal coalgebras, we present many
 examples in Section~\ref{sec:app}.  Prior to this, we have two other
 contributions which extend the basic result in the same ways as the
 results we saw in Section~\ref{sec:bialg}.  We have a ``sandwiched''
 version of Outline~\ref{plan:lrps}, and we also have results about
 cia structures. And in our main result,
 Theorem~\ref{thm:sandwich_rps}, we prove that every sandwiched
 $\ell$-rps (see Definition~\ref{dfn:sandwiched_ell-rps} has a unique
 solution in $C$ which extends the cia structure for $\ext K H \ext K$
 on $C$ given by Theorem~\ref{thm:sandwich}.

 \takeout{
In this section we start with an abstract GSOS rule
$\ell:K(H\times\Id)\to H\ext K$ specifying all given operations (such
as stream addition) whose type is modeled by the functor
$K$. First, we introduce a special form of abstract operational rule $e: V(H\times \Id)
\to H \ext{K+V}$ called a \emph{recursive program scheme
  w.\,r.\,t.~$\ell$} (or, \emph{$\ell$-rps}, for short) recursively specifying new
operations (such as the shuffle product of streams) whose type is
given by $V$.  We prove in
Theorem~\ref{thm:ellrps} that every
$\ell$-rps has a unique solution in the terminal coalgebra $C$, and
this solution extends the cia structure for $H\ext K$ on $C$ given by
Theorem~\ref{thm:distcia}.  This is a modularity result similar to the
one given in Theorem~\ref{thm:extcia} for ordinary rps's.
}



\begin{ass}  We continue to work under Assumption~\ref{ass:1}, and in addition 
 we fix an abstract GSOS rule $\ell: K(H\times\Id) \to H\ext K$ and an endofunctor $V: \A\to\A$.
 \end{ass}

\begin{notation}
  \label{not:coinj}
  (1)~We overload the notation from
  Remark~\ref{rem:free_algebras_monad}(3) and write, for any functor $F$ on $\A$, $\varphi:F\ext F\to \ext F$ and $\eta:\Id\to \ext F$ for (the natural transformations given by) the structures and universal morphisms of the free $F$-algebras, as well as
  $\mu:\ext F\ext F\to \ext F$ and $\kappa:F\to \ext F$ for the multiplication and universal natural transformation of the free monad $\ext F$.

  \medskip\noindent
  (2)~Let $F$ and $G$ be endofunctors of $\A$. The coproduct
  injections $\inl: F \to F + G \leftarrow G: \inr$ lift to monad
  morphisms on the corresponding free monads, and we denote those
  monad morphisms by $\ext{\inl}: \ext{F} \to \ext{F+G} \leftarrow
  \ext{G}: \ext{\inr}$.
\end{notation}

\begin{rem}
  \label{rem:extinl}
  (1)~Notice that the monad morphism $\ext\inl: \ext F \to \ext{F +G}$ is
  uniquely determined by the commutativity of the following square of
  natural transformations:
  $$
  \xymatrix{
    F \ar[r]^-\kappa
    \ar[d]_\inl
    &
    \ext F 
    \ar[d]^{\ext\inl}
    \\
    F + G
    \ar[r]_-{\kappa}
    &
    \ext{F + G}
    }
  $$
  Similarly for $\ext\inr$.
  
  \medskip\noindent
  (2)~Recall from Notation~\ref{not:extalg} that for every $F$-algebra
  $(A,a)$ we have the corresponding Eilenberg-Moore algebra $\extalg a:
  \ext F A \to A$ and that $a = \extalg a \cdot \kappa_A$.
  \takeout{
  Notice that we have the following commutative triangle:
  \begin{equation}
    \label{diag:kappa}
    \vcenter{
      \xymatrix{
        FA 
        \ar[r]^-{\kappa_A}
        \ar[rd]_{a}
        &
        \ext F A
        \ar[d]^-{\extalg a}
        \\
        & 
        A
      }
    }
  \end{equation}
  }
  Moreover, the category of $F$-algebras is isomorphic to the 
  category of  Eilenberg-Moore
  algebras for $\ext F$.
More precisely, $a \mapsto
  \extalg a$ and precomposition with $\kappa_A$ extend to mutually
  inverse functors.

  \medskip\noindent
  (3)~Combining parts (1) and (2) of this remark, we see that for
  every algebra $a:(F+G)A\to A$ the equation
  $\extalg a \cdot \ext\inl_A = \extalg{a\cdot\inl_A}$ holds. Indeed,
  both sides are equal when precomposed with $\kappa_A$:
  $$\extalg a \cdot \ext\inl_A \cdot \kappa_A = \extalg a \cdot \kappa_A \cdot \inl_A = a \cdot \inl_A = \extalg{a\cdot\inl_A} \cdot \kappa_A \,\text{.}$$
  If we make the coproduct algebra structure explicit as in $a=[a_0,a_1]$, we obtain
  $$\extalg{[a_0,a_1]} \cdot \ext\inl_A = \extalg{a_0} \quad\qquad\text{and}\qquad\quad \extalg{[a_0,a_1]} \cdot \ext\inr_A = \extalg{a_1} \,\text{.}$$
\end{rem}

\takeout{ 
\begin{rem}
  \label{rem:lambda}
  Recall from Remark~\ref{rem:liftdist}(4) that the distributive law $\lambda: \extK H
  \to H\extK$ is obtained componentwise as the unique homomorphism
  of $K$-algebras making the following diagram commutative:
  $$
  \xymatrix{
    K \extK HX \ar[rr]^-{\varphi_{HX}} \ar[d]_{K\lambda_X} && \extK H X \ar[d]^{\lambda_X} & HX \ar[l]_-{u_{HX}}
    \ar[ld]^{Hu_X} \\
    K H \ext K X \ar[r]_-{\ell_{\extK X}} & H\ext K\extK X
    \ar[r]_-{H\mu_X} & H \extK X
    }
  $$
\end{rem}
}

\begin{defi}
  A \emph{recursive program scheme w.\,r.\,t.~$\ell$} (shortly,
  \emph{$\ell$-rps}) is a natural transformation
  $$e: \Var (H\times\Id) \to H\ext{F} \,\text{,}$$
  where (throughout this section) $F = K + V$.
\end{defi}

\begin{construction}
  \label{constr:s}  
  Let $e: \Var (H\times \Id) \to H \ext{F}$ be an $\ell$-rps. 
  This gives an abstract GSOS rule $n: F (H\times \Id) \to H\ext{F}$ defined on its coproduct components as displayed below:
  $$
  \xymatrix{
    \Var (H\times\Id)
    \ar[rd]^-e 
    \ar[d]_{\inr (H\times\Id)}
    \\
    F (H\times\Id)
    \ar[r]^-n
    &
    H \ext{F}
    \\
    K(H\times\Id)
    \ar[r]_-\ell
    \ar[u]^{\inl (H\times\Id)}
    &
    H\ext{K}
    \ar[u]_{H\ext{\inl}}
    }
  $$
  We write $a: FC \to C$ for the $n$-interpretation in $C$ (cf.~Definition~\ref{dfn:interp}). 
 \end{construction}
 
 \begin{rem}
  The  construction of the abstract GSOS rule $n$ above  is reminiscent
  of the composition of two distributive laws $\lambda: MD \to DM$ and
  $\lambda':M'D \to DM'$ of the monads $M$ and $M'$ over the same
  copointed functor $D$ using the coproduct of monads, see~\cite{pow03,lpw04}. For
  free monads $M = \ext G$ and $M' = \ext{G'}$ and for $D = H \times
  \Id$, $\lambda$ and $\lambda'$ are equivalently presented by
  abstract GSOS rules $\ell: G(H\times \Id) \to H\ext G$ and $\ell':
  G'(H\times \Id) \to H \ext{G'}$ and the construction of
  loc.~cit.~amounts to forming $\wt n$ as shown in the diagram below:
  \[
  \xymatrix{
    G(H\times \Id)
    \ar[r]^\ell
    \ar[d]_{\inl(H\times\Id)}
    &
    H\ext G
    \ar[d]^{H\ext\inl}
    \\
    (G+G')(H\times \Id)
    \ar[r]^-{\wt n}
    &
    H\ext{G+G'}
    \\
    G'(H\times \Id)
    \ar[r]_-{\ell'}
    \ar[u]^{\inr(H\times \Id)}
    &
    H\ext{G'}
    \ar[u]_{H\ext\inr}
    }
  \]
  Notice that there is no interplay between $\ell$ and $\ell'$. So, as
  explained in loc.~cit., in the case of operational rules of process
  combinators the formation of $\wt n$ corresponds
  precisely to forming the disjoint union of two transition system
  specifications with mutually independent operations. In contrast, in
  Construction~\ref{constr:s} the left-hand component $e$
  depends on $K$. So in the case of operational rules for process
  combinators our construction corresponds to combining a given
  transition specification (modeled by $\ell$) with another one
  (modeled by $e$) which makes use of the operators specified by the
  first specification.
\end{rem}

Our next result substantiates the details presented in Outline~\ref{plan:lrps}.
 
\begin{prop}\label{substant}
  Let $b: KC \to C$ be the interpretation of $\ell: K(H\times\Id) \to H\ext K$. Then for $a: FC \to C$
  from Construction~\ref{constr:s} we have
  $$
  b = (\xymatrix@1{
    KC \ar[r]^-{\inl_C}
    &
    FC
    \ar[r]^-{a}
    &
    C
  })\,\text{.}
  $$
\end{prop}

\begin{proof}
  Consider the diagram below:
  $$
  \xymatrix@C+1pc{
    KC 
    \ar[d]_{\inl_C}
    \ar[r]^-{K\langle c,\id_C\rangle}
    &
    K(H\times\Id)C
    \ar[r]^-{\ell_C}
    \ar[d]_{\inl_{(H\times\Id)C}}
    &
    H\ext K C
    \ar[d]^{H\ext{\inl}_C}
    &
    \\
    FC
    \ar[r]^-{F\langle c,\id_C\rangle}
    \ar[d]_{a}
    &
    F(H\times\Id)C 
    \ar[r]^-{n_C}
    &
    H\ext F C
    \ar[d]^{H\extalg{a}}
    &
    \\
    C
    \ar[rr]^-c 
    &&
    HC
    \ar@{<-} `r[ru] `[uu]_{H(\extalg{a \cdot \inl_C})} [uu]
    &
  }
  $$
  The lower square commutes since $a$ is the $n$-interpretation in $C$ (cf.~Theorem~\ref{thm:interp_gsos_rule}) and the upper
  right-hand one by the definition of $n$
  (cf.~Construction~\ref{constr:s}). The upper left-hand square commutes
  by the naturality of $\inl:K\to F$, and for the right-hand part we remove $H$ and use Remark \ref{rem:extinl}(3). Thus the outside commutes.

  Now recall that $b$ is uniquely determined by the commutativity of
  the diagram in Theorem~\ref{thm:interp_gsos_rule}. Thus, $a \cdot \inl_C = b$ holds, as
  desired. 
\end{proof}

\begin{summary}\label{sum:1}
Let us summarize and review our work in 
this section so far. We say that the tuple
$$
(c: C\to HC, \quad \ell: K(H\times\Id) \to H\ext{K}, \quad b: KC \to C)
$$
is \emph{appropriate} if $c$ is a terminal coalgebra structure, 
$\ell$ is an abstract GSOS rule, and $b$ is its interpretation.
Given an appropriate tuple, the work in this section shows how to obtain 
another appropriate tuple: let $V$ be a functor, let $F = K + V$,
and let 
$e:V(H\times \Id)\to H\ext{F}$ be an $\ell$-rps.    We 
consider
$$n = [H\ext{\inl} \cdot \ell, e]:  F(H\times \Id) \to H\ext{F} \  .$$
This is an abstract GSOS rule involving $F$, so it has an interpretation $a: FC \to C$.
So 
\begin{equation}
  \label{eq:tuple}
  (c: C\to HC, \quad n: F(H\times\Id) \to H\ext{F}, \quad a: FC \to C)
\end{equation}
is again appropriate.   It extends the earlier appropriate tuple in the sense that 
$a\cdot \inr = b: VC \to C$. 
\end{summary}
\begin{exa}\label{expl}  We continue our exploration of
stream operations as defined by behavioral differential equations. 
We want to study
 the shuffle product on streams mentioned in  the introduction.  It is specified by
\begin{equation} \label{eq:shuffle_product_again}
(\sigma \otimes \tau)_0 = \sigma_0\cdot\tau_0 \qquad\qquad
(\sigma \otimes \tau)' = \sigma \otimes \tau' + \sigma'\otimes \tau \,\text{.}
\end{equation}
The behavior functor $H:\Set\to\Set$ here is $HX = \Real \times X$.
On the right side of \refeq{eq:shuffle_product_again}, we see the stream 
addition operation $+$. This is a binary operation and so corresponds to 
a $K$-algebra structure on $C$, where $KX = X\times X$.
Our work in this section shows how to define $\otimes$ ``on top of''
$(C,+)$. First observe that the given operation $+: KC \to C$ is
obtained as the interpretation of the abstract GSOS rule $\ell: K(H
\times \Id) \to H\ext K$ given by~\refeq{eq:stream_addition} in
Example~\ref{ex:dist}(2). Now we take $VX = X \times X$, the type functor of the binary
operation $\otimes$. For $F = K + V$ we may then identify $\ext{F}X$
with the set of all terms obtained by applying the operation symbols
$\otimes$ and $+$ to variables in $X$. %
%
%
(Note that $\otimes$ and $+$
are regarded as uninterpreted symbols at this point, and for the given
operation $+$ we have the interpretation on $C$ induced by $\ell$.)
To obtain an interpretation of $\otimes$ on $C$, we  use the $\ell$-rps $e$ whose components
$$
e_X:  (\Real \times X\times X)\times (\Real \times X\times X) \to \Real\times \ext{F}X
$$
are given by 
\[
e_X((r,x,x'),(s,y,y')) = (r\cdot s,(x\otimes y')+(x'\otimes y)) \,\text{.}
\]
The abstract GSOS rule $n:F(H \times \Id) \to H\ext F$ from
Construction~\ref{constr:s} now induces the algebra structure $a: FC
\to C$ whose left-hand component is, by Proposition~\ref{substant}, $a \cdot \inl = +: K C \to C$
and whose right-hand component is easily seen to be the desired
shuffle product $a \cdot\inr = \otimes: VC \to C$. So
Proposition~\ref{substant} is a kind of sanity check: the operation
induced by $n$ for the left-hand component $K$ of $F = K + V$ is the
same operation $+$ we started with as given.
\takeout{ 
Then the unique solution of $e$ in $C$ (cf.~Theorem~\ref{thm:interp_gsos_rule})
gives us $b: FC\to C$ as in \refeq{thm:interp_gsos_rule}.
Define $+: C\times C \to C$ by $\sigma + \tau = F(\sigma +\tau)$;
on the right we have the formal symbol $+$.  Define $\otimes$ similarly.
Then  \refeq{thm:interp_gsos_rule} is just the statement that for all streams
$\sigma$ and $\tau$,
$$\begin{array}{lcl} x \otimes y & = & 
(\sigma_0 \cdot \tau_0,  \sigma \otimes \tau' + \sigma'\otimes \tau)\ .
\end{array}
$$
This matches what we want in \refeq{eq:shuffle_product_again}.
}
\end{exa}

\begin{exa} \label{ex:sandlrps}
We present another 
example related to the zipping of streams; see Examples~\ref{ex:dist}(2)
and~\ref{ex:moredist}.   We have seen that there is an abstract GSOS
rule $\ell: K(H\times \Id) \to H\ext K$ with $KX = X \times X$ presenting
the $\zip$ operation on streams.  
Using an $\ell$-rps, we can show that there is a unique function $f: C\to C$ satisfying 
\[
f(\sigma) = \zip(\sigma, f(\sigma))\, .
\]
We might mention a subtlety: consider the alternative definition
\[
g(\sigma) = \zip(g(\sigma), \sigma)\, .
\]
The uniqueness assertion  fails for this.  The point is that each $(g(\sigma))_0$ is arbitrary.
As for $f$,  note that $(f(\sigma))_0 = \sigma_0$ for all streams $\sigma$.  In fact, this 
observation is central to obtaining $f$ via Theorem~\ref{thm:ellrps}.
The point is that we can write
\[
f(\sigma) = \zip(\sigma_0.\sigma', f(\sigma))
  =  \sigma_0.\zip(f(\sigma),  \sigma')\, .
\]
The last formulation makes it clear that we can obtain $f$ using the
method of this section: one takes $VX = X$, and then the last equation
gives rise to 
\[
e: V(H \times \Id) \to H\ext{F} 
\quad\text{with}\quad
e_X: (r, x', x) \mapsto (r, \zip (f(x), x')).
\]
Incidentally, functions like $f$ play a role in the theory of
paperfolding sequences; cf.~\cite{dmfp82}, Observation 1.3.
\end{exa}

\subsection{Interpreted solutions and cia structures}
\label{sec-iscias}

At this point, the reader might wish to revisit Outline~\ref{plan:lrps}.
The work we have done so far builds a $V$-algebra structure on $C$
by taking an $F$-algebra structure and (in effect) throwing away the 
interpretation of the ``givens'', since we know them to be the same
as in the original $K$-algebra structure.    Our next result, Theorem~\ref{thm:ellrps},
provides a different way to look at things.  It provides a direct notion of 
an \emph{interpreted solution}; this is a $V$-algebra.  The theorem shows that
$\ell$-rps's have unique interpreted solutions.    
In addition, Theorem~\ref{thm:ellrps} shows that cia structures are propagated 
according to our development.

\begin{defi}
Let $e: \Var (H\times\Id) \to H\ext{F}$ be an $\ell$-rps.
  An \emph{interpreted solution} of $e$ in $C$ is a $V$-algebra structure
  $s: \Var C \to C$ such that the triangle below commutes:
  \begin{equation}
    \label{eq:lambdainterp}
    \vcenter{
      \xymatrix@C+1pc{
        \Var C
        \ar[rr]^-s
        \ar[d]_{V\langle x, \id_C\rangle}
        &&
        C
        \\
        \Var(HC \times C)
        \ar[d]_{e_C}
        &
        HC
        \ar[ru]_-{c^{-1}}
        \\
        H \ext F C
        \ar[ru]_-{H\ext{[b,s]}}
      }}
    \end{equation}
    \takeout{
  \begin{equation}
    \label{eq:lambdainterp}
    s = (
    \xymatrix@C+1pc{
      \Var C
      \ar[r]^-{\Var \langle c,\id_C\rangle}
      &
      \Var (HC\times C)
      \ar[r]^-{e_C}
      &
      H\ext{F}C
      \ar[r]^-{H \extalg{[b,s]}}
      &
      HC
      \ar[r]^-{c^{-1}}
      &
      C
    }
    ) \,\text{,}
  \end{equation}} 
  where $b:KC\to C$ is the $\ell$-interpretation in $C$ (cf.~Definition~\ref{dfn:interp}).
\end{defi}

\begin{thm}
  \label{thm:ellrps}
  For every $\ell$-rps there exists a unique interpreted solution $s$ in $C$. In
  addition, $s$ extends the cia structure on $C$, i.\,e., the following
  is the structure of a cia for $H\ext{F}$ on $C$:
  \begin{equation}
    \label{eq:extcia}
  \xymatrix@1{
    H\ext{F} C
    \ar[rr]^-{H\extalg{[b,s]}}
    &&
    HC
    \ar[r]^-{c^{-1}}
    &
    C \,\text{.}
    }
  \end{equation}
\end{thm}
\begin{proof} 
  Given the abstract GSOS rule $\ell$ and the $\ell$-rps $e$, form $n$
  as in Construction~\ref{constr:s} and define 
  $$
  s = (\xymatrix@1{
    VC
    \ar[r]^-\inr
    &
    FC 
    \ar[r]^-{a}
    &
    C
  })\,\text{,}
  $$
  where $a$ is the $n$-interpretation in $C$. Observe immediately
  that, by Proposition~\ref{substant},
  \[
  a = [b,s]:FC\to C,
  \]
  where $b: KC \to C$ is the
  $\ell$-interpretation in $C$, and so we have
\begin{equation}\label{cor:a}
  \extalg a = \extalg{[b,s]}: \ext F C \to C \,\text{.}
  \end{equation}
   
  \medskip\noindent
  (1)~We prove that $s$ is a solution
  of $e$ in $C$. Indeed, consider the commutative diagram
  \begin{equation}
    \label{diag:sols}
    \vcenter{
    \xymatrix@!C=2cm{
      VC
      \ar[d]^{\inr_C}
      \ar[r]^-{V\langle c,\id_C\rangle}
      &
      V(H\times\Id)C
      \ar[rd]^-{e_C}
      \ar[d]_{\inr_{(H\times\Id)C}}
      \\
      FC
      \ar[d]^{a}
      \ar[r]^-{F\langle c,\id_C\rangle}
      &
      F(H\times\Id)C
      \ar[r]^-{n_C}
      &
      H\ext F C
      \ar[d]^{H\extalg{a}}
      \\
      C
      \ar[rr]^-c 
      \ar@{<-} `l[u] `[uu]^s [uu] 
      && 
      HC\,\text{.}
    }}
  \end{equation}
  The lower square commutes since $a$ is the $n$-interpretation in $C$, and the upper right-hand triangle by the definition of $n$ (cf.~Construction~\ref{constr:s}). The upper left-hand square commtes by the naturality of $\inl:V\to F$, and the left-hand part is the definition of $s$. Thus the outside commutes, and we see that $s$ is a solution of $e$ since
  $\extalg{a} = \extalg{[b,s]}$ holds by \refeq{cor:a}. 

  \medskip\noindent
  (2)~We now prove that $s$ is unique. Suppose that $t$ is any solution
  of $e$. We will prove that 
  \begin{equation}
    \label{eq:t}
    [b,t] = a,
  \end{equation}
  which implies the desired equation $s = t$. 

  In order to prove~\refeq{eq:t} we have to verify the commutativity
  of the following diagram (cf.~Theorem~\ref{thm:interp_gsos_rule}):
  $$
  \xymatrix@C+1pc{
    FC 
    \ar[r]^-{F\langle c,\id_C\rangle}
    \ar[d]_{[b,t]}
    &
    F(H\times\Id)C
    \ar[r]^-{n_C}
    &
    H\ext F C
    \ar[d]^{H\extalg{[b,t]}}
    \\
    C
    \ar[rr]^-c
    &&
    HC
  }
  $$
  We verify this for the two coproduct components of $FC = KC + VC$
  separately.

  For the right-hand component we obtain   
  diagram~\refeq{diag:sols} above with $s$ replaced by $t$ and $a$ by
  $[b,t]$, which commutes since $t$ is a solution of $e$. For the
  left-hand component we obtain the diagram below:
  $$
  \xymatrix@!C=2.9cm{
    KC 
    \ar[r]^-{K\langle c,\id_C\rangle}
    \ar[dd]_{b}
    &
    K(H\times\Id)C
    \ar[r]^{\inl_{(H\times\Id)C}}
    \ar[dr]_{\ell_C}
    &
    F(H\times\Id)C
    \ar[r]^-{n_C}
    &
    H\ext F C
    \ar[dd]^{H\extalg{[b,t]}}
    \\
    &&
    H\ext{K}C
    \ar[ur]_{H\ext{\inl}_C}
    \ar[dr]_{H\extalg{b}}
    &
    \\
    C
    \ar[rrr]^-c
    &&&
    HC
  }
  $$
  The big left-hand part commutes since $b$ is the $\ell$-interpretation in $C$ (cf.~Theorem~\ref{thm:interp_gsos_rule}),
  the upper triangle commutes
  by the definition of $n$ (see Construction~\ref{constr:s}), and for
  the right-hand triangle remove $H$ and notice that
  $$
  \extalg{[b,t]} \cdot \ext{\inl}_C = \extalg{b}
  $$
  by Remark \ref{rem:extinl}(3).

  \medskip\noindent
  (3)~To complete the proof we will show that $c^{-1} \cdot H\extalg{[b,s]} : H\ext FC
  \to C$ is the structure of a cia for $H\ext F$. But this is a consequence
  of Theorem~\ref{thm:distcia}; indeed, recall that $[b,s]$
  is the interpretation of the abstract GSOS rule $n: F(H\times\Id) \to H
  \ext F$ in $C$, see~\refeq{cor:a}. 
\end{proof}

\begin{rem}
\label{rem:comprps}
Notice that the fact that the unique solution $s$ of an $\ell$-rps extends
the cia structure on $C$ means that the operations on $C$ defined in
this way may be part of recursive definitions according to the
Corollaries~\ref{cor:sol} and~\ref{cor:rps} (where $M = \ext{K +
  \Var}$).
\takeout{ 
  \medskip\noindent

(2)~In addition, we have a \emph{modularity principle} for
solutions of $\ell$-rps's---the operations provided by the unique solution $s:\Var C
\to C$ may occur as givens in subsequent $\ell$-rps's (and we will make use of this feature in our
applications in Section~\ref{sec:app}). More precisely, as explained
in Outline~\ref{plan:lrps} of this section the functors $K$ and $V$ in a given
$\ell$-rps $e: V(H\times \Id) \to H\ext{K+V}$ provide the type of given
and newly specified operations. Now $e$ gives rise to the abstract GSOS rule
$n: F(H \times \Id) \to H\ext F$, where $F = K + \Var$, with the
interpretation $[b,s]: FC \to C$. So every $n$-rps $f: W(H\times \Id) \to
H\ext{F+W}$ specifying new operations of type $W$ and using operations
of type $F = K +V$ as givens has a unique solution $t: WC \to C$.}%
\end{rem}

As we shall see in Section~\ref{sec:app}, in conrete instances many
operations are definable as unique solutions of $\ell$-rps's. But
there are operations that cannot be defined by any $\ell$-rps (or
abstract GSOS rule):
\begin{exa}
  \label{ex:tl}
  The tail function $\tl: \Real^\omega \to \Real^\omega$ on
  streams cannot be defined as (part of) an interpretation of any abstract GSOS
  rule $\ell: K(\Real \times \Id \times \Id) \to \Real \times \ext
  K$. For if it were possible to obtain the tail function in this way
  we would have the induced cia $c^{-1}\cdot H\ext b: H\ext K C \to C$
  according to Theorem~\ref{thm:distcia}. Then, for $X = \{x\}$ the
  $\ell$-equation (see Definition~\ref{dfn:elleq})
  $e: X \to \Real \times \ext K X$ given by $e(x) = (r,
  \tl(x))$ would have a unique solution $e^\dag: X \to C$ for every $r
  \in \Real$. But this is clearly not the case: every stream $\sigma$
  with $r$ at its head yields a solution $e$, since  $e^\dag(\sigma) = \sigma$. 
\end{exa}

The next proposition shows that for two $\ell$-rps's that do not
interact the order in which the algebra $b: KC \to C$ is extended by their
solutions does not matter. 

\begin{prop}
  \label{prop:order}
  Let $e_i: \Var_i (H\times\Id) \to H \ext{K + \Var_i}$, $i = 1, 2$, be two
  $\ell$-rps's. Then the cia structure on $C$ extended by the unique
  solutions $s_i: \Var_i C \to C$ of the $e_i$ is independent of the
  order of extension.
\end{prop}

\begin{rem}
More precisely, we may first take $s_1: \Var_1 C \to C$ to obtain an
extended cia structure as in~\refeq{eq:extcia}, and then take the
solution of $s_2: \Var_2 C \to C$ in the new cia, or vice versa. Either way, the
resulting extended cia structure is
\begin{equation} \label{eq:order}
\vcenter{
\xymatrix@1{
  H\ext{(K+\Var_1+\Var_2)} C \ar[rr]^-{H\extalg{[b, s_1, s_2]}} &&
  HC 
  \ar[r]^-{c^{-1}}
  &
  C\,\text{.}
}
}
\end{equation}
\end{rem}

\begin{proof}[Proof of Proposition~\ref{prop:order}]
It is sufficient to prove that the cia structure on $C$ obtained by
extending $k:H\ext KC \to C$ first by $s_1$ and then by $s_2$ is~\refeq{eq:order}.

So take $s_1$ and extend the cia structure $(C, k)$ to obtain the cia
\[c^{-1} \cdot H\extalg{[b,s_1]}: H\ext{K + \Var_1} (C) \to C\]
(cf.~\refeq{eq:extcia}). Recall from the proof of
Theorem~\ref{thm:ellrps} that this cia structure 
is obtained as follows: one first forms the abstract GSOS rule 
$$
n = [H \ext\inl \cdot \ell, e_1]: 
(K+\Var_1)(H\times\Id) \to H \extalg{K + \Var_1}
$$
whose interpretation is $b' = [b, s_1]: (K+\Var_1) C \to C$,
cf.~\refeq{cor:a}, and then one applies
Theorem~\ref{thm:distcia}.

Now we form the following $n$-rps
$$
\xymatrix@1{
  \Var_2 (H\times\Id) 
  \ar[r]^-{e_2}
  &
  H\ext{K+\Var_2}
  \ar[rr]^-{H\ext{[\inl,\inr]}}
  &&
  H\ext{(K+\Var_1+\Var_2)}\,\text{.}
}
$$
Its unique solution is easily seen to be $s_2$; indeed, consider the
following diagram (and notice that the right-hand arrow is $H\extalg{[b', s_2]}$): 
$$
\xymatrix@C+1pc{
  \Var_2 C
  \ar[r]^-{\Var_2 \langle c,\id_C\rangle}
  \ar[d]_{s_2}
  &
  \Var_2 (H\times\Id) C 
  \ar[r]^-{{(e_2)}_C}
  &
  H \ext{K+\Var_2} C
  \ar[rd]_{H\extalg{[b,s_2]}}
  \ar[r]^-{H\ext{[\inl,\inr]}_C}
  &
  H\ext{(K+\Var_1+\Var_2)}C
  \ar[d]^{H\extalg{[b,s_1,s_2]}}
  \\ 
  C
  &&&
  HC
  \ar[lll]^-{c^{-1}}
  }
$$
The left-hand part commutes since $s_2$ is the unique solution of
$e_2$ and for the right-hand part we remove $H$ and then precompose
with $\kappa_C$ to obtain
$$
\begin{array}{rcl@{\qquad}p{5cm}}
  \ext{[b,s_1,s_2]}\cdot\ext{[\inl,\inr]}_C \cdot \kappa_C & = & \ext{[b,
    s_1,s_2]} \cdot \kappa_C \cdot [\inl,\inr]_C & cf.~Remark~\ref{rem:extinl}(1) \\
  & = & [b,s_1,s_2] \cdot [\inl, \inr]_C & see Remark~\ref{rem:extinl}(2)
  \\
  & = & [b, s_2] \\
  & = & \ext{[b,s_2]} \cdot \kappa_C & see Remark~\ref{rem:extinl}(2)
\end{array}
$$
The desired equality now follows since precomposition with $\kappa$ yields
an isomorphism of categories, see Remark~\ref{rem:extinl}(2).
\end{proof}

\begin{rem} \label{rem:inverse_of_c}
Notice that we can always consider the algebraic operation provided by $c^{-1}:HC\to C$ as a given operation in any $\ell$-rps for an abstract GSOS rule $\ell:K(H\times\Id)\to H\ext K$. More precisely, we can assume that $K=K'+H$ and that the $\ell$-interpre\-ta\-tion $b:KC\to C$ has the form $b=[b',c^{-1}]$. Indeed, given any abstract GSOS rule $\ell':K'(H\times\Id)\to H\ext{K'}$ with $\ell'$-interpretation $b'$, we define the $\ell'$-rps
$$
e = (
\xymatrix@1{
  H(H\times\Id)
  \ar[r]^-{H\pi_0}
  &
  HH
  \ar[r]^-{H\inr}
  &
  H(K'+H)
  \ar[r]^{H\kappa}
  &
  H\ext{K'+H}
}
) \,\text{.}
$$
It is easy to verify that $c^{-1}$ is its solution; and, as we see from the proof of Theorem~\ref{thm:ellrps}, we obtain a new abstract GSOS rule $n:(K'+H)(H\times\Id)\to H\ext{K'+H}$ defined as in Construction~\ref{constr:s} and having the $n$-interpretation $[b',c^{-1}]$. According to Proposition~\ref{prop:order} we can do this construction at any step when defining operations, the result being always a GSOS rule $\ell:K(H\times\Id)\to H\ext K$ with $K=K'+H$ and an $\ell$-interpretation $[b',c^{-1}]$ containing the algebraic structure $c^{-1}$.
\end{rem}

\begin{summary}\label{sum:2}
  Again, we summarize and review our work in this section.  
  Given an appropriate tuple
  $$
  (c: C\to HC, \quad \ell: K(H\times\Id) \to H\ext{K}, \quad b: KC \to C)
  $$
  as in Summary~\ref{sum:1} we know from Theorem~\ref{thm:distcia}
  that $c^{-1} \cdot H\ext b$ is a cia for the functor $H\ext K$. The
  work in this section shows that the appropriate tuple
  in~\refeq{eq:tuple} is obtained as follows: 
  Let $V$ be a functor, let $F = K + V$, and let $e:V(H\times \Id)\to H\ext{F}$ be
  an $\ell$-rps.  We take the unique interpreted solution $s: VC\to C$ of
  $e$, and then we obtain the appropriate tuple
  $$
  (c: C\to HC, \quad n: F(H\times\Id) \to H\ext{F}, \quad a = [b,s]: FC \to C).
  $$  
  extending the earlier appropriate tuple. 
  \takeout{ 
    We say that the tuple
    $$
  (c: C\to HC, \quad \ell: K(H\times\Id) \to H\ext{K}, \quad b: KC \to C)
  $$
  is an \emph{appropriate cia structure} if $c$ is a terminal 
  coalgebra structure,  $\ell$ is an abstract GSOS rule,
  $b$ is its interpretation, and $c^{-1}\cdot H\ext{b}$ is a cia for
  the functor $H\ext{K}$.  Given an appropriate cia, the work in this
  section shows how to obtain another appropriate cia: Let $V$ be a
  functor, let $F = K + V$, and let $e:V(H\times \Id)\to H\ext{F}$ be
  an $\ell$-rps.  We take the interpreted solution $s: VC\to C$ of
  $e$, and then we consider
  $$
  (c: C\to HC, \quad n: F(H\times\Id) \to H\ext{F}, \quad [b,s]: FC \to C)
  $$  
  This is again an appropriate cia,  and it extends the earlier
  appropriate cia. 
  }

  This provides the desired \emph{modularity principle} for 
  solutions of $\ell$-rps's as discussed in the introduction and again
  (restricted to flat equations) in Remark~\ref{rem:comp}. On the
  level of concrete syntactic specifications, the operations provided
  by the unique solution $s:\Var C \to C$ of any $\ell$-rps $e$ may
  occur as givens in any subsequent $n$-rps $f: W(H \times \Id) \to
  H\ext {F+W}$, which in turn has a unique solution $t: WC \to C$ that
  yields another appropriate tuple etc. So our results iterate as
  desired. 
\end{summary}

\subsection{Sandwiched $\ell$-rps's}

We will now prove a version of the Sandwich Theorem~\ref{thm:sandwich}
for $\ell$-rps's. The goal is to be able to solve specifications from
a wider class uniquely. Let us again explain the idea using 
the example of streams of reals. Here we have $HX=\Real\times X$ on $\Set$. Suppose that $K$ and $V$ both are polynomial functors associated with a signature of givens and newly defined operations on the terminal coalgebra $C=\Real^\omega$. The inverse of the coalgebra structure $c=\langle\hd,\tl\rangle$ yields the family of prefix operations $r.(-)$ prepending the number $r$ to a stream. The format of an $\ell$-rps means that the new operations of type $V$ are always defined by an equation with a prefix operation as a guard at its head, see e.\,g.~the following specification of the shuffle product reformulated from the behavioral differential equation of~\refeq{eq:shuffle_product_again}:
$$r.x\otimes s.y = rs.((r.x\otimes y)+(x\otimes s.y)) \,\text{.}$$
This guard $rs,-$ is sufficient to ensure a unique solution $\otimes$. However, in
general it is not necessary that the guard appears at the head of the
term on the right-hand side of an equation. We shall now define
the more general format of \emph{sandwiched} recursive program schemes that allow the guard
to occur further inside the term.
A concrete example of this kind of specification is the ``parallel composition'' of
non-well founded sets in~\refeq{eq:function_non-well-founded}; here
the guarding operation is set-bracketing $(x_i)_{i \in I} \mapsto \{
x_i \mid i \in I \}$, and this occurs inside the union operation (we
shall come back to this example in Section~6.5).

\begin{defi}
\label{dfn:sandwiched_ell-rps}
  A \emph{sandwiched recursive program scheme w.\,r.\,t.~$\ell$} (shortly,
  \emph{$\ell$-srps}) is a natural transformation $e: \Var (H\times\Id) \to
  \ext{K}H\ext{F}$.

  An \emph{interpreted solution} of $e$ in $C$ is a $\Var$-algebra structure
  $s: \Var C \to C$ such that 
  $$
  s = (
  \xymatrix@C+0.5pc{
    \Var C
    \ar[r]^-{\Var \langle c,\id_C\rangle}
    &
    \Var (HC\times C)
    \ar[r]^-{e_C}
    &
    \ext{K}H\ext{F}C
    \ar[r]^-{\ext{K} H\extalg{[b,s]}}
    &
    \ext{K}HC
    \ar[r]^-{\ext{K} c^{-1}}
    &
    \ext KC
    \ar[r]^-{\extalg b}
    &
    C
  }
  )\,\text{.}
  $$
\end{defi}

\begin{rem}
  (1)~The notion of an $\ell$-srps subsumes the one of an
  $\ell$-rps. Indeed, for a given $\ell$-rps $e: V(H\times \Id) \to
  H\ext{F}$ one forms the $\ell$-srps
  \[
  \xymatrix@1{
    V(H \times \Id) \ar[r]^-{e} & H\ext{K+V} \ar[rr]^{\eta H\ext{F}} &&
    \ext K H\ext{F}}.
  \]
  And one readily proves that solutions of this $\ell$-srps in $C$ are in
  one-to-one correspondence with solutions of $e$ in $C$.

  \medskip\noindent
  (2)~The notion of an $\ell$-(s)rps is incomparable to the three
  formats of rps's we saw in Section~\ref{sec:solthms}
  (see Section~\ref{rem:eqns}(2)). Indeed, for polynomial functors on $\Set$, the
  latter rps's are systems of recursive function equations that allow
  infinite trees on the right-hand sides while $\ell$-(s)rps's
  correspond to recursive equations where right-hand sides are
  restricted to consist of terms (or finite trees) only.

  However, notice that classical recursive program schemes as
  in~\cite{guessarian} are restricted to allow only finite trees on
  right-hand sides of equations (as, e.\,g., in~\refeq{eq:rps}). Such
  a restricted rps corresponds to a natural transformation
  \[
  e = (
  \xymatrix@1{
    \Var \ar[r] & \ext{H+\Var} \ar[r]^-m & T^{H+\Var}
    }),
  \]
  where $m: \ext{H+\Var} \to T^{H+\Var}$ is the unique monad morphism from
  the free monad on $H+\Var$ to $T^{H+\Var}$ induced by $\kappa^{H+\Var}: H + \Var \to
  T^{H + \Var}$. This given recursive program scheme is guarded if we have $f: \Var \to
  H\ext{H+\Var}$ satisfying a similar property as $f$ in
  Definition~\ref{dfn:algeq}. Now this clearly yields an $\ell$-rps
  \[
  e' = (
  \xymatrix{
    \Var(H \times \Id) \ar[r]^-{\Var \pi_1} & \Var \ar[r]^-f & H
    \ext{H + \Var}
  }),
  \]
  where $K = H$ and 
  \[
  \ell = (
  \xymatrix{
    H(H\times \Id) \ar[r]^-{H\pi_1} & H \ar[r]^-{H\eta} & H \ext H
    }).
  \]
  Again one readily proves that solutions of $e'$ are in one-to-one
  correspondence with solutions of $e$, and so the restricted rps's
  are a special instance of $\ell$-rps's. 
\end{rem}

We are now ready to state and prove the main result of this paper; it
provides a unique solution theorem for $\ell$-rps's, the most general
specification format we consider in this paper. In order to simplify
notation we will write $(M,\eta,\mu)$ for the free monad $\ext{K+\Var}$ in
the rest of this section.
  
\begin{thm} {\rm (Sandwich Theorem for $\ell$-srps's)}
  \label{thm:sandwich_rps}
  For every $\ell$-srps there exists a unique interpreted solution $s$ in $C$.
  In addition, $s$ extends the cia structure on $C$;   more precisely, the following is the structure
  of a cia for $M H M$ on $C$:
  \begin{equation}
    \label{eq:sandwich_extcia}
  \xymatrix@1{
    M H M (C)
    \ar[rr]^-{M H\extalg{[b,s]}}
    &&
    MH(C)
    \ar[rr]^-{Mc^{-1}}
    &&
    M(C)
    \ar[r]^-{\extalg{[b,s]}}
    &
    C \,\text{.}
    }
  \end{equation}
\end{thm}

\begin{proof}
Recall from Remark~\ref{rem:inverse_of_c} that we can assume $K=K'+H$ and that the $\ell$-inter\-pre\-ta\-tion has the form $[b',c^{-1}]$. Furthermore recall from Remark~\ref{rem:ell_lambda}(2) that $\ell$ gives rise to a distributive law $\lambda:\ext K(H\times\Id)\to (H\times\Id)\ext K$ of the free monad $\ext K$ over the cofree copointed functor $H\times\Id$.

  Given an $\ell$-srps $e: \Var (H\times\Id) \to
  \ext{K}HM$, we form the (ordinary) $\ell$-rps $\ol e:\Var(H\times\Id) \to HM$ by defining
  $$
  \xymatrix{
    \ol e = (\; \Var (H\times\Id)
    \ar[r]^-{e}
    &
    \ext{K}HM
    \ar[rr]^-{\ext K\langle HM,\inm M\rangle}
    &&
    \ext K(HM\times (K'+H+V)M)
    \ar `d[l] `l[dlll]_(0.4){\ext K(HM\times\varphi)} [dlll]
    \\
    \genfrac{}{}{0pt}{}{\ext K(HM\times M)}{= \ext K(H\times \Id)M}
    \ar[r]^-{\lambda M} 
    &
    (H\times\Id)\ext KM
    \ar[r]^-{\pi_0\ext KM}
    &
    H\ext{K}M
    \ar[r]^-{H\ext{\inl}M}
    &
    HMM \xrightarrow{H\mu} HM \; ) \,\text{.}
  }
  $$
  Then we verify that solutions of $e$ and $\ol e$ are in one-to-one correspondence. Consider the following diagram (we drop the indices denoting components of natural transformations):
  \begin{equation}
    \label{diag:solution_ell-srps}
  \vcenter{
  \xymatrix{
    & \Var C \ar[rrr]^{s} \ar[d]_{\Var \langle c,\id\rangle}	& \ar@{}[dd]|{\text{(i)}}	&	& C \ar[dddd]^{\langle c,\id\rangle} \ar@{}[dddddddr]|{\text{(v)}}	&	\\
    & \Var (HC\times C) \ar[d]_{e} \ar `l[dl] `[dddddd]_{\ol e} [dddddd] \ar@{}[ddddddl]_(0.435){\text{(ii)}}	&	&	&	\\
    & \ext KHMC \ar[d]|{\ext K\langle \id,\phi\cdot\inm M\rangle} \ar[r]^{\ext KH\extalg{[b,s]}}	& \ext KHC \ar[r]^{\ext Kc^{-1}} \ar@{}[d]|{\text{(iii)}}	& \ext KC \ar[uur]^{\extalg b} \ar[d]^{\ext K\langle c,\id\rangle} \ar@{}[r]|{\text{(iv)}}	&	&	\\
    & \ext K(H\times\Id)MC \ar[d]_{\lambda M} \ar[rr]^{\ext K(H\times\Id)\extalg{[b,s]}}	& \ar@{}[d]|{\text{(vi)}}	& \ext K(H\times\Id)C \ar[d]^{\lambda}	&	&	\\
    & (H\times\Id)\ext KMC \ar[d]_{\pi_0\ext KM}
    \ar[rr]^{(H\times\Id)\ext K\extalg{[b,s]}}	&
    \ar@{}[d]|{\text{(vii)}}	& (H\times\Id)\ext KC \ar[d]^{\pi_0
      \ext K} \ar[r]^(.4){\begin{turn}{-45}$\labelstyle (H\times\Id)\extalg b$\end{turn}}	& (H\times\Id)C \ar[ddd]^{\pi_0}	&	\\
    & H\ext{K}MC \ar[d]_{H\ext\inl M} \ar[rr]^{H\ext K\extalg{[b,s]}}	& \ar@{}[d]|{\text{(ix)}}	& H\ext KC \ar[d]^(0.55){H\ext\inl} \ar[ddr]^{H\extalg b} \ar@{}[r]|{\text{(viii)}}	&	&	\\
    & HMMC \ar[d]_{H\mu} \ar[rr]^{HM\extalg{[b,s]}}	& \ar@{}[d]|{\text{(xi)}}	& HMC \ar[dr]_{H\extalg{[b,s]}} \ar@{}[r]|(0.35){\text{(x)}}	&	&	\\
    & HMC \ar[rrr]^{H\extalg{[b,s]}}	&	&	& HC \ar@{<-} `r[ur] `[uuuuuuu]_{c} [uuuuuuu]	&
  }
  }
  \end{equation}
  We first show that all inner parts except part (i) commute: for part (ii) use the definition of $\ol e$, part (iv) is the commutative diagram in~\refeq{diag:interp_gsos_lambda}, part (v) trivially commutes, the parts (vi), (vii), (viii) and (ix) commute by the naturality of $\lambda$, $\pi_0:H\times\Id \to H$ and $\ext\inl:\ext K\to M$, respectively, for part (x) use Remark~\ref{rem:extinl}(3), and part (xi) commutes since $\extalg{[b,s]}$ is the structure of an Eilenberg-Moore algebra for $M$. It remains to verify the commutativity of part (iii); we remove $\ext K$ and consider the product components separately: the left-hand component commutes using $c \cdot c^{-1} = \id_C$ and for the right-hand one we have the diagram
  $$
  \xymatrix@C+4pc{
    HMC \ar[rr]^{H\extalg{[b,s]}} \ar[d]_{\inm M}	&	& HC \ar[dd]^{c^{-1}} \ar[dl]_{\inm}	\\
    (K'+H+\Var)MC \ar[r]^{(K'+H+\Var)\extalg{[b,s]}} \ar[d]_{\varphi}	& (K'+H+\Var)C \ar[dr]_{[b,s]}	&	\\
    MC \ar[rr]_{\extalg{[b,s]}}	&	& C \,\text{.}
  }
  $$
  Its upper part commutes by the naturality of $\inm:H\to K'+H+\Var$, for the lower part recall from Notation~\ref{not:extalg} the definition of $\extalg{[b,s]}$ as a homomorphism of algebras for $K'+H+\Var$, and the right-hand triangle commutes since $b=[b',c^{-1}]$, see Remark~\ref{rem:inverse_of_c}.

  Now if $s$ is a solution of $\ol e$, then the outside of Diagram~\refeq{diag:solution_ell-srps} commutes.
  Therefore part (i) commutes, and this 
   proves that $s$ is a solution of $e$. Conversely, if $s$ is a solution of $e$ part (i) commutes and then $s$ is also a solution of $\ol e$. By Theorem~\ref{thm:ellrps}, $\ol e$ has a unique solution; thus, $e$ has a unique solution, too.

  We still need to prove that \refeq{eq:sandwich_extcia} is the
  structure of a cia for $MHM$. %
  \takeout{
  Let $e':X\to \ext{K}HMX + C$ be a flat equation morphism and consider the following diagram:
  $$
  \xymatrix{
    X \ar[ddd]_{e'} \ar[rr]^{e^{\dagger}}	&	& C	\\
    	& \ext{K}C + C \ar[ur]^{[\extalg b,\id_C]} \ar[r]^{\ext{\inl}_C+\id_C}	& MC + C \ar[u]_{[\extalg{[b,s]},\id_C]}	\\
    	& \ext{K}HC + C \ar[u]_{\ext{K}c^{-1}+\id_C} \ar[dr]^{\ext{\inl}_{HC}+\id_C}	&	\\
    \ext{K}HMX + C \ar[r]^{\ext{K}HMe^{\dagger}+\id_C} \ar[d]_{\ext{\inl}_{HMX}+\id_C}	& \ext{K}HMC + C \ar[u]^{\ext{K}H\extalg{[b,s]}+\id_C} \ar[dr]_(0.3){\ext{\inl}_{HMC}+\id_C}	& MHC + C \ar[uu]_{Mc^{-1}+\id_C}	\\
    MHMX + C \ar[rr]^{MHMe^{\dagger}+\id_C}	&	& MHMC + C \ar[u]_{MH\extalg{[b,s]}+\id_C}
  }
  $$
  All small inner parts commute by naturality of $\ext{\inl}$ or by $\extalg{[b,s]} \cdot \ext{\inl}_C = \extalg b$ (see Remark \ref{rem:extinl}(3)). So we see that $e^{\dagger}$ is a solution of $e'$ (i.\,e., the big inner part commutes) if and only if it is a solution of $(\ext{\inl}_{HMX}+\id_C) \cdot e'$ (i.\,e., the outside commutes). But for the latter we know
  }
  But this follows from Theorem \ref{thm:sandwich}. Indeed, from the $\ell$-rps $\ol e$ we form the abstract GSOS rule $n:(K+\Var)(H\times\Id)\to HM$ analogously as in Construction~\ref{constr:s}, and the $n$-interpretation is $[b,s]$ for the unique solution $s$ of $\ol e$ (or, equivalently, of $e$).
   Now
   the morphism in \refeq{eq:sandwich_extcia} is the structure $k'$ from the statement of Theorem~\ref{thm:sandwich}.
\end{proof}

\begin{rem}
  Notice that every operation on the terminal coalgebra $C$ definable 
  by a sandwiched $\ell$-rps is also definable by an
  ordinary $\ell$-rps. In fact, the proof of the above theorem gives a
  reduction of a given sandwiched $\ell$-rps to an ordinary one with
  the same solution. However, sandwiched $\ell$-rps extend the
  syntactic format of recursive specifications uniquely specifying 
  operations on $C$. 
\end{rem}

\subsection{Modularity, again}
\label{sec:modagain}

We formulated  modularity principles in the Summaries~\ref{sum:1}
and~\ref{sum:2}. The same principles apply in the ``sandwiched'' case, mutatis mutandis.

Furthermore, we have modularity at the level of cia's:
\begin{enumerate}[(1)]
\item
  For $\ell$-srps's the same modularity principle as discussed in
  Remark~\ref{rem:comprps}(2) applies. Given an $\ell$-srps $e: V(H
  \times \Id) \to \ext K H \ext{K+V}$, then its unique solution $s: VC
  \to C$ arises as the unique solution of the $\ell$-rps $\ol e: V(H
  \times \Id) \to H \ext{K +V}$. Thus, as before one can form the abstract
  GSOS rule $n: F(H\times \Id) \to H \ext F$ for $F = K + V$ and use
  operations of type $F$ as givens in subsequent $n$-(s)rps's. 

\item In addition, one can readily formulate and prove a version of
  Proposition~\ref{prop:order} for sandwiched $\ell$-rps's. We leave this
  straightforward task as an exercise for the reader. 
\end{enumerate}

\section{Applications}
\label{sec:app}

In this section we present five applications illustrating how to use our results from Section \ref{sec:cias}--\ref{sec:lambdarps} to obtain unique solutions of recursive definitions in five different areas of theoretical computer science.

\medskip\noindent
{\bf 6.1.\ \ Process Algebras.} 
Recall Example~\ref{ex:dist}(3) where $HX = \Pow_\kappa (A \times
X)$. We shall first explain more in detail how the abstract GSOS rule
$\ell$ is obtained. Recall that $K$ is the polynomial
functor corresponding to the types of the CCS combinators, i.\,e., $KX$
is a coproduct of the following components (we shall denote elements in each component
by the corresponding flat terms):

\begin{enumerate}[(1)]
\item $A \times X$ with elements $a.x$, $a\in A$, for prefixing,
\item $\coprod_{n < \kappa} X^n$ with the $n$-the component consisting
  of elements $\sum_{i=1}^n x_i$, for summation,\footnote{The empty
    sum (for $n=0$) is denoted by $0$ as usual.}
\item $X \times X$ with elements $x_1|x_2$, for parallel composition,
\item $\coprod_f X$, where $f$ ranges over functions on
the action set $A$ with $\overline{f(a)}=f(\bar a)$ and $f(\tau) =
\tau$, with elements $x[f]$, for relabeling, and
\item $\coprod_{L \subseteq A\setminus\{\,\tau\,\}} X$ with elements
  $x \backslash L$, for restriction.
\end{enumerate}
The abstract GSOS rule $\ell: K(H\times\Id) \to H\ext K$ is given by
the sos rules in Example~\ref{ex:dist}(3) in 
terms of the components of the coproduct $K(H\times\Id)$, i.\,e., for each
combinator separately (in the same order as above):
\begin{enumerate}[(1)]
\item $\ell_X(a, S, x) = \{(a, x)\}$ where $S \subseteq A \times X$,
\item $\ell_X((S_i, x_i)_{i < n}) = \bigcup_{i < n} S_i$ for every $n < \kappa$, where $(S_i)_{i < n}$ is
  an $n$-tuple of sets $S_i \subseteq A \times X$,
\item $\ell_X(S_1, x_1, S_2, x_2)$ is given by the union of the three sets
$$\{(a,x|x_2) \mid (a,x)\in S_1\}\,\text{,}\quad \{(a,x_1|x) \mid (a,x)\in S_2\}\quad \text{and}$$
$$\{(\tau,x|y) \mid (a,x)\in S_1, (\bar a,y)\in S_2 \text{ for some } a \in A \setminus \{\tau\}\},$$
where $S_1, S_2 \subseteq A\times X$,
\item $\ell_X(S, x) = \{(f(a), y[f]) \mid (a,y) \in S\}$, and
\item $\ell_X(S, x) = \{(a,y\backslash L) \mid (a,y) \in S, a,\bar a \not\in L\}$.
\end{enumerate}
The form of these definitions is very similar to the ones given by
Aczel~\cite{aczel} in the setting of non-well-founded set theory. We already
mentioned the $\ell$-interpretation $b: K C \to C$ giving the desired
operations on CCS agents, and this gives the two new cia structures for
$H\ext K$ and $\ext K H \ext K$ as in Theorems~\ref{thm:distcia}
and~\ref{thm:sandwich}. 

\takeout{ 
\begin{rem}
  \smnote{I'd delete this remark.}\dsnote{Ok with me}
  If we replaced the second component $\coprod_{n < \kappa} X^n$ of $K$ by
  $\Pow_\kappa X$ we still have an abstract GSOS rule.
  Furthermore, in both cases the induced (binary) operation of summation is automatically 
  commutative, associative and idempotent: these three laws that have to be
  proved in process theory come ``for free'' by
  encoding them in the 
  abstract GSOS rule using the union operation.
\end{rem}} 

Now let us recall Milner's solution theorem for CCS agents
from~\cite{milner}. Suppose that $E_i$, $i \in I$, are agent
expressions with the free variables $x_i$, $i \in I$. Suppose further
that each variable $x_j$ in each $E_i$, $i,j\in I$ is \emph{weakly guarded}, i.\,e.,
it only occurs within the scope of some prefix combinator $a.-$. Then there
is a unique solution of the system 
\[
x_i = E_i, \qquad i \in I,
\]
of mutually recursive equations. More precisely, let $\sim$ denote strong
bisimilarity, and let $E_i[\vec{P}/\vec{x}]$ denote simultaneous
substitution of $P_j$ for $x_j$ for every $j$. Then we have
\begin{thm} {\rm\cite{milner}}
  \label{thm:proc}
  There exist, up to $\sim$, unique CCS agents $P_i$ such that $P_i \sim E_i[\vec{P}
  / \vec{x} ]$ holds for each $i \in I$. 
\end{thm}

It is easy to see that this theorem is a consequence of our
Theorem~\ref{thm:sandwich_ellsolution} because a system $x_i = E_i$ where each
variable is weakly guarded is essentially the same as  a map $X \to \ext K H \ext K
X$, where $X = \{x_i\mid i \in I\}$. 

Furthermore, our Theorem~\ref{thm:sandwich}
allows us to obtain unique solutions of flat equation morphisms $X \to
\ext K H \ext K X + C$. The extra summand $C$ allows us to use
constant agents in recursive specifications. 

\begin{exa}
  \label{ex:comp}
  Consider the recursive equation
  \[
  x = a.(x|c) + b.0
  \]
  from the introduction, where $0$ denotes the empty sum (i.\,e.~the
  ``inactive agent''). For any set $X$, we identify elements of 
  $
  H\ext K X = \Pow_\kappa(A \times \ext K X)
  $
  with terms of the form
  \[
  \sum\limits_{i < \kappa} a_i.t_i
  \]
  (modulo associativity, commutativity, idempotency and the unit laws
  for the sum), where $a_i\in A$ and $t_i$ is a process term on variables from $X$
  for every $i < \kappa$. 

  So the above equation clearly gives an $\ell$-equation $f_0: \{x\}
  \to H \ext K\{x\}$, and so we have its unique solution $\sol f_0(x) =
  P$ in $C$ from Theorem~\ref{thm:ellsolution}. Now consider the
  following system
  \begin{equation}
    \label{eq:yP}
    y = b.(y+z)|a.z 
    \qquad
    z = P.
  \end{equation}  
  This gives a flat equation morphism $\{y,z\} \to \ext K H \ext
  K\{y,z\} + C$, which has a unique solution in $C$ by
  Theorem~\ref{thm:sandwich}. Next recall Remark~\ref{rem:comp} and
  notice that the latter equation morphism has the form $\sol f \after
  e$, where $e: \{y,z\} \to \ext K H \ext K \{y,z\} + \{x\}$ is given
  by two equations: the first equation in~\refeq{eq:yP} and $z = x$, and $f: \{x\} \to
  \ext K H \ext K \{x\} + C$ arises
  from $f_0$ by forming 
  \[
  f = \inl \cdot (\eta H \ext K)_{\{x\}}
  \]
  as explained in Section~\ref{rem:eqns}(1). So we see that we obtain
  the same solution when we use the constant $P \in C$ in the
  recursive equation~\refeq{eq:yP} in lieu of forming the composed
  system
  \[
  x = a.(x|c) + b.0 \qquad y = b.(y+z)| a.z \qquad z = a.(x|c) + b.0
  \]
  that corresponds to the equation morphism $f \plus e$. 
  \takeout{ 
  So, for example, we can obtain the agent $P$ as the
  unique solution of
  $$x = a.(x|c) + b.0,$$
  where $0$ denotes the empty sum (the so-called ``inactive agent'')
  in the introduction and then use it in a system like
  $$x = b.(x + y) \qquad y = P$$
  which has a unique solution by Theorem~\ref{thm:sandwich}.}
\end{exa}
We now turn to an applications of Theorem~\ref{thm:sandwich_rps} that
shows how to uniquely define new process combinators.  Suppose we want
to define the binary combinator ``$\alt$'' which performs alternation
of two processes. For its definition we shall first need another
binary combinator, sequential composition of two processes (denoted by
the infix ``$;$''). This is defined by the following operations rules
in GSOS format:
\[
\frac{E \stackrel{a}{\to} E'}{E;F \stackrel{a}{\to} E'; F}
\qquad
\frac{F \stackrel{a}{\to} F'}{0;F \stackrel{a}{\to} F'}\,.
\]
Here we suppose that the latter combinator is
already included in our basic calculus---more precisely, we add a
sixth coproduct component $X\times X$ to $K$ for sequential
composition in the above definition and complete $\ell$ above by
$$\ell_X(S_1,x_1,S_2,x_2) = \begin{cases}\{(a,x;x_2) \mid (a,x)\in S_1\} & \text{if }S_1\neq\emptyset \\ S_2 & \text{if }S_1=\emptyset\end{cases}$$
for this coproduct component. The $\ell$-interpretation then gives indeed the desired combinator
for sequential composition. Moreover, as a consequence of
Theorem~\ref{thm:ellsolution} (now applied to the extended $K$), we
see that that Theorem~\ref{thm:proc} still holds for the calculus including this sixth combinator. Now for this extended $\ell$ we give a sandwiched $\ell$-rps $e:\Var(H\times\Id)\to \ext KH\ext{K+\Var}$, where $\Var X=X\times X$, in order to define the combinator $\alt$:
$$e_X(S_1,x_1,S_2,x_2) = \begin{cases}S_1;\{(a,x;\alt(x_1,x_2)) \mid (a,x)\in S_2\} & \text{if }S_2\neq\emptyset \\ \{(a,x;\alt(x_2,x_1)) \mid (a,x)\in S_1\} & \text{if } S_1\neq\emptyset, S_2=\emptyset \\ \emptyset & \text{if }S_1=S_2=\emptyset \,\text{.}\end{cases}$$
Notice that the term in the first line of this definition does not lie
in $H\ext{K+\Var}(X)$, so Theorem~\ref{thm:ellrps} cannot be
applied. Notice also that the above definition of $e$ cannot be
directly translated into an operational rule in GSOS format as for the
sequential composition before. However, $e$ could first be transformed into
an ordinary $\ell$-rps using the construction from the proof of
Theorem~\ref{thm:sandwich_rps} and then translated to a set of
operational rules in GSOS format. More directly, observe that the term
in the first line does lie in $\ext KH\ext{K+\Var}$(X), so
Theorem~\ref{thm:sandwich_rps} tells us that $e$ has a unique solution
$s:C\times C\to C$. It is not difficult to see that $s = \mathsf{alt}$
is the desired alternation combinator. Furthermore,
Theorem~\ref{thm:sandwich_rps} tells us that $s$ extends the cia
structure on $C$ from Theorem~\ref{thm:sandwich} applied to the
abstract GSOS rule $\ell$. And, as before, this means that Theorem~\ref{thm:proc} remains true for the calculus extended by sequential composition and alternation of processes, without further work.


\takeout{ 
Finally, suppose we want to define two unary combinators $\op_1$ and
$\op_2$ by the rule
$$
\frac{E \stackrel{a}{\to} F}
{\op_1 (E) \stackrel{a}{\to} F | \op_2(F + E)\quad \op_2 (E) \stackrel{a}{\to} F + \op_1(F | E) }\,\text{.}
$$
Then Theorem~\ref{thm:ellrps} tells us that this rule uniquely
determines the two combinators. Indeed, we translate the rule into an
$\ell$-rps: let $\Var = \Id + \Id$ (two unary combinators are
specified) and let $e: \Var (H\times\Id) \to H\ext{K +\Var}$ be given by 
$$e(S, x) = \{(a, y | \op_2(y + x)) \mid (a,y) \in S\}$$
on the first component and 
$$e(S, x) = \{(a, y + \op_1(y | x)) \mid (a,y) \in S\}$$
on the second one. The unique solution of $e$ gives us two 
new unary combinators on $C$ extending its cia structure. Again this means
that Theorem~\ref{thm:proc} remains true for CCS extended by the two
combinators $\op_1$ and $\op_2$ without further work.
} 

\medskip
\noindent
{\bf 6.2.\ \ Streams.} Recall from Example~\ref{ex:dist}(2) that here we take $HX = \Real
\times X$ and we have $C = \Real^\omega$ with the structure given by
$\< \hd, \tl\>: C \to \Real \times C$. %
\takeout{
Further recall that stream operations are defined by behavioral differential equations. For example, the componentwise addition of two streams $\sigma$ and $\tau$ is specified by
\begin{equation} \label{eq:stream_addition}
(\sigma + \tau)_0 = \sigma_0 + \tau_0 \qquad\qquad
(\sigma + \tau)' = (\sigma' + \tau') \,\text{,}
\end{equation}
and the shuffle product from the introduction is specified by
\begin{equation} \label{eq:shuffle_product_again}
(\sigma \otimes \tau)_0 = \sigma_0\cdot\tau_0 \qquad\qquad
(\sigma \otimes \tau)' = (\sigma \otimes \tau' + \sigma'\otimes \tau) \,\text{.}
\end{equation}
}%
Rutten gives in \cite{rutten_stream} a general theorem for the existence of the solution of systems of behavioral differential equations. We will now recall this result and show that it is a special instance of our Theorem~\ref{thm:ellrps}. For a system of behavioral differential equations one starts with the signature $\Sigma$ of all the operations to be specified. One uses an infinite supply of variables, and for each variable $x$ there is also a variable $x'$ and a variable $x(0)$ (also written as $x_0$). For each operation symbol $f$ from $\Sigma$ one specifies
\begin{equation} \label{eq:general_beh_diff_eq}
f(x_1,\dots,x_n)_0 = h_f(x_1(0),\dots,x_n(0)) \qquad\qquad
f(x_1,\dots,x_n)' = t_f \,\text{,}
\end{equation}
where $h_f$ denotes a function from $\Real^n$ to $\Real$ and $t_f$ is a term built from operation symbols from $\Sigma$ on variables $x_i$, $x_i'$ and $x_i(0)$, $i=1,\dots,n$. Theorem A.1 of \cite{rutten_stream} asserts that for every $f$ from $\Sigma$ there exists a unique function $(\Real^\omega)^n\to\Real^\omega$ satisfying the equation \refeq{eq:general_beh_diff_eq} above.

We shall now show that every system of behavioral differential equations~\refeq{eq:general_beh_diff_eq} gives rise to an $\ell$-rps for a suitable abstract GSOS rule $\ell$. To this end let $KX=\Real$ be the constant functor and let
$$
  \ell = (
  \xymatrix@1{
    K(H\times\Id)
    \ar[r]^-{\ell'}
    &
    HK
    \ar[r]^-{H\kappa}
    &
    H\ext K
  }
  )
$$
with $\ell'$ given by $\ell_X'(r) = (r,0)$. Then the $\ell$-interpretation $b:\Real\to C$ assigns to every $r\in\Real$ the stream $b(r) = (r,0,0,\dots)$.

Now given the system \refeq{eq:general_beh_diff_eq} let $\Var$ be the polynomial functor associated with $\Sigma$, see Example~\ref{ex:sigma_algebras}. Notice that $K+\Var$ is the polynomial functor of the signature $\Sigma$ extended with a constant symbol $r$ for every real number $r$. We translate the system \refeq{eq:general_beh_diff_eq} into an $\ell$-rps $e:V(H\times\Id)\to H\ext{K+\Var}$ as follows. For every $f$ from $\Sigma$ the corresponding component of $e_X$ is defined by
$$
e_X((r_1,x_1',x_1),\dots,(r_n,x_n',x_n)) = (h_f(r_1,\dots,r_n),\ol{t_f})\,\text{,}
$$
where the term $\ol{t_f}\in \ext{K+\Var}(X)$ is obtained by replacing in $t_f$ all variables $x_i(0)$ by the constant $r_i$. Notice also, that here $h_f(r_1,\dots,r_n)$ is a real number (the value of $h_f$ at $(r_1,\dots,r_n)$) whereas in \refeq{eq:general_beh_diff_eq} we have formal application of $h_f$ to the variables $x_i(0)$.

It is now straightforward to verify that a solution of $e$ in $C$ corresponds precisely to a solution of the system \refeq{eq:general_beh_diff_eq}. Thus, we obtain from Theorem~\ref{thm:ellrps} the

\begin{thm}[\cite{rutten_stream}] \label{thm:beh_diff_eq_streams}
Every system of behavioral differential equations has a unique solution.
\end{thm}

\begin{exa}
For the system given by \refeq{eq:stream_addition} and \refeq{eq:shuffle_product_again} we have $\Var X = X\times X + X\times X$ and $e$ given componentwise as follows: for the $+$ component we have
\begin{equation} \label{eq:stream_addition_ell-rps}
e_X((r,x',x),(s,y',y)) = (r+s,x'+y')
\end{equation}
and for the $\otimes$ component we have
\begin{equation} \label{eq:shuffle_product_ell-rps}
e_X((r,x',x),(s,y',y)) = (r\cdot s,(x\otimes y')+(x'\otimes y)) \,\text{.}
\end{equation}
\end{exa}

Observe that the systems of behavioral differential equations do not
distinguish between given operations and newly defined ones. However,
our result in Theorem~\ref{thm:ellrps} allows us to make this
distinction, and the modularity principle for solutions of
$\ell$-rps's (cf.~Summary~\ref{sum:2}) means that operations specified by
behavioural differential equations may be used in subsequent
behavioral differential equations as given operations in the terms
$t_f$ from \refeq{eq:general_beh_diff_eq}. We believe this
modularity of unique solutions for behavioral differential equations is a new
result.

\begin{exa}
Take $\Var X = X\times X$ and the $\ell$-rps $e: V(H\times \Id) \to
H\ext{K+V}$ given by \refeq{eq:stream_addition_ell-rps} whose solution
is the operation of stream addition $+: VC \cdot C$. As shown in
Construction~\ref{constr:s}, $\ell$ and $e$ yield the abstract GSOS
rule $n:F(H\times\Id)\to H\ext F$ for $F=K+\Var$. Now let $\Var_1X =
X\times X$. Then \refeq{eq:shuffle_product_ell-rps} yields an $n$-rps
$e_1: V_1(H\times \Id) \to H\ext{F+V_1}$ whose solution is the shuffle
product $\otimes: V_1 C \to C$.
\end{exa}

Next, we present an example illustrating Proposition~\ref{prop:order}.
\begin{exa}
Continuing the previous example, consider the convolution product of streams specified by
$$
(\sigma \times \tau)_0 = \sigma_0\cdot\tau_0 \qquad\qquad
(\sigma \times \tau)' = (\sigma' \times \tau + \sigma_0 \times \tau') \,\text{,}
$$
see \cite{rutten_stream}. Let $\Var_2X = X\times X$ and let the
$n$-rps $e_2: V_2(H \times \Id) \to H\ext{F+V_2}$ be given by
\begin{equation} \label{eq:convolution_product_ell-rps}
(e_2)_X((r,x',x),(s,y',y)) = (r\cdot s,(x'\times y)+(r\times y')) \,\text{.}
\end{equation}
Notice that this illustrates why we introduced the constants $r$; in this way we are able to deal
 with $\sigma_0$ in the equation for $(\sigma\times\tau)'$. Then the unique solution of $e_2$ is the convolution product as expected. Proposition~\ref{prop:order} asserts that the extended cia structure for $H\ext F$, $F=K+\Var+\Var_1+\Var_2$ on $C$ does not depend on the order of taking the solution of \refeq{eq:shuffle_product_ell-rps} and \refeq{eq:convolution_product_ell-rps}---either way this is given by the constants coming from $b$ and the operations of stream addition as well as convolution and shuffle product.
\end{exa}

\takeout{ 

  Our results also allow us to obtain unique solutions of
  specifications that go beyond behavioral differential equations, and
  we now provide one example.
\begin{exa}\label{ex:unzip}
We give a sandwiched rps w.\,r.\,t.~$\ell$ defining the binary operation $\sigma\circ\tau$ assigning to two streams $\sigma$ and $\tau$ the ``under-sampled'' zipped stream
$$
(\sigma_0, \sigma_1, \ldots) \circ (\tau_0, \tau_1, \ldots) = 
(0,\sigma_0,0,\tau_0,0,\sigma_1,0,\tau_1,\dots) \,\text{.}
$$
Indeed, let $KX=1+\Real+X\times X+X\times X$ be the polynomial functor
for the signature having a constant symbol $\mathds{X}$, constant
symbols for every $r\in\Real$, and two binary symbols $+$ and
$\times$. The abstract GSOS rule $\ell:K(H\times\Id)\to H\ext K$ is
given by the following system of behavioral differential equations:
\[
\begin{array}{rcl@{\qquad}rcl}
  \mathds{X}_0 & = & 0 
  & 
  \mathds{X}' & = & 1 
  \\
  r_0 & = & r 
  & 
  r' & = & 0 
  \\
  (\sigma + \tau)_0 & = & \sigma_0 + \tau_0 
  & 
  (\sigma + \tau)' & = & \sigma' + \tau ' 
  \\
  (\sigma \times \tau)_0 & = & \sigma_0 \cdot \tau_0 
  & 
  (\sigma \times \tau)' & = & (\sigma' \times \tau) + (\sigma_0 \times \tau')
\end{array}
\]
\takeout{ 
for each of the coproduct components separately by
\[
\ell_X(\mathds{X})=(0,1)
\qquad\text{and}\qquad
\ell_X(r)=(r,0)
\]
for the first component and for the $\Real$ component, respectively,
and similarly as in~\refeq{eq:stream_addition_ell-rps}
and~\refeq{eq:convolution_product_ell-rps} for $+$ and $\times$:
\begin{eqnarray*}
  \ell_X((r,x',x),(s,y',y)) & = & (r+s, x'+y'), \\
  \ell_X((r,x',x),(s,y',y)) & = & (r\cdot s, (x'\times y) + (r \times y')).
\end{eqnarray*}
}
Now let $\Var X=X\times X$ (i.\,e., $\Var$ corresponds to one binary
operation $\circ$). For any set $X$, the elements of $\ext K H
\ext{K+V}(X)$ can be identified with terms of the given operations
(i.\,e.~$\mathds{X}$, $r$, $+$ and $\times$) on the set of terms of
the form $r.t$, where $r \in \Real$ and $t$ is a term on $X$ using givens and $\circ$. So let 
\[
e:\Var(H\times\Id)\to \ext KH\ext{K+\Var}
\]
be the $\ell$-srps given by
\[
e_X((r,x',x),(s,y',y)) = \mathds{X} \times r.(y\circ x') \,\text{.}
\]
It has a unique interpreted solution $C\times C\to C$ by
Theorem~\ref{thm:sandwich_rps}, and one easily verifies that this
solution is the desired operation of ``undersampled'' zipping
$(\sigma, \tau) \mapsto \sigma \circ \tau$.
\end{exa}
}


\paragraph{\bf Stream circuits.}
We now turn to another method to define operations on streams---stream circuits \cite{rutten_tut}, which are also called (signal) flow graphs in the literature. We shall demonstrate that specification of operations by stream circuits arises as a special case of our results. Stream circuits are
usually defined as pictorial compositions of the following basic
stream circuits
$$
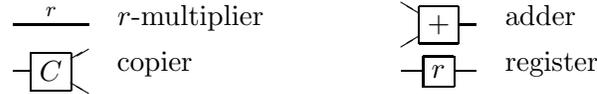
\begin{array}{c@{\quad}p{3cm}@{\qquad}c@{\quad}p{3cm}}
\xy
\POS (000,000) = "a"
   , (010,000) = "b"
\ar@{-}^r "a";"b"
\endxy
&
$r$-multiplier
&
\xy
\POS (000,000) = "o"
   , (-05,000) *+[F-]{+} = "+"
   , (-10,-03) = "i1"
   , (-10,003) = "i2"
\ar@{-} "o";"+"
\ar@{-} "i1";"+"
\ar@{-} "i2";"+"
\endxy
&
adder
\\
\xy
\POS (000,000) = "i"
   , (005,000) *+[F-]{C} = "C"
   , (010,-03) = "o1"
   , (010,003) = "o2"
\ar@{-} "i";"C"
\ar@{-} "o1";"C"
\ar@{-} "o2";"C"
\endxy
&
copier
&
\xy
\POS (000,000) = "i"
   , (005,000) *+[F-]{r} = "r"
   , (010,000) = "o"
\ar@{-} "i";"r"
\ar@{-} "o";"r"
\endxy
&
register
\end{array}
$$
The $r$-multiplier multiplies all elements in a stream by $r \in
\Real$, the adder performs componentwise addition, the copier yields
two copies of a stream, and the register prepends $r \in \Real$ to a
stream $\sigma$ to yield $r.\sigma$. The stream circuits are then
built from the basic circuits by plugging wires together, and there
may also be feedback (loops). For example the following picture shows
a simple stream circuit: 
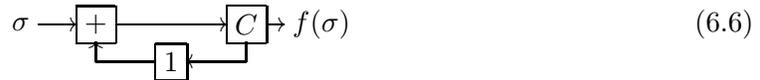
\begin{equation}
  \label{eq:circuit}
  \xy
  \POS (000,000) *+{\sigma} = "i"
     , (010,000) *+[F-]{+} = "+"
     , (020,-05) *+[F-]{1} = "r"
     , (030,000) *+[F-]{C} = "C"
     , (040,000) *+{f(\sigma)}  = "o"
     , (030,-05) = "aux1"
     , (010,-05) = "aux2"
  \ar "i";"+"
  \ar "+";"C"
  \ar "C";"o"
  \ar@{-} "C";"aux1"
  \ar@{-} "r";"aux2"
  \ar "aux1";"r"
  \ar "aux2";"+"
  \endxy
\end{equation}
It defines the following unary operation on stream:
\[
f(\sigma)=(1+\sigma_0,1+\sigma_0+\sigma_1,1+\sigma_0+\sigma_1+\sigma_2,\dots).
\]

For our treatment we shall consider the operations presented by
$r$-multipliers, adders and registers as givens. So let $K$ be the polynomial
functor associated with the signature $\Sigma$ given by these operations
(copying will be implicit via variable sharing). In
symbols, $KX = \Real \times X + X \times X + \Real \times X$. Our given
operations are defined by the behavioral differential equations:
$$
\begin{array}{rcl@{\qquad}rcl}
  (r\sigma)_0 & = & r\sigma_0 & (r\sigma)' & = & r\sigma' \\
  (\sigma + \tau)_0 & = & \sigma_0 + \tau_0 & 
  (\sigma + \tau)' & = & \sigma' + \tau' \\
  (r.\sigma)_0 & = & r & (r.\sigma)' & = & \sigma
\end{array}
$$
\sloppypar As explained above these definitions easily give rise to an
abstract GSOS rule $\ell: K(H\times\Id) \to H\ext{K}$ (see also
\cite[Section~3.5.1]{bartels_thesis}). We then get the
$\ell$-interpretation in $C$ and the corresponding extended cia
structures by Theorems~\ref{thm:distcia} and~\ref{thm:sandwich}. A
stream circuit is called \emph{valid} if every loop passes through at
least one register. It is well-known that every finite valid stream
circuit with one input and one output defines a unique stream
function (see~\cite{rutten_tut}). Of course, a similar result holds for
more than one input and output, and we present here a new proof of
this result based on our Theorem~\ref{thm:ellrps}.

\begin{thm}
  \label{thm:stream}
  Every finite valid stream circuit defines a unique stream function at every output.
\end{thm}
\begin{proof}
Let a finite valid stream circuit be given. We explain how to construct
an $\ell$-rps from the circuit. Notice first that the wires in a circuit can
be regarded as directed edges (cf.~\refeq{eq:circuit}). We take for every register $R$ in our circuit an
operation symbol $g_R$ and define its arity as the number of inputs
that can be reached by following all possible paths from $R$ backwards
through the circuit. Similarly, we take for every output edge $O$ of
the circuit an operation symbol $f_O$ with the arity obtained in the
same way. Let $\Sigma$ be the signature of all $f_O$ and $g_R$, and let
$\Var$ be the corresponding polynomial functor for $\Sigma$. To give an
$\ell$-rps $e: \Var (H\times\Id) \to H \ext{K + \Var}$ it suffices to
give a natural transformation $e':\Var H \to H \ext{K + \Var}$ and to
define $e = e' \cdot \Var\pi_0$, where $\pi_0:H\times\Id\to H$ is the projection. To obtain $e'$, we
give for each $n$-ary symbol $s$ from $\Sigma$ an assignment
$$
s(r_1.x_1, \ldots, r_n.x_n) \mapsto (r,t)
$$
where $r_1, \ldots, r_n, r\in \Real$ and $t$ is a term built from symbols of $\Sigma$
and of the signature $\Gamma$ of basic circuit operations using the
variables $x_1, \ldots, x_n$. Notice that the arguments of $s$ stand
for generic elements $(r_i,x_i)$ from $HX$ for some set $X$ and that
$r$ may depend on $r_i$ and $t$ may contain operation symbols $r_i.-$.
We now show how to define the above assignment for each operation
symbol $g_R$. Suppose that the register $R$ has the initial value $r$. Then
$$
g_R(r_1.x_1, \ldots, r_n.x_n) \mapsto (r, t_R),
$$
and we now explain how to obtain $t_R$: in order to construct the term
$t_R$ one starts in $R$ and traverses from there every possible path in
the circuit backwards (i.\,e., one follows edges from inputs to
outputs of basic circuits) adding for every basic stream circuit the
corresponding operation symbol to $t_R$ until 
\begin{enumerate}[(1)]
\item an input edge corresponding to some argument $r_i.x_i$ is met, or
\item some register is met.
\end{enumerate}
More precisely, we construct $t_R$ as a $(\Sigma+\Gamma)$-tree: we follow the
input edge of $R$ backwards until we reach either the output wire of
an $r$-multiplier, the output wire of an adder, an input wire of the whole circuit or the
output wire of a register. For an $r$-multiplier or an adder we add a
node to $t_R$ labeled by the corresponding operation symbol and
continue this process for each input node of the $r$-multiplier or
adder constructing the corresponding subtrees of $t_R$. For an input
wire corresponding to $r_i.x_i$ add a node labeled by the prefix
operation $r_i.-$ and below that a leaf labeled by $x_i$; for a
register $S$ add the tree (of height 2) given by $g_S (r_{i_1}.x_{i_1},
\ldots, r_{i_k}.x_{i_k})$, where the $r_{i_j}.x_{i_j}$ correspond to those
input wires of the circuit backwards reachable from the register $S$. (Notice
that these arguments of $g_S$ form a subset of the arguments $\{\,r_1.x_1, \ldots,
r_n.x_n\,\}$ of $g_R$ since every input that is backwards reachable
from $S$ is also backwards reachable from $R$. Also notice that copiers are
ignored while forming $t_R$.) Since the given circuit $C$ is valid we
have indeed constructed only a finite tree, whence a term $t_R$.

We still need to define the assignment corresponding to $e'$ for output
symbols $f_O$:
$$
f_O(r_1.x_1, \ldots, r_n.x_n) \mapsto (r, t_O) \,\text{.}
$$
We first form the tree $t_O'$ in essentially the same way as $t_R$ for a register $R$
with the difference that for every input wire and for every register
we just insert an unlabeled leaf for the moment. To obtain $r$, label
every leaf of $t'_O$ corresponding to the input $r_i.x_i$ by $r_i$ and every
leaf corresponding to a register by its initial value; now evaluate
the corresponding term to get $r$. In order to get $t_O$ one replaces
leaves of $t_O'$ corresponding to inputs $r_i.x_i$ by $x_i$, and
register leaves are replaced by the second components $t_S$ from the
right-hand sides of the equations for $g_S (r_{i_1}.x_{i_1}, \ldots,
r_{i_k}.x_{i_k})$.

Finally, the unique solution of $e$ yields a unique operation $f_O$ on
streams for every output $O$. By construction this is the operation
computing the stream circuit. 
\end{proof}

\takeout{
\begin{rem}
  \smnote{Delete this remark? Seems not clear enough.}\dsnote{Ok with me}
  Suppose we are given a finite valid stream circuit where, in
  addition, every path from an input to an output passes through a
  register. Then the construction in the  proof above would not need to
  refer to the behavior of the input $r_i.x_i$. That means that we could
  assume ``structureless'' inputs $x_1, \ldots, x_n$, and the 
  construction above then even gives a guarded rps $\Var \to
  \T{H\ext{K + \Var}}$. Corollary~\ref{cor:rps} allows us to obtain a
  unique solution of this rps, and this result even allows for
  the unique solution of infinite valid stream circuits where every
  path from an input to an output passes through a register. 
\end{rem}
}

The modularity principle for the unique solution of an
$\ell$-rps which we discussed in Summary~\ref{sum:2} yields \emph{modularity
of stream circuits}: every stream circuit can be used as
a building block as if it were a basic operation in subsequent stream
circuits. And Theorem~\ref{thm:stream} remains valid for the extended
circuits.

\begin{exa}
The proof of Theorem~\ref{thm:stream} 
essentially gives a translation of an arbitrary finite valid stream circuit
into an $\ell$-rps. We demonstrate this on the circuit given in~\refeq{eq:circuit}
above. First we introduce for the output a function symbol $f$ and for
the register output the function symbol $g$. To determine their arity we
count the number of input wires which have a (directed) path to the
register and the output, respectively. In both cases the arity is
one. Now we must give a definition of $f(r.x)$ and $g(r.x)$ for an
abstract input stream with head $r \in \Real$. These definitions
are each given by a pair $(s, t)$ where $s \in \Real$ and $t$ is a
term in the one variable $x$ over operations corresponding to the
basic circuits and $f$, $g$.  We define
$$
g(r.x) = (1, r.x + g(r.x))
\qquad
f(r.x) = (r+1, x + (r.x + g(r.x))) \,\text{.}
$$
For $g(r.x)$ we take the value $1$ of the register as first component,
and the right-hand term is obtained as follows: we follow all paths
from the register backwards until we find an input or a
register. So we get a finite tree or,
equivalently, the desired term. For $f(r.x)$ we first follow all paths
to inputs and registers backwards to get the term $t'=x_I+x_R$, where
$x_I$ represents the input and $x_R$ the register. For the first
component of $f(r.x)$ we evaluate $t'$ with the head $r$ of the input
and the initial value $1$ of the register, i.\,e., one evaluates
$t'[r/x_I,1/x_R]$. And for the second
component we replace in $t'$ the input by $x$ and the register by the
second component of the right-hand side of the definition above of
$g(r.x)$, i.\,e., one forms $t'[x/x_I, r.x + g(r.x) / x_R]$. The two equations above
are easily seen to yield an $\ell$-rps $e: \Var (H\times\Id) \to H\ext{K+\Var}$,
where $\Var = \Id + \Id$ is the polynomial functor for the signature
with two unary symbols $f$ and $g$.  The unique solution of $e$ gives
two unary operations $f_C, g_C: C \to C$, and $f_C$ is precisely the function computed by the
circuit~\refeq{eq:circuit} and $g_C$ is the stream function output by
the register in~\refeq{eq:circuit}. By the modularity of stream circuits
explained above, we can use $f$ (and also $g$) as
``black-boxes'' in subsequent stream circuits.
\end{exa}

\medskip
\noindent
{\bf 6.3.\ \ Infinite Trees.}
Rutten and Silva \cite{rs10} developed behavioral differential equations for infinite trees and proved a unique solution theorem for them. Here we shall show that we obtain their theorem as a special instance of our Theorem~\ref{thm:ellrps}.
The work is 
similar to what we saw in Theorem~\ref{thm:beh_diff_eq_streams} for streams.

Let $HX = X\times \Real \times X$. The terminal coalgebra $C$ for $H$ consists of all infinite binary trees with nodes labeled in $\Real$, and the terminal coalgebra structure $c:C\to C\times\Real\times C$ assigns to a tree the triple $(t_L,r,t_R)$ where $r$ is the node label of the root of $t$ and $t_L$ and $t_R$ are the trees rooted at the left-hand right-hand child nodes of the root of $t$. Single trees (constants) and operations on trees can be specified by behavioral differential equations. For example,
$$
\mathsf{pi}(\varepsilon) = \pi \qquad\qquad \genfrac{}{}{0pt}{0}{\mathsf{pi}_L = \mathsf{pi}}{\mathsf{pi}_R = \mathsf{pi}}
$$
specifies the tree with every node labeled by the number $\pi$. For every real number $r$ we have the constant $[r]$ specified by
$$
[r](\varepsilon) = r \qquad\qquad \genfrac{}{}{0pt}{0}{[r]_L = [0]\phantom{\,\text{.}}}{[r]_R = [0]\,\text{.}}
$$
The nodewise addition of numbers stored in the nodes of the trees $t$ and $s$ is defined by
$$
(t+s)(\varepsilon) = t(\varepsilon)+s(\varepsilon) \qquad\qquad \genfrac{}{}{0pt}{0}{(t+s)_L = t_L+s_L\phantom{\,\text{.}}}{(t+s)_R = t_R+s_R\,\text{.}}
$$
See~\cite{rs10} for further and more exciting examples.


In general a system of behavioral differential equations is specified
as follows. Again, we assume an infinite supply of syntactic
variables. For every variable $x$ we have the notational variants $x_L$, $x_R$ and also $x(\varepsilon)$. Furthermore, let $\Sigma$ be a signature of operations to be specified. For each operation symbol $f$ from $\Sigma$ of arity $n$ we provide equations of the form
\begin{equation} \label{eq:general_beh_diff_eq_trees}
\begin{tabular}{c|c}
\text{initial value}	& \text{differential equations}	\\
\hline
\rule{0mm}{20pt}
$(f(x_1,\dots,x_n))(\varepsilon) = c_f(x_1(\varepsilon),\dots,x_n(\varepsilon))$	& $\genfrac{}{}{0pt}{0}{f(x_1,\dots,x_n)_L = t_1}{f(x_1,\dots,x_n)_R = t_2}$
\end{tabular}
\end{equation}
where $c_f$ denotes a function $\Real^n\to\Real$ and $t_1$ and $t_2$ are $\Sigma$-terms on the variables $x_1,\dots,x_n$ and their three notational variants.

\begin{thm} {\rm (\cite{rs10}, Theorem 2)} \label{thm:trees}
Every system \refeq{eq:general_beh_diff_eq_trees} of behavioral differential equations has a unique solution, i.\,e., for every $f$ from $\Sigma$ there exists a unique function $f:C^n\to C$ satisfying \refeq{eq:general_beh_diff_eq_trees} (denoted by the same symbol).
\end{thm}

We present a new, short proof of this result based on Theorem~\ref{thm:ellrps}. Let $KX = \Real$ be the constant functor, and let the abstract GSOS rule $\ell$ be
$$
\ell = (
\xymatrix{
  K(H\times\Id)
  \ar[r]^-{\ell'}
  &
  HK
  \ar[r]^-{H\kappa}
  &
  H\ext K
}
)
$$
where the natural transformation $\ell'$ is given by $\ell_X'(r) = (0,r,0)$. Then the $\ell$-interpretation is $b:\Real\to C$ with $b(r)=[r]$. 
Let $\Var$ be the polynomial functor associated with $\Sigma$.
Every system \refeq{eq:general_beh_diff_eq_trees} gives an $\ell$-rps $e:\Var (H\times\Id)\to H\ext{K+\Var}$ as follows: 
let $e$ be given on each component corresponding to $f$ from $\Sigma$ by
$$
e_X(((x_1)_L,r_1,(x_1)_R,x_1),\dots,((x_n)_L,r_n,(x_n)_R,x_n)) = (\ol{t_1},c_f(r_1,\dots,r_n),\ol{t_2}) \,\text{,}
$$
where $\ol{t_i}$ is obtained from $t_i$ by replacing each $x_i(\varepsilon)$ by the corresponding constant $r_i$. The solutions of $e$ in $C$ correspond precisely to solutions of \refeq{eq:general_beh_diff_eq_trees}; thus, Theorem~\ref{thm:trees} follows from Theorem~\ref{thm:ellrps}.

In addition, we have again a modularity principle
(cf.~Summary~\ref{sum:2}): operations specified by behavioral
differential equations may be used as givens in subsequent behavioral
differential equations.

\medskip\noindent
{\bf 6.4.\ \ Formal Languages.} Recall Example~\ref{ex:dist}(3); here we have $HX = X^A \times 2$ on
$\Set$. A coalgebra $x : X \to X^A \times 2$ for $H$ is precisely a
deterministic automaton with the (possibly infinite) state set $X$.  Here $C = \Pow(A^*)$, and the
unique homomorphism $h: (X,x) \to (C,c)$ assigns to each state the
language it accepts. We shall now show how various standard operations on
formal languages can be defined in a modular way using
Theorem~\ref{thm:ellrps}.   Working in a bialgebraic setting, Jacobs~\cite{jacobs_kleene}
shows that these   operations can
be defined as interpretations of one abstract GSOS rule (or distributive law) in $C$.
 However, this account
does not explain why one may define these operations in a step-by-step
fashion by successive recursive definitions. This is the added value
of Theorem~\ref{thm:ellrps}. 

We start with the functor $K_0 = C_\emptyset$ (that means, we start
from scratch with no given operations) and with $\ell_0: C_\emptyset (H\times\Id)
\to H \ext{C_\emptyset} = H$ given by the empty maps. The corresponding
interpretation is the empty map $b:\emptyset\to C$, and $\extalg b$ is
then the identity on $C$. Thus, the cia structure for
$H\ext{K_0}$ on $C$ given by Theorem~\ref{thm:distcia} is simply the
initial cia $(C, c^{-1})$ for $H$.  At each subsequent step we are
given a functor $K_i$ and an abstract GSOS rule 
\[
\ell_i: K_i(H\times\Id) \to H \ext{K_i}
\quad\text{with its interpretation}
\quad b_i: K_i C \to C \,\text{.} 
\]
We then give an $\ell_i$-rps 
\[e_i: \Var_i(H\times\Id) \to H\ext{K_i + \Var_i} \,\text{,}
\] 
and its unique solution $s_i: \Var_i C \to C$ extends the cia structure as follows:
let $K_{i+1} = K_i + \Var_i$ and let 
\[ 
\ell_{i+1} = [H\ext\inl \cdot \ell_i, e_i]: K_{i+1}(H\times\Id) \to
H\ext{K_{i+1}} \,\text{,}
\]
where $\ext\inl: \ext{K_i} \to
\ext{K_{i+1}}$ is the monad morphism induced by $\inl: K_i \to
K_{i+1}$ (cf.~Notation~\ref{not:coinj}(1)); so the formation of
$\ell_{i+1}$ from $\ell_i$ corresponds precisely to forming $n$ from
$\ell$ in Construction~\ref{constr:s}.  
By induction it is easy to see that the $\ell_{i+1}$-interpretation is 
\[
b_{i+1} = [s_j]_{j=0,\ldots,i}: K_{i+1} C \to C \,\text{.} 
\]
And this gives an extended cia 
\[
c^{-1}\cdot H\extalg b_{i+1}: H\ext{K_{i+1}} (C) \to C
\] 
by Theorem~\ref{thm:distcia}.

As a first step we define constants in $C$ for $\emptyset$,
$\{\eps\}$, and $\{a\}$, for each $a \in A$, 
as solutions of an $\ell_0$-rps. 
We express this as
an $\ell_0$-rps as follows: take the functor $\Var_0 X = 1 + 1 + A$ corresponding to these languages. We
define $e_0: \Var_0 (H\times\Id) \to H\ext{K_0 + \Var_0} = H\ext{\Var_0}$
componentwise. We write for every set $X$, $\emptyset$ for $\inj_1(*)
\in \Var_0 X$ and $\eps$ for $\inj_2(*) \in \Var_0 X$. Then $(e_0)_X$ is given by the assignments
\begin{align} \label{eq:constant_languages}
\emptyset	&\mapsto ((\emptyset)_{a \in A}, 0)	&& \nonumber \\
\eps		&\mapsto ((\emptyset)_{a \in A}, 1)	&& \\[-10pt]
a		&\mapsto ((t_b)_{b\in A}, 0), \qquad \text{where } t_b = \begin{cases}\eps & \text{if } b=a \\ \emptyset & \text{else}\,\text{.}\end{cases} \nonumber 
\end{align}
It is now straightforward to check that the unique
solution $s_0$ of $e_0$ yields the desired constants in $C$ extending
the cia structure.

Next we add the operations of union, intersection and language
complement to the cia structure. Let $K_1 = K_0 + \Var_0$ and let
$\ell_1 = [H \wh\inl \cdot \ell_0, e_0]$ as above with interpretation
$b_1 = s_0$. Let 
\[
\Var_1 X = X \times X + X \times X + X
\] 
be the polynomial functor corresponding to two binary symbols $\cup$ and
$\cap$ and one unary one $\ol{(-)}$. We give the $\ell_1$-rps 
\[
e_1: \Var_1 (H\times\Id) \to H\ext{K_1 + \Var_1}
\]
componentwise in the form of the three assignments in
\refeq{eq:union_intersection_complement} below. We write $((x_a), j, x)$ for
elements of $HX\times X$, where $(x_a)$ is an $|A|$-tuple, i.\,e., an
element of $X^A$. We also write elements of $\Var_2Z$, $Z=HX\times X$,
as flat terms $z_1\cup z_2$, $z_1\cap z_2$ and $\ol z$ on the
left-hand side for the three components of $(e_1)_Z$:
\begin{equation} \label{eq:union_intersection_complement}
\begin{array}{rcl@{\qquad}rcl}
  ((x_a), j, x) \cup ((y_a), k, y) & \mapsto & ((x_a \cup y_a), j \vee k) &
  \ol{((x_a), j, x)} & \mapsto & ((\ol{x_a}), \neg j) \\
  ((x_a), j, x) \cap ((y_a), k, y) & \mapsto & ((x_a \cap y_a), j \wedge k)
\end{array}
\end{equation}
where $\vee$, $\wedge$ and $\neg$ are the evident operations on $2 =
\{0,1\}$. The corresponding unique solution $s_1: \Var_1 C \to C$ is
easily checked to provide the desired operations extending the cia
structure on $C$.

The next step adds (language) concatenation to the cia structure on $C$. For this
let 
\[
\Var_2 X = X \times X
\qquad
\text{and let}
\qquad
e_2: V_2(H \times \Id) \to H\ext{K_2 + V_2}
\] 
be given by the assignment (where elements of $V_2 Z$ are written as $z_1 \cdot z_2$ on the left-hand side)
\begin{equation} \label{eq:concatenation}
((x_a), j, x) \cdot ((y_a), k, y) \mapsto ((t_a), j \wedge k) \qquad \text{ where } t_a = \begin{cases}(x_a \cdot y) \cup y_a & \text{if }j = 1 \\ x_a \cdot y & \text{else}\,\text{.}\end{cases}
\end{equation}

As the final step we add the Kleene star operation by taking 
\[
\Var_3 X = X
\qquad\text{and}\qquad
e_3: V_3(H \times \Id) \to H\ext{K_3 + V_3}
\]
to be given by $$e_3((x_a), j, x) = ((x_a\cdot x^*),1)$$
with the unique solution $s_3= (-)^*: C \to C$.
Notice that this definition makes use of concatenation which was
a solution at the previous stage, and concatenation makes use of union
which was a solution at stage $1$. This corresponds to the fact that 
for all languages $L$, $(L^*)^a = L^a \cdot L^*$.
Its unique solution $s_2: C \times C \to C$ is
the concatenation operation.

\begin{rem} \label{rem:definability}
There are many further operations on formal languages that are definable by $\ell$-rps's, including the following ones:
\begin{iteMize}{$\bullet$}
\item prefixing $a.L = \{aw \mid w\in L\}$ for any $a\in A$;
\item the operation given by $c^{-1}:C^A\times 2\to C$ (see Remark~\ref{rem:inverse_of_c})
$$((L_a),j) \mapsto \begin{cases}\bigcup_{a\in A}a.L_a & \text{if }j=0 \\ \bigcup_{a\in A}a.L_a \cup \{\varepsilon\} & \text{else}\,\text{;}\end{cases}$$
\item $\shf(L_1,L_2) = \bigcup_{w_1\in L_1, w_2\in L_2}\shf(w_1,w_2)$ where $\shf(w_1,w_2)$ is the usual operation merging the words $w_1$ and $w_2$.
\end{iteMize}
\end{rem}

\noindent We leave it to the reader to work out the details. An example of an
operation on languages not definable by any $\ell$-rps (or abstract
GSOS rule) is the language derivative $L \mapsto L^a = \{w \mid aw \in
L\}$ for every $a \in A$; the argument is similar as for the
non-definablity of the tail operations on streams in
Example~\ref{ex:tl}, and so we leave those details to the reader, too.

\takeout{ 
Notice, however,
that there exist operations that cannot be defined by $\ell$-rps's (or
abstract GSOS rules):
\begin{exa}
  The language derivative $L^a = \{w \mid aw \in L\}$ for some $a\in A$
  cannot be defined by $\ell$-rps's. Indeed, if this was definable by
  an $\ell$-rps, Theorem~\ref{thm:ellrps} would yield a cia structure on
  $C$ for the functor $H\ext{K+V}$, where $VX=X$ corresponds to the
  unary operation $(-)^a$. For a term $t$ in $\ext{K+V}(X)$, we shall
  use the notation $a.t$ as a shortcut for $((t_b), 0) \in
  H\ext{K+V}(X)$ for the $|A|$-tuple $(t_b)$ with $t_a = t$ and $t_b =
  \emptyset$ for every $b\in A\setminus\{a\}$. Thus, for $X=\{x\}$ the flat equation
  morphism $e:X\to H\ext{K+V}(X)+C$ given by $e(x) = a.x^a$ would have a
  unique solution. 
  But this is clearly not the case: every
  formal language $L'$ whose words all start with $a$ gives a solution
  $e^\dag(x)=L'$ of the flat equation morphism $e$.
\end{exa}
} 

\paragraph{\bf Contex-free grammars.}
Next we show how context-free grammars in Greibach Normal Form and
their generated languages are special instances of flat equations of the form
$e:X\to \ext KH\ext K$ and their unique solutions in $C$ for a
suitable functor $K$. 
(Note that these flat equations do not involve elements of $C$;
that is, we do not need equations of the form $e:X\to \ext KH\ext K + C$.)
A coalgebraic description of context-free
grammars in Greibach normal form has recently been given by Bonsangue,
Rutten and Winter~\cite{wbr11}, and previously, a coalgebraic
approach to context-free grammars was given by Hasuo and Jacobs~\cite{hj05}. Our
approach here is completely different.

 Recall (e.\,g.~from \cite{hmu}) that a \emph{context-free grammar} is a four-tuple $G=(A,N,P,S)$ where $A$ is a non-empty finite set of \emph{terminal} symbols, $N$ a finite set of \emph{non-terminal} symbols, $P\subseteq N\times (A+N)^*$ is a finite relation with elements called \emph{production rules} of $G$, and $S\in N$ is the starting symbol. As usual we write $n\to w$ for $(n,w)\in P$. 
A context-free grammar $G$ is in \emph{Greibach Normal Form} (\emph{GNF}, for short) if all its production rules are of the form $n\to aw$ with $a\in A$ and $w\in N^*$. 
The language \emph{generated} by a context-free grammar $G$ is the set of all words over $A$ that arise by starting with the string $S$ and repeatedly substituting substrings according to the production rules of the grammar, and eliminating $\varepsilon$ from the string whenever it occurs.

To see that context-free grammars in GNF yield flat equation morphisms,
 we consider the constant $\emptyset$ and the operations of union and
 concatenation as given operations. More precisely, let 
\[
KX=1+X\times X+X\times X
\] 
be the polynomial functor corresponding to $\emptyset$, $\cup$ and $\cdot$, and let $\ell:K(H\times\Id)\to H\ext K$ be the abstract GSOS rule given by the corresponding assignments in \refeq{eq:constant_languages}, \refeq{eq:union_intersection_complement} and \refeq{eq:concatenation} with the interpretation $b:KC\to C$ given as desired. By Theorem~\ref{thm:sandwich} we obtain the cia structure $k':\ext KH\ext KC\to C$. Now observe that a flat equation morphism assigns to each $x\in X$ either an element $e(x)\in C$ or $e(x)$ corresponds to a term of given operations in $H\ext KX$.  
 
We now show that for every 
context-free grammar in GNF there is a flat equation morphism $e: X\to H\ext{K}X$ with $X$ finite
whose solution $e^\dag: X \to X$ has the property that $e^\dag(S)$ is the language of the grammar.
(Once again, note that each $e(x)$ will belong to $H\ext{K}X$, not just to the larger set $H\ext{K}X +C$.)

\begin{construction}
  Let $G = (A, N, P, S)$ be a context-free grammar in GNF. We define the flat
  equation morphism $e_G:X\to \ext KH\ext KX$. The set $X$ is simply
  the set $N$ of non-terminals of $G$, and $e_G$ is the following map:
  for $n\in N$ for which there is no production rule $n\to aw$ in $P$
  we take $e_G(n) = \emptyset$, the constant term in $\ext K(H\ext
  KX)$. Otherwise for each production rule $r = n\to aw$ in $P$ with
  $n \in N$ on the left-hand side we define the term $t_r\in H\ext KX$ as
  $$
  t_r = \left\{
  \begin{array}{ll}
    a.(n_1\cdot n_2\cdot \cdots \cdot n_k) & \text{if $w = n_1n_2
      \cdots n_k$ and $n_1, \ldots, n_k \in N$,} \\
    a.\emptyset & \text{if}\ w = \eps
  \end{array}
  \right.
  $$
  using the concatenation operation. Recall from
  Remark~\ref{rem:definability} the notation $a.t$ 
  to see that
  $t_r\in H\ext KX$. We define $e_G(n)$ as (the term in $\ext K(H\ext KX)$
  representing) the ``union'' of all right-hand sides of production
  rules $r_i = n \to aw_i$, $i=1,\dots,l$, in $P$; so in symbols we have
  $$e_G(n) = t_{r_1}\cup t_{r_2}\cup\cdots\cup t_{r_l} \,\text{.}$$

The point of using grammars in GNF is that each $t_r$ really belongs to $H\ext{K} X$,
and so each $e_G(n)$ belongs to $K(H\ext{K} X)$, hence to $\ext{K}(H\ext{K} X)$.
\end{construction}

It is not difficult to see that the language generated by the grammar $G$ is precisely the language $e_G^\dag(S)$, where $S$ is the starting symbol of $G$. So as a consequence of Theorem~\ref{thm:sandwich} we see that the language generated by $G$ arises as the unique solution of the flat equation morphism $e_G$.

Analogously, it is also possible to translate right-linear grammars (which are a special case of context-free grammars generating regular languages) into flat equation morphisms using the constant empty and empty-word languages as well as union as the given operations. Again Theorem~\ref{thm:sandwich} implies that there is a unique solution which yields the language generated by the given grammar by the translation.

\begin{rem}
  We have seen that by defining operations via $\ell$-rps's (or $\ell$-srps's) we
  obtain cia structures for $H\ext{K}$ (or $\ext KH\ext K$) on $C$. It is
  interesting to ask what formal languages can arise as solutions of $\ell$-equations
  $e: X \to H\ext{K} X$ (or sandwiched ones $e:X \to \ext K H\ext K
  X$) according to Theorems~\ref{thm:ellsolution} and~\ref{thm:sandwich_ellsolution} if
  $X$ is \emph{finite}. Not surprisingly one obtains
  precisely the regular languages for stages $i = 0, 1, 2$ in our 
  definition process above, i.\,e., when we add the constant languages $\emptyset$,
  $\{\varepsilon\}$ and $\{a\}$ for every $a\in A$, and the operations $\cup$,
  $\cap$ and $\ol{(-)}$. %
  But adding concatenation one obtains non-regular languages: if for $i = 3$ one
  restricts to using union and concatenation in the terms in
  $\ext{K_3}H\ext{K_3}$, the flat equation morphisms essentially correspond to
  context-free grammars in GNF.  
  However, using intersection and/or complement allows one to
  obtain non-context-free languages as solutions. Precisely what class of
  languages can be defined by (sandwiched) $\ell$-equations using different
  combinations of operations remains the subject of further work.
\end{rem}

\medskip
\noindent
{\bf 6.5.\ \ Non-well-founded Sets.} 
Finally, we come to an application not directly related to
computation. The theory of non-well-founded sets originated as a
framework for providing the semantics of general circular definition.  
For background on non-well-founded sets, the antifoundation axiom (\AFA), and
classes, please see the books~\cite{aczel,bm}. We work here on the
category $\A = \Class$ of classes. The results of Section~\ref{sec:lambdarps} hold
true for $\Class$ since every endofunctor of $\Class$ has terminal
coalgebras and free algebras (see~\cite{amv_classes}).

Consider $\Pow: \Class \to \Class$
taking a class $X$ to the class $\Pow X$ of sub\emph{sets} of
$X$. \AFA~is equivalent to the assertion that $(C,c)$ is a terminal
coalgebra, where $C$ is the class of all sets (usually written $V$), and $c: C\to \pow C$
takes a set and considers it a set of sets.  (That is, $c(s) = s$ for
all $s$.)  Let us note some natural transformations:
$$\begin{array}{l@{\qquad}l@{\qquad}l}
p: \Pow \to \Pow\Pow 
&  
op: \Id\times \Id \to \Pow\pow 
&
cp:  \Pow \times \Pow \to \Pow(\Id \times \Id)
\\
p_X(x) = \Pow(x) 
& 
op_X (x, y) =\set{\set{x},\set{x,y}}   
& 
cp_X(x,y) = x \times y
\end{array}
$$    
Also note that $c^{-1}$ is the operation on $C$ taking a family $x
\subseteq C$ of sets to the set $\{\,y \mid y \in x\,\}$. 

We will now define three additional operations on $C$: 

\begin{tabbing}
$\bullet$ the powerset operation \hspace{1em}\= $b_1: x \mapsto \{\,y \mid y
\subseteq x\,\}$, \\
$\bullet$ the Kuratowski pair \> $b_2: (x,y) \mapsto \{\{\,x\,\}, \{\,x,y\}\}$, and \\
$\bullet$ the cartesian product \> $b_3: (x,y) \mapsto x \times y$. 
\end{tabbing}
So let $K$ be the functor given by 
\[
KX = X + (X\times X) + (X\times X)+ \pow X + \Pow\Pow X;
\] 
its first three components represent (the type of) our three desired
operations, the fourth component $\pow$ represents $c^{-1}$ and the
fifth one represents $c^{-1}\cdot \pow c^{-1}$---the latter two are needed
for the definition of the former three.  We write the coproduct
injections of $K$ as $\inj_1, \ldots, \inj_5$.  We define a natural
transformation $\ell': K\pow \to \pow K$ componentwise, using
$$
\begin{array}{r@{\ }c@{\ }l}
  \ell' \cdot \inj_1\Pow & = & (\xymatrix@1@C+2.3pc{
  \Pow 
  \ar[r]^-{p}
  &
  \Pow\Pow
  \ar[r]^-{\Pow\inj_4} 
  &
  \Pow K
  })
  \\
  \ell'\cdot\inj_2\Pow & = &(\xymatrix@1@C+1pc{
    \Pow \times \Pow  \ar[r]^-{op\Pow}
  &    
  \Pow\Pow\pow \ar[r]^-{\Pow \inj_5} 
  & 
  \Pow K
  })
  \\
  \ell'\cdot\inj_3\Pow & = & (\xymatrix@1{
    \Pow \times \Pow  \ar[r]^-{cp}
    &    
    \Pow(\Id\times\Id) \ar[r]^-{\Pow \inj_2} 
    & 
    \pow K
  })
\end{array}
\qquad
\begin{array}{r@{\ }c@{\ }l}
  \ell'\cdot\inj_4\Pow & = & (\xymatrix@1@C+0.75pc{
    \Pow\Pow
    \ar[r]^-{\Pow\inj_4} 
    &
    \Pow K
  })
  \\
  \ell'\cdot\inj_5\Pow & = & (\xymatrix@1{
    \Pow\Pow\Pow
    \ar[r]^-{\Pow\inj_5} 
    &
    \Pow K
  })
\end{array}
$$
Then $\ell'$ yields an abstract GSOS rule
$$
\ell = (
\xymatrix@1{
  K(\Pow\times\Id)
  \ar[r]^-{K\pi_0}
  &
  K\Pow
  \ar[r]^-{\ell'}
  &
  \Pow K
  \ar[r]^-{\Pow\kappa}
  &
  \Pow\ext{K}
}
) \,\text{.}
$$
Let $b: K C \to C$ be the $\ell$-interpretation in $C$. %
\takeout{
Let us write $b: KC \to C$ for the evident map,  so that
    $b$ is the unique morphism making
the diagram below commute:
$$
      \xymatrix@C+2pc{
        KC  
        \ar[r]^-{Kc}
        \ar[d]_b
        &
      K  \Pow C 
        \ar[r]^-{\lambda_C}   
        &
        \Pow K C  
                \ar[d]^{\Pow b}
        \\
        C \ar[rr]_-c & & \Pow C
      }
$$ 
}%
Let us write $b_1, \ldots, b_5$ for the components of $b$, so
$b_i = b \cdot (\inj_i)_C$.
To obtain explicit formulas for these, we use
Diagram~\refeq{diag:interp_gsos_rule} and the above definitions to write:
$$\begin{array}{l}
c \cdot b_1 = \Pow b_4 \cdot p_C \cdot c  \\
c \cdot b_2 =   \pow b_5 \cdot  op_{\pow C} \cdot (c\times c)
\\
c \cdot b_3 = 
\pow b_2 \cdot cp_C \cdot (c\times c)  \\
\end{array}
\qquad
\begin{array}{l}  
c \cdot b_4 = \pow b_4 \cdot \pow c\\
c\cdot b_5 =   \pow b_5  \cdot \pow^2 c\\ 
\\ 
\end{array}
$$
We check easily that  $b_4 = c^{-1}$ and $b_5 = c^{-1} \cdot \pow c^{-1}$
satisfy the last two equations.
From these  we see that 
\begin{align*}
  b_1 &=  c^{-1} \cdot \Pow c^{-1} \cdot p_C \cdot c,
  \\
  b_2 &=  c^{-1} \cdot \pow b_5 \cdot  op_{\pow C} \cdot (c\times c),
  \quad\text{and}\quad
  \\
  b_3 &= c^{-1}\cdot \pow b_2 \cdot cp_C \cdot (c\times c).
\end{align*}
In words, $b_4$ and $b_5$ are the identity, and $b_1$, $b_2$ and $b_3$
are as desired. 

By Theorem~\ref{thm:distcia}, we have a cia structure $(C, c^{-1}\cdot \pow \extalg b)$
for the composite $\pow \ext K$.
\begin{rem}
  We could have obtained the various operations on $C$ in a step-by-step 
  fashion starting with $b_4$ and $b_5$ and then defining $b_1,
  b_2, b_3$ by successive applications of
  Theorem~\ref{thm:ellrps} as in the previous section on formal
  languages. We decided against this, to keep the
  presentation short. 
\end{rem}
Continuing our discussion of non-well-founded sets, we may solve
systems of equations which go beyond what one finds in the standard
literature on non-well-founded sets~\cite{aczel,bm}.  For example, one
may solve the system
$$x = \set{\pow(y)} \qquad y = \set{y\times y, z} \qquad z = \emptyset \,\text{,}$$
which gives rise to a flat equation morphism $e: X\to \pow \ext KX+C$
where $X=\{x,y,z\}$. The unique solution of $e$ satisfies $\sol e(z) =
\emptyset$ and assigns to $y$ the
non-well-founded set $\sol e(y)$ containing two elements: $\emptyset$
and the cartesian product of $\sol e (y)$ with itself. And $\sol e(x)$ is the singleton non-well-founded set
containing the powerset of $\sol e(y)$ as its only element. 

Finally, let us show how to obtain the operation specified by~\refeq{eq:function_non-well-founded}
in the introduction as a unique solution of a sandwiched $\ell$-rps according
to Theorem~\ref{thm:sandwich_rps}. Here we have the union operation
$\cup$ as a given operation. So let $KX = X \times X$ be the
corresponding polynomial functor, and let us consider the natural
transformation
\begin{align*}
  u_X: K\pow X \to \pow X \,\text{,} \\
  u_X(x,y) = x \cup y \,\text{.}
\end{align*}
We have the abstract GSOS rule
\[
\ell = ( \xymatrix@1{
  K(\pow \times \Id) \ar[r]^-{K\pi_0}
  &
  K\pow \ar[r]^-u
  &
  \pow
  \ar[r]^-{\pow \eta}
  &
  \pow \ext K
})
\]
whose interpretation is the union operation $\cup: C\times C \to C$. Now let $WX =
X\times X$. Then the equation~\refeq{eq:function_non-well-founded}
corresponds to a natural transformation
$e_0: W(\pow \times \Id) \to \ext K \pow W$, where $(e_0)_X$ maps a
pair 
\[
\big((\{x_i \mid i \in I\}, x), (\{y_j \mid j \in J\}, y)\big)
\in W(\Pow X \times X)
\]
to
\[
\{(x, y_j) \mid j \in J\} 
\cup \{(x_i,y) \mid i \in I\} 
\cup \{(x_i,y_j) \mid i \in I, j \in J\}
\in \ext K \Pow W X.
\]
From this we obtain an $\ell$-srps 
\[
e = (\xymatrix@1@C+2pc{
  W(\pow \times \Id) \ar[r]^-{e_0}
  &
  \ext K \pow W
  \ar[r]^-{\ext K \pow \inr}
  &
  \ext K \pow(K + W)
  \ar[r]^-{\ext K \pow \kappa^{K+W}}
  &
  \ext K \pow \ext{K + W}
}).
\]
Its unique solution in $C$ is the desired ``parallel composition''
$\|: C \times C \to C$ of non-well-founded sets.

\takeout{ 
Further, one may solve recursive function definitions such as
$$g(x) = \{g(\pow(x)) \times x, x\}$$
from the introduction uniquely. Indeed, for
$W = \Id$ this equation yields an $\ell$-rps $e: W(\pow\times\Id) \to \pow\ext{K +
  W}$ whose unique solution given by Theorem~\ref{thm:ellrps} is a
function $g_C: C \to C$ behaving as specified.}
 
\takeout{
We have $c\cdot b_5 =   \pow b_5  \cdot \pow\inj_5  \cdot \pow^2 c:
\pow^3 C \to C$.
Thus 
$b_5 = c^{-1} \cdot \pow c^{-1}$.

We also have $c \cdot b_4 = \pow b_4 \cdot \pow c$, and thus $b_4 = c^{-1}$.

We also have $c \cdot b_1 = \Pow b_4 \cdot p_C \cdot c$, and so $b_1 = c^{-1} \cdot \Pow c^{-1} \cdot p_C \cdot c$.

And 
$c \cdot b_3 = 
\pow b_5 \cdot \pow\, op_C \cdot cp_C \cdot (c\times c)$.

So $b_3 = c^{-1}\cdot \pow b_5 \cdot \pow\, op_C \cdot cp_C \cdot (c\times c)$.

    Using that commutative diagram one 
    easily verifies that the left-hand component of $b$ is the inverse of the
    terminal coalgebra structure $c$ on $C$, in symbols: $b \cdot \inl = c^{-1}$. The
    right-hand component is, of course, the identity by the unit law of the
    algebra structure $b$, in symbols: $b \cdot \inr = \id_C$. Finally, the middle
    component of $b$ turns out to assign to any set in $C$ its power set.

 Unordered pair. Take $KX = X \times X + \Pow X$. We will define a
    distributive law $\ell: K\Pow \to \Pow K$ of $K$ over $\Pow$. Let the
    natural transformation  $u: \Id\times \Id \to \Pow$ be given by $u_X (x, y)
    = \{\,x,y\,\}$. Now let
    $$
    \ell \equiv
    \xymatrix@1{K \Pow = \Pow \times \Pow + \Pow \Pow X \ar[rr]^-{[u\Pow, \id]}
      &&  \Pow\Pow \ar[r]^-{\Pow \inr} & \Pow(\Id \times \Id + \Pow)
      = \Pow K \,\text{.}}
    $$

    This lifts to a distributive law $\lambda$ of the pointed endofunctor $M =
    K + \Id$ over $\Pow$ (see Definition~\ref{def:dist}, Remark (1)).
    We get $b: MC\to C$ making the diagram below commute:
      \begin{equation}
    \vcenter{
      \xymatrix@C+2pc{
         C^2 + \Pow C + C
        \ar[r]^-{Mc}
        \ar[d]_b
        &
       \Pow C^2 + \Pow \Pow C + \Pow C
        \ar[r]^-{\lambda_C}   
        &
        \Pow (C^2 + \Pow C + C)
        \ar[d]^{\Pow b}
        \\
        C \ar[rr]_-c & & \Pow C
      }}
    \label{eq:interp:M:C:power}
  \end{equation}
    It is not difficult to verify that the middle component of  $b$ is the inverse of the terminal coalgebra
    structure $c: C \to\Pow C$. It follows that the left-hand component of $b$
    is the unordered pair operation on $C$ with $(a,b) \mapsto \{\,a,b\,\}$.
By Theorem~\ref{thm:distcia}, $(C,k)$ is a cia for $\Pow M$, where $k: \Pow M C \to C$ is
 $c^{-1}\cdot \Pow b$.
As a concrete example of a flat equation morphism in $C$, let $X = \set{x,y}$, and
$e: X\to  \pow( X^2 + \Pow X +X ) + C$ given by
$x\mapsto \set{(y,y)}$ in $\pow (X^2)$ and $y\mapsto \set{\set{x}} \in \pow^2 X$.
Its solution $\sol{e}$ corresponds to sets $a$
and $b$ satisfying $a = \set{b}$ and $b = $
 } 

\section{Conclusions}
In many areas of theoretical computer science, one is interested in
recursive definitions of functions on terminal coalgebras $C$ for
various functors $H$.  This paper provides a more comprehensive
foundation for recursive definitions than had been presented up until
now.  The overall idea is to present operations in terms of an
abstract GSOS rule $\ell:K(H\times\Id)\to H\ext K$. We
proved that $\ell$ induces new completely iterative algebra
structures for $H\ext K$ and $\ext KH\ext K$ on $C$.  As a result, we
are able to apply the existing body of solution theorems for cias to
obtain new unique solution theorems for recursive equations of more
general formats ``for free''. 

Next we introduced the notion of an $\ell$-rps and showed how to
solve recursive function definitions uniquely in $C$ which are given by
an $\ell$-rps.  Our results explain why taking unique solutions of
such equations is a modular process.  And we have seen that our
results can be applied to provide the semantics of recursive
specifications in a number of different areas of theoretical computer
science.

We also generalized this point to \emph{sandwiched} $\ell$-rps's.  The reason
for doing this was not to gain expressive power: every operation on
$C$ definable by a sandwiched $\ell$-rps is also definable by an
ordinary $\ell$-rps.  The reason for the generalization was to have
more usable syntactic specification formats. 

In concrete applications $\ell$-rps's (or abstract GSOS rules) are
mostly given by finite sets
of rules or equations. But the conversion from an abstract GSOS rule
$\ell: K(H \times \Id) \to H\ext K$ (for set functors) to a coalgebra
for $H$ involves a finite-to-infinite blow-up, i.\,e., one forms 
a $H$-coalgebra on $\ext KX$, which is typically an infinite set. We
leave as an open question the question to investigate exactly which
operations are definable by \emph{finite} $\ell$-(s)rps's.

Another related question concerns defining operations on the rational
fixpoint of a functor $H$~\mbox{\cite{amv_atwork,m_linexp}}; for a set
functor $H$ this is the subcoalgebra of the terminal $H$-coalgebra
given by the behavior of all finite coalgebras (e.\,g., regular
processes, rational trees or regular languages). The question is: when
does an abstract GSOS rule induce operations on the rational
fixpoint for $H$? The answer to this is discussed in the recent
paper~\cite{bmr12}, which is inspired by the work of
Aceto~\cite{aceto94} who studied specification formats for operations
on regular processes (see also~\cite{afv01}). 
 
There remain a number of other topics for further work. Firstly, it
should be interesting to identify in the various applications concrete
syntactic formats of operational rules that correspond to
$\ell$-(s)rps's.  For example, in the case of process algebras as
discussed in Section~6.1, Turi and Plotkin~\cite{tp} proved that
abstract GSOS rules correspond precisely to transition system
specifications with operational rules in the GSOS format
of~\cite{bim}. Bartels gave in his thesis~\cite{bartels_thesis}
concrete syntactic rule formats for abstract GSOS rules in several
other concrete cases and Klin~\cite{klin09} studied the instantiation
of abstract GSOS rules for weighted transition systems. But concrete
syntactic formats that are equivalently characterized by $\ell$-srps's
have not been studied yet.

Secondly, there remains the question whether our results can be
generalized to other algebras than the initial cia (alias terminal
coalgebra) along the lines of Capretta, Uustalu and Vene~\cite{cuv}
who generalized Bartels' results from~\cite{bartels,bartels_thesis}. A
different kind of generalization to be investigated concerns the step
from using free monads in the definition of $\ell$-(s)rps's to
arbitrary monads.  Thirdly, it should be interesting to work out
further applications. For example, it is clear that for weighted
transition system specifications as considered by Klin~\cite{klin09},
our results can be used to obtain unique solutions of recursive
specifications. It should also be interesting to investigate whether
our results yield unique solution theorems for name and value passing
process calculi as considered by Fiore and Turi~\cite{ft01}. Finally,
it would be of interest to extend the results of~\cite{mm_prop} on
properties of recursive program scheme solutions to the richer
settings of this paper.

\subsection*{Acknowledgements} We thank the anonymous referees for
their comments which helped us to improve the presentation of this
paper.

%
%
\bibliographystyle{abbrv}
\bibliography{ourpapers,coalgebra}

\ifappendix
%
%
\clearpage
\appendix

\section{Results for Pointed Functors}
\label{app:A}

We mentioned in Remark~\ref{rem:arbitrary_monad} that
Theorems~\ref{thm:distcia} and \ref{thm:sandwich} hold more generally
for pointed endofunctors $M$ in lieu of a free monad $M=\ext
K$. However, in this case we need our base category to be
cocomplete. In this appendix we provide the details.

\begin{ass}
  We assume that $\A$ is a cocomplete category, that $H:\A\to\A$ is a functor and that $(M,\eta)$ is a \emph{pointed functor} on $\A$, i.\,e., $M:\A\to\A$ is a functor and $\eta:\Id\to M$ is a natural transformation. As before $c:C\to HC$ is a terminal coalgebra for $H$.
\end{ass}

\begin{defi}\label{def:dist}
  (1)~An \emph{algebra for $(M,\eta)$} is a pair $(A, a)$
  where $A$ is an object of $\A$ and $a: MA \to A$ is a morphism satisfying the
  \emph{unit law} $a \cdot \eta_A = \id_A$. 

  \medskip\noindent
  (2)~A \emph{distributive law of $M$ over $H$} is a
  natural transformation $\lambda: MH \to HM$ 
  such that the diagram 
  \begin{equation} \label{diag:dl_pointed}
  \vcenter{
  \xymatrix@1{
    	& H \ar[dl]_{\eta H} \ar[dr]^{H\eta}	&	\\
    MH \ar[rr]^{\lambda} &	& HM
    }
    }
  \end{equation}
  commutes.

  \medskip\noindent
  (3)~Let $(D,\varepsilon)$ be a copointed endofunctor on $\A$. A \emph{distributive law of $(M,\eta)$ over $(D,\varepsilon)$} is a distributive law $\lambda:MD\to DM$ of $(M,\eta)$ over the functor $D$ that makes, in addition to \refeq{diag:dl_pointed} with $H$ replaced by $D$, the diagram
  $$
  \xymatrix@1{
    MD \ar[rr]^{\lambda} \ar[dr]_{M\varepsilon} &	& DM \ar[dl]^{\varepsilon M}	\\
    	& M	&
    }
  $$
  commute.
\end{defi}

\begin{rem}
\takeout{
  (1)~Recall from Remark~\ref{rem:ell_lambda}(2) that every abstract GSOS rule $\ell:K(H\times\Id)\to H\ext K$ gives a distributive law $\lambda$ of $\ext K$ considered now as a pointed functor $(\ext K,\eta)$ over the cofree copointed functor $H\times\Id$.

  \medskip\noindent}
  (1)~Every distributive law $\lambda:MH\to HM$ gives a distributive law of the cofree copointed functor $H\times\Id$ via
  \begin{equation} \label{diag:dl_lift_cofree_copointed}
    \vcenter{
    \xymatrix{
      MH \ar[r]^{\lambda}	& HM	\\
      M(H\times\Id) \ar@{-->}[r] \ar[u]^{M\pi_0} \ar[dr]_{M\pi_1}	& (H\times\Id)M \ar[u]_{\pi_0} \ar[d]^{\pi_1}	\\
      	& M
    }
    }
  \end{equation}
  but not conversely.

  \medskip\noindent
  (2)~Analogously to Theorem~\ref{thm:interp_gsos_rule}, we have for any distributive law $\lambda$ of $M$ over the cofree copointed functor $H\times \Id$ 
  a unique \emph{$\lambda$-interpretation}, i.\,e., a unique morphism $b:MC\to C$ such that the diagram below commutes
  $$
  \xymatrix@C+1pc{
    MC \ar[r]^-{M\langle c,\id_C\rangle} \ar[d]_{b}	& M(HC\times C) \ar[r]^-{\lambda_C}	& HMC\times MC \ar[r]^-{\pi_0}	& HMC \ar[d]^{Hb}	\\
    C \ar[rrr]^c	&	&	& HC
  }
  $$
  and $(C,b)$ is an algebra for the pointed functor $M$, see
  \cite{bartels_thesis}. Notice that $b$ here corresponds to $\extalg
  b:\ext KC\to C$ in Theorem~\ref{thm:interp_gsos_rule}. If we have a
  distributive law $\lambda:MH\to HM$, then we obtain one of $M$ over
  the copointed functor $H \times \Id$ as in
  \refeq{diag:dl_lift_cofree_copointed}. We again call the resulting
  morphism $b:MC\to C$ the $\lambda$-interpretation in $C$.  In this
  case, the diagram above simplifies to
  \begin{equation} \label{eq:interp:M:c}
  \vcenter{
  \xymatrix{
    MC \ar[r]^-{Mc} \ar[d]_{b}	& MHC \ar[r]^-{\lambda_C}	& HMC \ar[d]^{Hb}	\\
    C \ar[rr]^c	&	& HC \,\text{.}
  }
  }
  \end{equation}
\end{rem}

Next we shall need a version of Theorem~\ref{thm:ellsolution} for a given distributive law $\lambda$ of $M$ over $H$ (or over the cofree copointed functor $H\times\Id$). This is a variation of Theorem~4.2.2 of Bartels \cite{bartels_thesis} (see also Lemma~4.3.2 in loc.\,cit.) using the cocompleteness of $\A$. Since one part (the uniqueness part) of the proof in \cite{bartels_thesis} is only presented for $\Set$ we give a full proof here for the convenience of the reader.

\begin{thm} \label{thm:lambdasolution}
  Let $\lambda:MH\to HM$ be a distributive law of the pointed functor $M$ over the functor $H$. Then for every \emph{$\lambda$-equation} $e:X\to HMX$ there exists a unique solution, i.\,e., a unique morphism $e^\dag:X\to C$ such that the diagram below commutes:
  \begin{equation} \label{diag:sol_lambda-equation}
  \vcenter{
  \xymatrix{
    X \ar[rr]^e \ar[d]_{e^\dag}	&	& HMX \ar[d]^{HMe^\dag}	\\
    C \ar[r]^c	& HC	& HMC \ar[l]_{Hb}
  }
  }
  \end{equation}
\end{thm}

Before we proceed to the proof of the statement we need some auxiliary
constructions and lemmas. We begin by defining an endofunctor $S$ on our
cocomplete category $\A$ as a colimit. We denote by $M^n$, $n \in \Nat$, the
$n$-fold composition of $M$ with itself. Now we consider the diagram $D$ in
the category of endofunctors on $\A$ given by the natural transformations in
the picture below:
$$
\xymatrix@1@+1pc{
  \Id \ar[r]^-{\eta} 
  & 
  M 
  \ar@/^.5pc/[r]^-{\eta M}
  \ar@/_.5pc/[r]_-{M\eta}
  &
  MM
  \ar@/^1.2pc/[r]^-{\eta MM}
  \ar[r]^-{M\eta M}
  \ar@/_1.2pc/[r]_-{MM\eta}
  &
  \cdots
  }
$$
More formally, the diagram $D$ is formed by all natural transformations
$$
\xymatrix@1@+1pc{M^{i+j} \ar[r]^-{M^i\eta M^j} & M^{i+1+j}}
\qquad 
i, j \in \Nat \,\text{.}
$$
Let $S$ be a colimit of this diagram $D$:
$$
S = \colim D \qquad \textrm{with injections $\inj^i: M^i \to S$.}
$$
Then $S$ is a pointed endofunctor with the point $\inj^0: \Id = M^0 \to S$. 

Recall that colimits in the category of endofunctors of $\A$ are formed
objectwise. So for any object $X$, $SX$ is a colimit of the diagram $D$ at that
object $X$ with colimit injections $\inj^n_X: M^n X \to SX$, $n \in
\Nat$. This implies that for any endofunctor $F$ of $\A$ the functor $SF$ 
is a colimit with injections $\inj^n F: M^n F \to SF$. 

The above definition of $S$ appears in Bartels~\cite{bartels_thesis}. Next we define
additional data using the universal property of the colimits $SM$ and $SH$: 
\begin{enumerate}
\item a natural transformation $\chi: SM \to S$ uniquely determined by the
  commutativity of the triangles below:
  $$
  \vcenter{
  \xymatrix{
    M^{n+1}
    \ar[d]_{\inj^n M}
    \ar[dr]^{\inj^{n+1}}
    \\
    SM \ar[r]_-{\chi} & S
    }}
  \qquad
  \textrm{for all $n \in \Nat$.}
  $$
\item a natural transformation $\eps: SM \to MS$ uniquely determined by the
  commutativity of the triangles below:
  $$
  \vcenter{
  \xymatrix{
    M^{n+1} 
    \ar[d]_{\inj^n M}
    \ar[rd]^{M \inj^n}
    \\
    SM \ar[r]_{\eps} & MS
    }}
  \qquad
  \textrm{for all $n \in \Nat$.}
  $$
\item a natural transformation $\lambda^*: SH \to HS$; indeed, define
  first $\lambda^n: M^n H \to HM^n$ recursively as follows:
  \begin{eqnarray*}    
    \lambda^0 & = & \id_H: H \to H; \\
    \lambda^{n+1} & = &
    \xymatrix@1{
      M^{n+1} H = MM^n H
      \ar[r]^-{M \lambda^n} 
      &
      MHM^n \ar[r]^-{\lambda M^n} 
      & HMM^n = HM^{n+1} \,\text{.}
      }
  \end{eqnarray*}
  Then $\lambda^*$ is uniquely determined by the commutativity of the squares
  below: 
  $$
  \vcenter{
  \xymatrix{
    M^n H
    \ar[d]_{\inj^n H}
    \ar[r]^-{\lambda^n}
    &
    HM^n
    \ar[d]^{H \inj^n}
    \\
    SH \ar[r]_-{\lambda^*} 
    &
    HS
    }}
  \qquad\textrm{for all $n \in \Nat$.}
  $$
  Observe that $\lambda^*$ is a distributive law of the pointed endofunctor $S$
  over $H$; the unit law is the above square for the case $n = 0$. 
\end{enumerate}

We now need to verify that the three natural transformations above are
well-defined. More precisely, we need to prove that those natural
transformations are induced by appropriate cocones. For
$\chi: SM \to S$ and $\lambda^*: SH \to HS$, this follows from Lemma~4.3.2 in
Bartels' thesis~\cite{bartels_thesis}. Hence, we make the explicit verification only
for $\eps$ and leave the details for the other two natural transformations for
the reader. To verify that the natural transformations $M\inj^n:
M^{n+1} \to MS$ form a cocone for the appropriate diagram with colimit
$SM$ consider the triangles below:
$$
\vcenter{
\xymatrix{
  M^{1+i+j}
  \ar[rd]_{M \inj^n}
  \ar[rr]^-{MM^i\eta M^j}
  & & 
  M^{1 + i + 1 + j}
  \ar[ld]^{M \inj^{n+1}}
  \\
  &
  MS
  }}
\qquad\textrm{for all $n \in \Nat$, $n = i +j$.}
$$
These triangles commute since $\inj^n: M^n \to S$ form a cocone. 

Next, notice that in the definition of $\lambda^*$ above there are two possible
canonical choices for $\lambda^{n+1}$. We now show that these two choices are
equal: 

\begin{lem}\label{lem:lamsucc}
  For all natural numbers $n$ we have the commutative square below:
  $$
  \xymatrix@C+1pc{
    M^{n+1} H 
    \ar[r]^-{M^n \lambda} 
    \ar[d]_{M \lambda^n} 
    &
    M^n H M 
    \ar[d]^{\lambda^n M}
    \\
    MHM^n 
    \ar[r]_-{\lambda M^n} 
    &
    HM^{n+1} \,\text{.}
    }
  $$
\end{lem}
\begin{proof}
  We prove the result by induction on $n$. The base case $n=0$ is clear: both 
  composites in the desired square are simply $\lambda: MH \to HM$. For the
  induction step we need to verify that the diagram below commutes:
  $$
  \xymatrix@C+3pc{
    M^{n+1}M H
    \ar[d]_{MM \lambda^n}
    \ar[r]^-{M^{n+1} \lambda = MM^n \lambda}
    &
    M^{n+1} H M
    \ar[d]^{M\lambda^n M}
    \\
    MMHM^n
    \ar[r]^-{M\lambda M^n}
    \ar[d]_{M\lambda M^n}
    &
    MHM^nM
    \ar[d]^{\lambda M^n M}
    \\
    MH M^{n+1}
    \ar[r]_{\lambda M^{n+1}}
    \ar@{<-} `l[u] `[uu]^{M \lambda^{n+1}} [uu]
    &
    HM^{n+1}M
    \ar@{<-} `r[u] `[uu]_{\lambda^{n+1} M} [uu]
    }
  $$
  The left-hand and right-hand parts both commute due to the definition of
  $\lambda^{n+1}$. The lower square obviously commutes, and for the
  commutativity of the upper one apply the functor $M$ to the induction
  hypothesis. Thus the desired outside square commutes. 
\end{proof}

Next we need to establish a couple of properties connecting the three natural
transformations $\chi$, $\eps$ and $\lambda^*$.

\begin{lem}\label{lem:chieps}
  The following diagram of natural transformations commutes:
  $$
  \xymatrix{
    SMM
    \ar[r]^-{\chi M}
    \ar[d]_{\eps M}
    &
    SM
    \ar[d]^\eps
    \\
    MSM
    \ar[r]_-{M \chi}
    &
    MS \,\text{.}
    }
  $$
\end{lem}
\begin{proof}
  To verify that the square in the statement commutes we extend that square by
 precomposing with the injections into the colimit $SMM$. This yields the following diagram:
  $$
  \xymatrix{
    M^nMM
    \ar@{=}[rrr]
    \ar@{=}[ddd]
    \ar[rd]^{\inj^n MM}
    &
    &
    &
    M^{n+1} M
    \ar[ld]^{\inj^{n+1}M}
    \ar@{=}[ddd]
    \\
    &
    SMM 
    \ar[r]^-{\chi M}
    \ar[d]_{\eps M}
    &
    SM
    \ar[d]^\eps
    \\
    &
    MSM 
    \ar[r]_-{M \chi}
    & 
    MS
    \\
    MM^n M
    \ar[ru]^{M\inj^n M}
    \ar@{=}[rrr]
    &
    &
    &
    MM^{n+1}
    \ar[lu]_{M \inj^{n+1}}
    }
  $$
  The left-hand and right-hand inner squares commute by the definition of
  $\eps$,   and the upper and lower inner square commute by the definition of
  $\chi$. Since the outside commutes obviously, so does the desired middle
  square when precomposed by any injection $\inj^n MM$ of the colimit $SMM$. Thus,
  the desired middle square commutes. 
\end{proof}

\begin{lem}\label{lem:lamchi}
  The following square of natural transformations commutes:
  $$
  \xymatrix{
    SMH 
    \ar[r]^-{S\lambda}
    \ar[d]_{\chi H}
    &
    SHM
    \ar[r]^-{\lambda^* M}
    &
    HSM
    \ar[d]^{H\chi}
    \\
    SH
    \ar[rr]_-{\lambda^*}
    &&
    HS \,\text{.}
    }
  $$
\end{lem}
\begin{proof}
  It suffices to verify that the desired square commutes when we extend it by
  precomposition with
  an arbitrary  colimit injection $\inj^n MH$ of 
 $SMH$. To this end we consider the diagram below:
  $$
  \xymatrix{
    M^n MH
    \ar[rd]^{\inj^n MH}
    \ar[rr]^-{M^n \lambda}
    \ar@{=}[ddd]
    &
    &
    M^nHM
    \ar[d]^{\inj^n HM}
    \ar[rr]^-{\lambda^n M}
    &
    &
    HM^nM
    \ar[ld]_{H\inj^n M}
    \ar@{=}[ddd]
    \\
    &
    SMH
    \ar[r]^-{S\lambda}
    \ar[d]_{\chi H}
    &
    SHM
    \ar[r]^-{\lambda^* M}
    &
    HSM
    \ar[d]^{H \chi}
    \\
    &
    SH
    \ar[rr]_-{\lambda^*}
    &&
    HS
    \\
    M^{n+1}H
    \ar[ru]^{\inj^{n+1} H}
    \ar[rrrr]_-{\lambda^{n+1}}
    &&&&
    HM^{n+1}
    \ar[lu]_{H \inj^{n+1}}
    }
  $$
  The left-hand and right-hand parts commute by the definition of $\chi$, and
  the lower and the upper right-hand parts commute by the definition of
  $\lambda^*$. The upper left-hand part commutes by the naturality of
  $\inj^n$. Finally, the outside commutes by the definition of
  $\lambda^{n+1}$ together with Lemma~\ref{lem:lamsucc}. Thus, the desired middle square commutes when extended by
  any colimit injection $\inj^n MH$ of the colimit $SMH$.  
\end{proof}

\begin{lem}\label{lem:lameps}
  The following diagram of natural transformations commutes:
  $$
  \xymatrix{
    SMH 
    \ar[r]^-{S\lambda}
    \ar[d]_{\eps H}
    & 
    SHM
    \ar[r]^-{\lambda^* M}
    &
    HSM
    \ar[d]^{H \eps}
    \\
    MSH
    \ar[r]_-{M \lambda^*} 
    &
    MHS
    \ar[r]_-{\lambda S}
    &
    HMS
    }
  $$
\end{lem}
\begin{proof}
  Once more it is sufficient to verify that the desired square commutes when
  extended by any injection of the colimit $SMH$. So consider the diagram
  below:
  $$
  \xymatrix{
    M^n M H
    \ar[rd]^{\inj^n MH}
    \ar@{=}[ddd]
    \ar[rr]^-{M^n \lambda}
    &&
    M^nHM
    \ar[d]^{\inj^n HM}
    \ar[rr]^-{\lambda^n M}
    &&
    HM^n M
    \ar[ld]_{H\inj^n M}
    \ar@{=}[ddd]
    \\
    &
    SMH 
    \ar[r]^-{S\lambda}
    \ar[d]_{\eps H}
    &
    SHM
    \ar[r]^-{\lambda^* M}
    &
    HSM
    \ar[d]^{H \eps}
    \\
    &
    MSH
    \ar[r]_-{M\lambda^*} 
    &
    MHS
    \ar[r]_-{\lambda S}
    &
    HMS
    \\
    MM^nH
    \ar[ru]^{M\inj^n H}
    \ar[rr]_-{M \lambda^n}
    &&
    MHM^n
    \ar[u]^{MH\inj^n}
    \ar[rr]_-{\lambda M^n}
    &&
    HMM^n
    \ar[lu]_{HM\inj^n}
    }
  $$
  The left-hand and right-hand parts commute by the definition of
  $\eps$, and the lower left-hand and upper right-hand parts commute
  by the definition of $\lambda^*$. The upper left-hand and the lower
  right-hand parts both commute due to the naturality of $\inj^n$ and
  $\lambda$, respectively. Finally, the outside commutes by
  Lemma~\ref{lem:lamsucc}. Thus, the desired inner square commutes
  when extended by any colimit injection $\inj^n MH: M^n MH \to SMH$.
\end{proof}

We are now prepared to prove the statement of
Theorem~\ref{thm:lambdasolution}.

\begin{proof}[Proof of Theorem~\ref{thm:lambdasolution}]
Let $e: X \to HMX$ be any $\lambda$-equation. We form the following $H$-coalgebra:
\begin{equation} \label{eq:lambda-equation_H-alg}
\ol{e} \equiv
\xymatrix@1{
  SX \ar[r]^-{Se}
  &
  SHMX \ar[r]^-{\lambda^*_{MX}}
  &
  HSMX
  \ar[r]^-{H\chi_X}
  &
  HSX \,\text{.}
  }
\end{equation}
Since $c: C \to HC$ is a terminal $H$-coalgebra there exists a unique
$H$-coalgebra homomorphism $h$ from $(SX, \ol{e})$ to $(C, c)$. We shall prove
that the morphism
\begin{equation}
\label{eq:sollambdaeq}
\sol{e} \equiv
\xymatrix@1{
  X \ar[r]^-{\inj^0_X} & SX \ar[r]^-{h} & C
  }
\end{equation}
is the desired unique solution of the $\lambda$-equation $e$. 

(1)~$\sol{e}$ is a solution of $e$. It is our task to establish that the outside
of the diagram below commutes (cf.~Diagram~\ref{diag:sol_lambda-equation}):
$$
\xymatrix@C+1pc{
  X 
  \ar[dddd]_e
  \ar[r]^-{\inj^0_X}
  &
  SX
  \ar[d]_{S e}
  \ar[r]^-h
  &
  C 
  \ar[ddd]^c 
  \ar@{<-} `u[l] `[ll]_{\sol{e}} [ll]
  \\
  & 
  SHMX
  \ar[d]^{\lambda^*_{MX}}
  \\
  &
  HSMX
  \ar[d]^{H \chi_X}
  \\
  &
  HSX
  \ar[r]_-{Hh}
  &
  HC
  \\
  HMX
  \ar[ru]_-{H\inj^1_X}
  \ar@/^.5pc/[ruuu]^(0.65){\inj^0_{HMX}}
  \ar[ruu]|*+{\labelstyle H\inj^0_{MX}}
  \ar[rr]_-{HM\sol{e}}
  &&
  HMC \ar[u]_{Hb}
  }
$$
The upper part commutes by the definition of $\sol{e}$, and the upper
right-hand square commutes since $h$ is a coalgebra homomorphism. The upper
left-hand part commutes due to the naturality of $\inj^0$, the triangle below
that commutes by the definition of $\lambda^*$, and the lowest triangle
commutes by the definition of $\chi$. It remains to verify that the lowest part
commutes. To this end we will now establish the following equation
\begin{equation}\label{eq:todo}
  b \cdot M \sol{e} = h \cdot \inj^1_X \,\text{.}
\end{equation}
Consider the diagram below:
$$
\xymatrix{
  MX
  \ar[rr]^-{M\sol{e}}
  \ar[rd]^{M \inj^0_X}
  \ar[rdd]_(0.65){\inj^0_{MX}}
  \ar[ddd]_{\inj^1_X}
  &&
  MC 
  \ar[ddd]^b
  \\
  &
  MSX 
  \ar[ru]^-{Mh}
  \\
  &
  SMX 
  \ar[u]_{\eps_X}
  \ar[ld]^{\chi_X}
  \\
  SX
  \ar[rr]_h
  &&
  C
  }
$$

The upper triangle commutes by the definition of $\sol{e}$, the left-hand
triangle commutes by the definition of $\chi$ and the inner triangle commutes
by the definition of $\eps$. In order to establish that the right-hand part
commutes we will use that $C$ is a terminal $H$-coalgebra. Thus, we 
shall exhibit $H$-coalgebra structures on the five objects and then  
 show that
all edges of the right-hand part of the diagram are $H$-coalgebra homomorphisms. 
Then by the uniqueness of coalgebra homomorphisms into the terminal
coalgebra $(C,c)$, we conclude that the desired part of the above diagram commutes. 

For $C$, we use $c: C\to HC$, and for $MC$ we use $\lambda_C\cdot Mc$.
We already know that $b: MC \to C$ is a coalgebra homomorphism
(see~\refeq{eq:interp:M:c}).
For $SX$, we use $\ol{e}$ from \refeq{eq:lambda-equation_H-alg}; again, 
we already know that $h: SX \to C$ is a coalgebra homomorphism. 
For $MSX$ we use $\lambda_{SX}\cdot M\ol{e}$.  The verification that 
$Mh$ is a coalgebra morphism comes from 
 the diagram below:
$$
\xymatrix@C+1pc{
  MSX 
  \ar[r]^-{Mh}
  \ar[d]_{M\ol{e}}
  & 
  MC
  \ar[d]^{Mc}
  \\
  MHSX
  \ar[r]^-{MHh}
  \ar[d]_{\lambda_{SX}}
  &
  MHC
  \ar[d]^{\lambda_C}
  \\
  HMSX
  \ar[r]_-{HMh}
  &
  HMC
  }
$$
To see that the upper square commutes, remove $M$ and recall that $h$ is a
coalgebra homomorphism from $(SX, \ol{e})$ to $(C,c)$. The lower square
commutes by the naturality of $\lambda$. 

Now we show that $\eps_X: SMX \to MSX$ is a coalgebra homomorphism, where
the structure on $SMX$ is the composite on the left below:
$$
\xymatrix@C+1pc{
  SMX 
  \ar[rr]^-{\eps_X}
  \ar[d]_{SMe}
  &&
  MSX
  \ar[d]^{MSe}
  \ar `r[rd] `[ddd]^{M\ol{e}} [ddd]
  \\
  SMHMX
  \ar[rr]^-{\eps_{HMX}}
  \ar[d]_{S\lambda_{MX}}
  &&
  MSHMX
  \ar[d]^{M\lambda^*_{MX}}
  &
  \\
  SHMMX
  \ar[d]_{\lambda^*_{MMX}}
  &&
  MHSMX
  \ar[d]^{MH\chi_X}
  \ar[ld]_{\lambda_{SMX}}
  &
  \\
  HSMMX
  \ar[r]^-{H \eps_{MX}}
  \ar[d]_{H\chi_{MX}}
  &
  HMSMX
  \ar[rd]_{HM\chi_X}
  &
  MHSX
  \ar[d]^{\lambda_{SX}}
  &
  \\
  HSMX
  \ar[rr]_-{H\eps_X}
  &&
  HMSX
  }
$$
The upper square commutes by the naturality of $\eps$, and the inner triangle
commutes by the naturality of $\lambda$. To see that the right-hand part
commutes, remove $M$ and consider the definition of $\ol{e}$. The lowest part
commutes due to Lemma~\ref{lem:chieps}, and the middle part commutes by
Lemma~\ref{lem:lameps}. 

Finally, we show that $\chi_X: SMX \to SX$ is a coalgebra homomorphism.
To do this we consider the following diagram:
$$
\xymatrix@C+1pc{
  SMX
  \ar[r]^-{\chi_X}
  \ar[d]_{SMe}
  &
  SX
  \ar[d]^{Se}
  \\
  SMHMX
  \ar[r]^-{\chi_{HMX}}
  \ar[d]_{S\lambda_{MX}}
  &
  SHMX
  \ar[dd]^{\lambda^*_{MX}}
  \\
  SHMMX
  \ar[d]_{\lambda^*_{MMX}}
  \\
  HSMMX
  \ar[r]^-{H\chi_{MX}}
  \ar[d]_{H\chi_{MX}}
  &
  HSMX 
  \ar[d]^{H\chi_X}
  \\
  HSMX
  \ar[r]_-{H\chi_X}
  &
  HSX
  }
$$
The upper square commutes by the naturality of $\chi$, the middle square
commutes by Lemma~\ref{lem:lamchi}, and the lower square commutes obviously. 
This concludes the proof that $\sol{e}$ is a solution of $e$. 

(2)~$\sol{e}$ in \refeq{eq:sollambdaeq} is the unique solution of $e$. Suppose now that $\sol{e}$ is any
   solution of the $\lambda$-equation $e$. Recall that the object $SX$ is a
   colimit of the diagram $D$ at object $X$ with the colimit injections
   $\inj^n_X: M^n X \to SX$, $n \in \Nat$. We will use the universal property
   of that colimit to define a morphism $h: SX \to C$. To this end we need to
   give a cocone $h_n: M^n X \to C$, $n \in \Nat$, for 
   the appropriate diagram. We define this cocone inductively as follows:
   \begin{eqnarray*}
     h_0 & = & \sol{e}: M^0 X = X \to C;\\
     h_{n+1} & = & 
     \xymatrix@1{
       M^{n+1} X = MM^n X \ar[r]^-{Mh_n} & MC \ar[r]^-b & C
       },
     \qquad n \in \Nat.
   \end{eqnarray*}
   We now verify by induction on $n$ that the morphisms $h_n$, $n \in \Nat$ do
   indeed form a cocone. For the base case consider the diagram below:
   $$
   \xymatrix{
     M^0 X = X 
     \ar[rr]^-{\eta_X} 
     \ar[rd]_{h_0}
     \ar `d[ddr]_{h_0} [ddrr]
     && 
     MX= M^1 X
     \ar[d]^{M h_0}
     \\
     & 
     C
     \ar[r]^-{\eta_C}
     \ar@{=}[rd]
     &
     MC
     \ar[d]^b
     \\
     &&C
     }
   $$
   The upper part commutes by the naturality of $\eta$, the lower triangle
   commutes since $b: MC \to C$ is an algebra for the pointed endofunctor $M$,
   and the left-hand part is trivial. 
   For the induction step consider for any natural number $n = i+j$ the
   following diagram: 
   $$
   \xymatrix{
     M^{n+1} X = MM^{i+j} X
     \ar[rr]^-{MM^i\eta_{M^j X}}
     \ar[rd]_{Mh_n}
     \ar `d[ddr]_{h_{n+1}} [ddr]
     &&
     MM^iMM^j = M^{n+2} X
     \ar[ld]^{M h_{n+1}}
     \ar `d[ddl]^{h_{n+2}} [ddl]
     \\
     &
     MC
     \ar[d]^b
     \\
     &
     C
     }
   $$
   This diagram commutes: for the upper triangle remove $M$ and use the
   induction hypothesis, and the remaining two inner parts commute by the
   definition of $h_{n+1}$ and $h_{n+2}$, respectively. 
     
   Now we obtain a unique morphism $h: SX \to C$ such that for any natural
   number $n$ the triangle below commutes:
   \begin{equation}\label{eq:hdef}
     \vcenter{
       \xymatrix{
         M^n X
         \ar[d]_{\inj^n X}
         \ar[rd]^{h_n}
         \\
         SX
         \ar[r]_-h
         &
         C \,\text{.}
         }
       }
   \end{equation}
   Next we show that $h: SX \to C$ is a coalgebra
   homomorphism from $(SX, \ol{e})$ to the terminal coalgebra $(C,c)$. 
   To this end we will now verify that the lower part in the diagram below
   commutes:
   $$
   \xymatrix{
     M^n X
     \ar[d]_{\inj^n_X}
     \ar[r]^-{M^n e}
     \ar `l[d] `[dd]_{h_n} [dd]
     &
     M^n HMX
     \ar[d]^{\inj^n_{HMX}}
     \ar[r]^-{\lambda^n_{MX}}
     &
     HM^nMX
     \ar[d]^{H\inj^n_{MX}}
     \ar@{=}[r]
     &
     HM^{n+1}X
     \ar[d]^{H\inj^{n+1}_X}
     \ar `r[d] `[dd]^{Hh_{n+1}} [dd]
     \\
     SX
     \ar[r]^-{Se}
     \ar[d]_h
     &
     SHMX
     \ar[r]^-{\lambda^*_{MX}}
     &
     HSMX
     \ar[r]^-{H\chi_X}
     &
     HSX
     \ar[d]^{Hh}
     \\
     C
     \ar[rrr]^-c
     &&&
     HC
     }
   $$
   It suffices to show that the desired lower part commutes when extended by 
   any co\-limit injection $\inj^n_X$. Indeed, the left-hand part of the 
   diagram above commutes by Diagram~\refeq{eq:hdef}, and for the commutativity of
   the right-hand part, remove $H$ and use Diagram~\refeq{eq:hdef} again. The
   upper left-hand square commutes by the naturality of $\inj^n$, the upper
   middle square commutes by the definition of $\lambda^*$, and for the
   commutativity of the upper right-hand part remove $H$ and use the definition
   of $\chi$. It remains to verify that the outside of the diagram commutes. 
   We will now prove this by induction on $n$. For the base case $n=0$ we
   obtain the following diagram
   $$
   \xymatrix@C+1pc{
     X 
     \ar[dd]_{h_0}
     \ar[r]^-e
     &
     HMX
     \ar[d]_{HMh_0}
     \\
     &
     HMC
     \ar[d]_{Hb}
     \\
     C 
     \ar[r]_-c
     &
     HC
     \ar@{<-} `r[u] `[uu]_{Hh_1} [uu]
     }
   $$
   This diagram commutes: for the commutativity of the right-hand part remove
   $H$ and use the definition of $h_1$, and the left-hand part commutes since
   $h_0 = \sol{e}$ is a solution of the $\lambda$-equation $e$. 

   Finally, for the induction step we consider the diagram below:
   $$
   \xymatrix@C+1pc{
     M^{n+1} X 
     \ar[d]_{Mh_n}
     \ar[r]^-{M^{n+1} e}
     \ar `l[d] `[dd]_{h_{n+1}} [dd]
     &
     M^{n+1} HMX
     \ar[r]^-{M\lambda^n_{MX}}
     \ar `u[r] `[rr]^{\lambda^{n+1}_{MX}} [rr]
     &
     MHM^nMX
     \ar[d]^{MHh_{n+1}}
     \ar[r]^-{\lambda_{M^{n+1} X}}
     &
     HM^{n+2} X 
     \ar[d]_{HMh_{n+1}}
     \ar `r[d] `[dd]^{Hh_{n+2}} [dd]
     \\
     MC
     \ar[rr]^-{Mc}
     \ar[d]_b
     &&
     MHC
     \ar[r]^-{\lambda_C}
     &
     HMC 
     \ar[d]^{Hb}
     \\
     C
     \ar[rrr]_-c
     &&&
     HC
     }
   $$
   We see that this diagram commutes as follows: the lower part commutes by the
   definition of $b$ (see~\refeq{eq:interp:M:c}), the left-hand part
   commutes by the definition of $h_{n+1}$, and for the commutativity of the right-hand
   part remove $H$ and use the definition of $h_{n+2}$. The small upper part commutes by the definition of $\lambda^{n+1}$, the upper right-hand
   square commutes by the naturality of $\lambda$, and finally, to see the
   commutativity of the upper left-hand square remove $M$ and use the induction
   hypothesis. 
   
   We have finished the proof that $h: SX \to C$ is a coalgebra
   homomorphism from $(SX, \ol{e})$ to the terminal coalgebra $(C,c)$. Since $h$
   is uniquely determined,   it follows that the solution $\sol{e} = h \cdot
   \inj^0_X$ is uniquely determined, too. This completes our proof.
\end{proof}

\begin{rem}
  As explained by Bartels in \cite{bartels_thesis}, Theorem~\ref{thm:lambdasolution} extends to the case where a distributive law $\lambda$ of $M$ over the cofree copointed functor $H\times\Id$ is given. We briefly explain the ideas.

  Let $D=H\times \Id$ and $\varepsilon=\pi_1:D\to\Id$. \\
  (1)~A coalgebra for the copointed functor $(D,\varepsilon)$ is a pair $(X,x)$ where $x:X\to DX$ is such that $\varepsilon_X \cdot x = \id_X$. Homomorphisms of coalgebras for $(D,\varepsilon)$ are the usual $D$-coalgebra homomorphisms. It is trivial to prove that
  $$
  \xymatrix@C+1pc{
  C
  \ar[r]^-{\langle c,\id_C\rangle}
  &
  HC\times C
  }
  $$
  is a terminal coalgebra for $(D,\varepsilon)$.

  \medskip\noindent
  (2)~One verifies that $\lambda$-equations $e:X\to HMX$ are in bijective correspondence with morphisms $f:X\to DMX$ such that
  $$
  \xymatrix{
    X \ar[r]^-{f} \ar[dr]_{\eta_X}	& DMX \ar[d]^{\varepsilon_{MX}}	\\
    	& MX
  }
  $$
  commutes, and also
  that solutions of $e$ correspond bijectively to solutions of $f$, i.\,e., morphisms $f^\dag:X\to C$ such that Diagram \refeq{diag:sol_lambda-equation} commutes with $H$ replaced by $D$ and $c$ replaced by $\langle c,\id_C\rangle$:
  $$
  \xymatrix{
    X \ar[rr]^f \ar[d]_{f^\dag}	&	& DMX \ar[d]^{DMf^\dag}	\\
    C \ar[r]^-{\langle c,\id_C\rangle}	& DC	& DMC \ar[l]_{Db}
  }
  $$
  See \cite{bartels_thesis}, Lemma 4.3.9.

  \medskip\noindent
  (3)~The same proof as the one for Theorem~\ref{thm:lambdasolution} shows that for every $f:X\to DMX$ as in (2) above there exists a unique solution $f^\dag$. One only replaces $H$ by $D$, $c$ by $\langle c,\id_C\rangle$, and one has to verify that the coalgebra $\ol e:SX\to DSX$ from \refeq{eq:lambda-equation_H-alg} is a coalgebra for the copointed endofunctor, see \cite{bartels_thesis}, Lemma 4.3.7.
\end{rem}

To sum up, we obtain the following

\begin{cor}
  Let $\lambda$ be a distributive law of the pointed functor $M$ over the copointed one $H\times \Id$. Then for every $e:X\to HMX$ there exists a unique solution, i.\,e., a unique $e^\dag:X\to C$ such that \refeq{diag:sol_lambda-equation} commutes.
\end{cor}

\begin{thm} \label{thm:distcia_pointed}
  Let $\lambda$ be a distributive law of the pointed functor $M$ over the copointed one $H\times\Id$, and let $b:MC\to C$ be its $\lambda$-interpretation. Consider the algebra
  $$
  k = (
  \xymatrix{
    HMC
    \ar[r]^{Hb}
    &
    HC
    \ar[r]^{c^{-1}}
    &
    C
  }
  ) \,\text{.}
  $$
  Then $(C,k)$ is a cia for $HM$.
\end{thm}

Indeed, to prove this result copy the proof of Theorem~\ref{thm:distcia} replacing $\extalg b:\ext KC\to C$ by $b:MC\to C$.

However, for our version of Theorem~\ref{thm:sandwich} in the current setting we need a different proof. We start with an auxiliary lemma.

\begin{lem} \label{lem:pointed_MM}
Let $\lambda:MH\to HM$ be a distributive law of the pointed functor $M$ over the functor $H$, and let $b:MC\to C$ be its interpretation. Then the natural transformation $\lambda'=\lambda M \cdot M\lambda:MMH\to HMM$ is a distributive law of the pointed functor $MM$ over $H$, and $Mb \cdot b$ is the $\lambda'$-interpretation in $C$.
\end{lem}
\begin{proof}
Clearly $(MM,\eta M \cdot \eta = M\eta \cdot \eta:\Id\to MM)$ is a pointed endofunctor.
The following commutative diagram
$$
\begin{xy}
\xymatrix{
	&	&	H \ar[dl]_{\eta H} \ar[dr]^{H\eta}	&	&	\\
	&	MH \ar[dl]_{M\eta H} \ar[dr]_{MH\eta}	&	&	HM \ar[dl]^{\eta HM} \ar[dr]^{H\eta M}	&	\\
MMH \ar[rr]_{M\lambda}	&	&	MHM \ar[rr]_{\lambda M}	&	&	HMM
}
\end{xy}
$$
shows that $\lambda'=\lambda M \cdot M\lambda$ is a distributive law of
the pointed functor $MM$ over $H$. In fact, the triangles commute by
the assumption on $\lambda$, and the remaining upper square commutes
by naturality of $\eta$; thus the outside triangle commutes.

To see that $b \cdot Mb$ is the $\lambda'$-interpretation in $C$, consider the following diagram:
$$
\begin{xy}
\xymatrix{
MMC \ar[r]^-{MMc} \ar[d]_{Mb}	&	MMHC \ar[r]^{M\lambda_C}	&	MHMC \ar[r]^{\lambda_{MC}} \ar[d]_{MHb}	&	HMMC \ar[d]^{HMb}	\\
MC \ar[rr]^{Mc} \ar[d]_{b}	&	&	MHC \ar[r]^{\lambda_C}	&	HMC \ar[d]^{Hb}	\\
C \ar[rrr]_{c}	&	&	&	HC
}
\end{xy}
$$
It commutes since $b$ is the $\lambda$-interpretation in
$C$ and by the naturality of $\lambda$. In addition, $b \cdot Mb$ is easily seen to be an algebra for the pointed functor $(MM,\eta M \cdot \eta)$ since $b$ is one for $(M,\eta)$ and $\eta$ is a natural transformation:
$$
\begin{xy}
\xymatrix{
C \ar[d]_{\eta_C} \ar@{=}[dr]	&	&	\\
MC \ar[d]_{\eta_{MC}} \ar[r]_{b}	&	C \ar[d]_{\eta_C} \ar@{=}[dr]	&	\\
MMC \ar[r]_{Mb}	&	MC \ar[r]_b	&	C
}
\end{xy}
$$
This concludes the proof.
\end{proof}

\begin{thm} \label{sandwich_pointed}
  Let $\lambda:MH\to HM$ be a distributive law of the pointed functor $M$ over the functor $H$, and let $b:MC\to C$ be its $\lambda$-interpretation. Consider the algebra
  \iffull
  $$k' =(\xymatrix@1{MHMC \ar[r]^-{Mk} & MC \ar[r]^-{b} & C}) \,\text{,}$$
  \else
  $k' = MHMC \xrightarrow{Mk} MC \xrightarrow{b} C \,\text{,}$
  \fi
  where $k = c^{-1} \cdot Hb$ as in Theorem~\ref{thm:distcia_pointed}. 
  Then $(C, k')$ is a cia for $MHM$.
\end{thm}
\begin{proof}
We have to prove 
that for every flat equation morphism $e:X\to MHMX+C$ for $MHM$
there is a unique solution $e^{\dagger}:X\to C$ in $k' = b \cdot
Mc^{-1} \cdot MHb: MHMC \to C$, i.\,e., a unique morphism $\sol e$ such that the
outside of the diagram 
$$
\begin{xy}
\xymatrix{
X \ar[dddd]_{e} \ar[ddr]^{\bar e} \ar[rrr]^{e^{\dagger}}	&
&	&	C  
\ar@{<-} `r[rd] `[dddd]^{[k',C]} [dddd] & \\
	&	&	HC+C \ar[ur]^{[c^{-1},C]}	&	MC+C
        \ar[u]_{[b,C]} \ar@<-3pt>[d]_{Mc + C} &	\\
	& HMMX+C \ar[dr]_(0.4){HMMe^{\dagger}+C}	&	HMC+C \ar[u]^{Hb+C}	&	MHC+C
        \ar[l]_{\lambda_C+C} \ar@<-3pt>[u]_{Mc^{-1}+C}	& \\
	&	&	HMMC+C
        \ar[u]^{HMb+C}	& &	\\
MHMX+C \ar[uur]^{\lambda_{MX}+C} \ar[rrr]_{MHMe^{\dagger}+C}	&	&	&	MHMC+C \ar[ul]^{\lambda_{MC}+C} \ar[uu]_{MHb+C}
}
\end{xy}
$$
commutes. To this end, we define the flat equation morphism 
$$
\ol e =
(\xymatrix@1{
  X \ar[r]^-e & MHMX + C \ar[rr]^-{\lambda_{MX}+C} && HMMX+C})
$$ 
for $HMM$ (this is the left-hand triangle). According to Lemma \ref{lem:pointed_MM} and Theorem
\ref{thm:distcia_pointed}, 
$$
\xymatrix@1{
  HMMC \ar[r]^-{HMb}
  &
  HMC \ar[r]^-{Hb} 
  &
  HC
  \ar[r]^-{c^{-1}}
  &
  C
}
$$
is a cia for $HMM$. So $\ol e$ has a unique solution $\sol e$ in this
cia, i.\,e., the big inner part of the diagram commutes.
In the upper right-hand part, $b$ is the $\lambda$-interpretation
in $C$. Since that part
and the two remaining squares also commute (due to naturality
of $\lambda$), the desired outside commutes.
Thus, $\sol e$ also is a solution of $e$ in the algebra $k': MHMC \to C$.

This solution is unique, since any other solution $e^{\ddagger}$ of
$e$ in $k'$ (i.\,e., the outside of the above diagram with $e^{\ddagger}$ in lieu of
$e^{\dagger}$ commutes) is a solution of $\ol e$ in the cia
$c^{-1} \cdot Hb \cdot HMb: HMMC \to C$ (i.\,e., the inner part commutes with
$e^{\ddagger}$ in lieu of $e^{\dagger}$), thus
$e^{\ddagger}=e^{\dagger}$.
\end{proof}
\fi 

\end{document}
